\documentclass[a4paper, 12pt]{article}
\usepackage{amssymb,amsmath,latexsym,color,graphicx,stmaryrd}
\usepackage{hyperref}

\newtheorem{dref}{Definition}[section] \newtheorem{lemma}[dref]{Lemma}
\newtheorem{theo}[dref]{Theorem} \newtheorem{prop}[dref]{Proposition}
\newtheorem{remark}[dref]{Remark}

\newenvironment{proof}{\par\noindent{{\bf Proof.}}}{\hfill$\Box$
\medskip}
\newenvironment{proofof}{\par\noindent{{\bf Proof} of }}{\hfill$\Box$
\medskip}

\title{Adiabatic evolution and shape resonances
\author{Michael Hitrik \\\small Department of Mathematics\\
 \small  University of California \\\small Los Angeles, CA
  90095-1555, USA\\\footnotesize hitrik@math.ucla.edu \and
  Andrea Mantile\\
  \small Laboratoire de Math\'ematiques -
  FR3399 CNRS\\ \small Universit\'e de Reims
  \\ \small Moulin de la Housse - BP 1039\\
\small 51687 Reims, France\\
  \footnotesize andrea.mantile@univ-reims.fr \and Johannes
  Sj\"ostrand\\ \small IMB {\footnotesize - UMR5584 CNRS},
  Universit\'e de Bourgogne\\
  \small 9, avenue Alain Savary - BP 47870\\
  \small 21078 Dijon cedex, France\\ \footnotesize
  johannes.sjostrand@u-bourgogne.fr }} \date{}
\begin{document}
\maketitle
\begin{abstract}  Motivated by a problem of one mode approximation for
  a non-linear evolution with charge accumulation in potential wells,
  we consider a general linear adiabatic evolution problem for a semi-classical
  Schr\"odinger operator with a time dependent potential with a
  well in an island. In particular, we show that we can
  choose the adiabatic parameter $\varepsilon $ with $\ln\varepsilon
  \asymp -1/h$, where $h$ denotes the semi-classical parameter,
 and get adiabatic approximations of exact solutions over a
  time interval of length $\varepsilon ^{-N}$ with an error ${\cal
    O}(\varepsilon ^N)$. Here $N>0$ is arbitrary. \footnote{While
    deciding the general strategy through joint discussions, the
    coauthors have invested various amounts of time in the actual
    elaboration. The main authors
    of the different sections are in indicated by their initials as follows:\\
Section \ref{int}: MH, AM, JS,\\
Sections \ref{1eig}, \ref{fad}: AM, JS,\\
Sections \ref{opesc}, \ref{nore}, \ref{sbd}, \ref{far}: JS,\\
Section \ref{rest}: MH, JS,\\
Section \ref{bics}: MH, AM, JS.
}

  \medskip\centerline{\bf R\'esum\'e} Motiv\'es par un probl\`eme
  d'approximation \`a un mode pour une \'evolution avec accumulation de
  charge dans des puits de potentiel, nous consid\'erons un probl\`eme
  d'\'evolution lin\'eaire pour un op\'erateur de Schr\"odinger avec un
  potentiel d\'ependant du temps avec un puits dans une {\^\i}le. En
  particular, nous montrons que nous pouvons choisir le param\`etre
  adiabatique $\varepsilon $ avec $\ln\varepsilon \asymp -1/h$, o\`u
  $h$ d\'esigne param\`etre semi-classique, et obtenir des
  approximations adiabatiques de solutions exactes sur des intervalles
  de temps de longueur $\varepsilon ^{-N}$ avec une erreur
  ${\cal O}(\varepsilon ^N)$. Ici $N>0$ est arbitraire.
\end{abstract}
\medskip

\tableofcontents

\section{Introduction and main results}\label{int}
\setcounter{equation}{0}

Our work is connected with the modelling
of the axial transport through resonant tunneling structures like highly doped
p-n semiconductor heterojunctions (Esaki diodes), multiple barriers or quantum
wells diodes. The scattering of charge carriers in such devices has been
described using non-linear Schr\"{o}dinger-Poisson Hamiltonians with quantum
wells in a 1D semiclassical island (see \cite{JoPrSj95}). The quantum wells
regime is defined as a perturbation of the semiclassical Laplacian
$-h^{2}\partial_{x}^{2}$ by the superposition of a potential barrier plus an
attractive term, with support of size $h$, modelling one or more \emph{quantum
wells}. In the simplest setting of a single well separating two linear
barriers, the linear part of the potential has the shape in Figure \ref{fig1}.
%TCIMACRO{\FRAME{dtbpFU}{3.8821in}{2.1959in}{0pt}{\Qcb{Fig.1}}{\Qlb{Fig_1}%
%}{fig1.bmp}{\special{ language "Scientific Word";  type "GRAPHIC";
%width 3.8821in;  height 2.1959in;  depth 0pt;  original-width 15.3598in;
%original-height 8.6401in;  cropleft "0";  croptop "1";  cropright "1";
%cropbottom "0";  filename 'fig1.bmp';file-properties "XNPEU";}} }%
%BeginExpansion
%\begin{center}
%\includegraphics[
%natheight=2.195900in,
%natwidth=3.882100in,
%height=2.1959in,
%width=3.8821in
%]%
%{C:/Users/Andrea/Desktop/Johannes/graphics/fig1__1.pdf}%
%\\
%Fig.1
%\label{Fig_1}%
%\end{center}
%EndExpansion
\begin{figure}
   \caption{\label{fig1} The potential with no charge accumulation}\begin{center}
   \includegraphics[width=12cm]{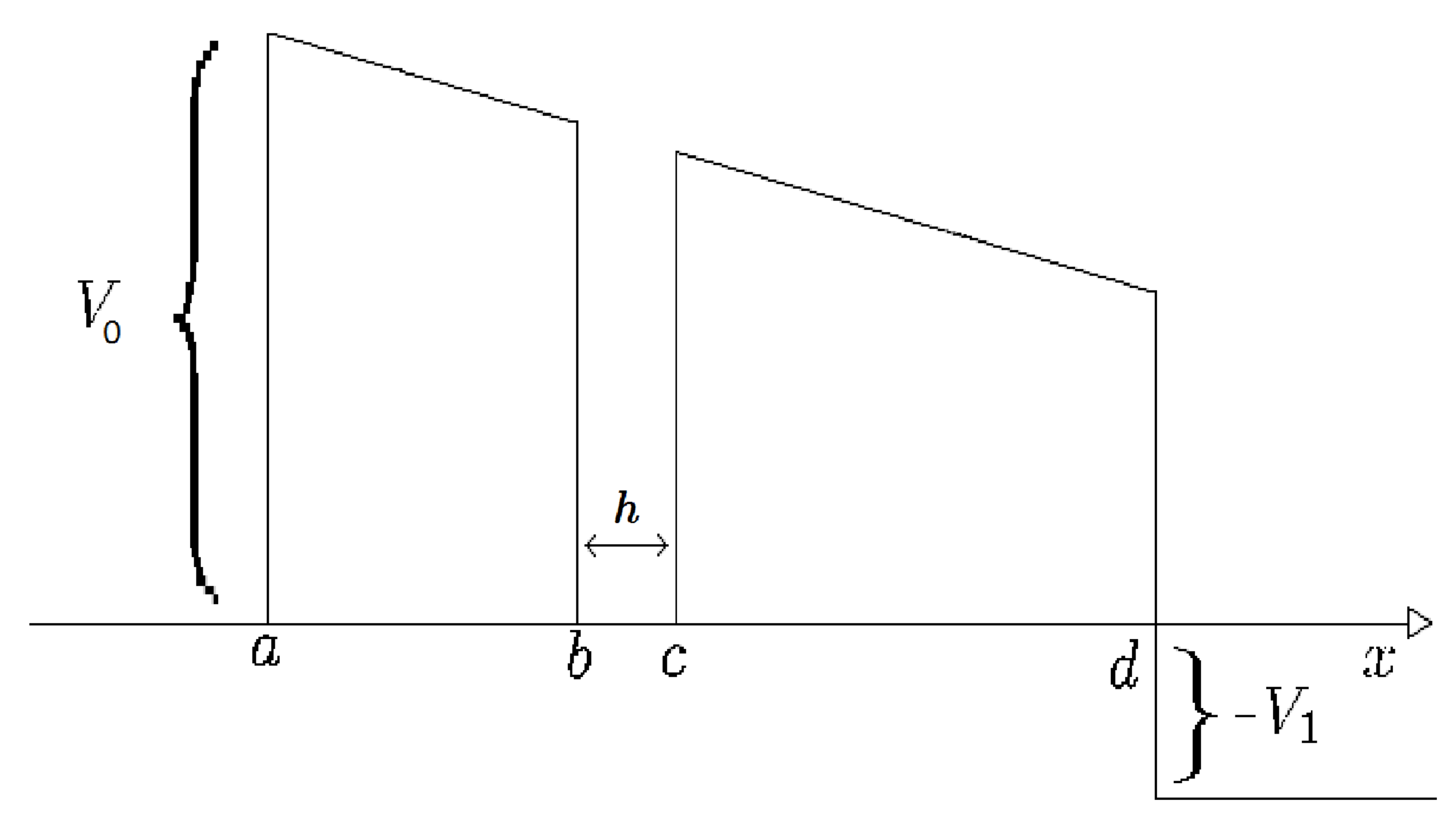}
\end{center}
\end{figure}
In connection with the modelling of a mesoscopic semi-conductor device, this
scheme represents \textquotedblleft metallic conductors\textquotedblright\ at
$]-\infty,a]$ and $[d,+\infty\lbrack$ while the double barrier describes the
interaction of charge carriers in a semiconductor junction. Here $a<b<d$ are
fixed, $c=b+h$ with $h\rightarrow0$, while $V_{1}$ defines an exterior voltage
applied between the two infinite conductors. In this framework, $h$
corresponds to a rescaled Fermi length fixing the quantum scale of the system
(see for instance \cite{BoNiPa06}) and, coherently with the features of the
physical model, is assumed to be small. The shape resonances (i.e. those with
energies below $V_{0}$) define the Fermi levels of the junction and the
corresponding resonant states describe (in the one-particle approximation) the
concentration of charges in the depletion region. In particular, the exterior
potential bias $-V_{1}$ is introduced in order to select only incoming waves
with positive momentum as contributions to the charging process (about this point
the reader may refer to the analysis developed in \cite{BoNiPa08}).

In linear models, the small-$h$ asymptotic behaviour of the shape resonances
generated by quantum wells has been understood in the work of
B. Helffer and J. Sj\"ostrand \cite{HeSj86}; for operators of the form%
\begin{equation}
H^{h}:=-h^{2}\partial_{x}^{2}+V^{h} \label{H_lin}%
\end{equation}
with $V^{h}$ fulfilling the scaling of Figure \ref{fig1} (i.e.: a semiclassical barrier
supported on $\left[  a,d\right]  $ plus quantum wells) and suitable
regularity assumptions, the approach of \cite{HeSj86} allows to localize the
shape resonances w.r.t. the spectrum of a corresponding Dirichlet operator%
\begin{equation}
H_{D}^{h}:=-h^{2}\triangle_{\left(  a,d\right)  }^{D}+1_{\left(  a,b\right)
}V^{h} ,
\end{equation}
where $\triangle_{\left(  a,d\right)  }^{D}$ denotes the Dirichlet Laplacian
on the barrier interval. In particular, very accurate Agmon-type estimates
show that to each $\lambda\in\sigma_{p}\left(  H_{D}^{h}\right)  \cap\left(
0,V_{0}\right)  $ corresponds a unique resonance of $H^{h}$, $E_{\mathrm{res}}
^{h}=E^{h}-i\Gamma^{h}$, and the estimates%
\begin{equation}
\left\vert \lambda-E^{h}\right\vert +\Gamma^{h}\lesssim e^{-S_{0}/h}\,,
\label{semiclassical_est1}%
\end{equation}
hold with $\Gamma^{h}>0$ fixing the imaginary part of $E_{\mathrm{res}}^{h}$ and
$S_{0}>0$ depending on the Lithner-Agmon distance separating the well from the
boundary of the barrier. In the time evolution problems, the imaginary part of
resonances fixes the lifetime of the corresponding \emph{quasiresonant
states}, which are $L^{2}$-functions defined by a cut-off of the resonant
states outside the interaction region (see the definition in \cite{GeSi92}).
This general idea has been investigated in \cite{Sk89} in the framework of
$3D$ Schr\"{o}dinger operators with exponentially decaying potentials (see
also \cite{SoWe98} and \cite{NaStZw03}); for operators exhibiting the scaling
introduced above, a precise exponential decay estimate has been provided in
\cite[Th. 4.3]{GeSi92}. Let $u_{E_{\mathrm{res}}^{h}}$ denote the resonant state of
$H^{h}$ for the resonance $E_{\mathrm{res}}^{h}$ (i.e.: a solution of $\left(
H^{h}-E_{\mathrm{res}}^{h}\right)  u_{E_{\mathrm{res}}^{h}}=0$); under the assumption
$\Gamma^{h}\gtrsim e^{-2S_{0}/h}$ (which holds for a large class of models
including the case of sharp barriers (see \cite{BoNiPa09})) we have%
\begin{equation}
e^{-it H^{h}}1_{\left(  a,d\right)  }u_{E_{\mathrm{res}}^{h}}=e^{-itE_{\mathrm{res}}^{h}
}1_{\left(  a,d\right)  }u_{E_{\mathrm{res}}^{h}}+R^{h}\left(  t\right)  \,,
\end{equation}
where $R^{h}\left(  t\right)  =\mathcal{O}\left(  e^{-S_{0}/h}\right)  $, in
the $L^{2}$-norm sense, on the time scale: $t\lesssim1/\Gamma^{h}$. Then the
estimate (\ref{semiclassical_est1}) implies%
\begin{equation}
\left\vert e^{-it H^{h}}1_{\left(  a,d\right)  }u_{E_{\mathrm{res}}^{h}}\right\vert
\approx e^{-t\Gamma^{h}}\left\vert 1_{\left(  a,d\right)  }u_{E_{\mathrm{res}}^{h}}\right\vert
\,. \label{Skib_est1}
\end{equation}
The comparison between the shape resonance problem for $H^{h}$ and the
eigenvalue equation for the corresponding truncated Dirichlet model $H_{D}%
^{h}$ also shows that the quasiresonant states are mainly supported near the
wells (see \cite{HeSj86}). Hence, according to the above relation, the time
evolution preserves this concentration of $L^{2}$-mass on the time scale
$1/\Gamma^{h}$ which is exponentially large w.r.t. $h$.

In the non-linear modelling, the repulsive effect due to the concentration of
charges in the depletion region is taken into account by a Poisson potential
term depending on the charge density. The corresponding non-linear steady
state problem%
\begin{equation}%
\begin{array}
[c]{ccc}%
\left(  H^{h}+V_{NL}^{h}-E\right)  u=0,\, &  & \partial_{x}^{2}V_{NL}%
^{h}=\left\vert u\right\vert ^{2}\,,
\end{array}
\end{equation}
has been investigated in \cite{BoNiPa08}-\cite{BoNiPa09} under
far-from-equilibrium assumptions; in this case, following the scaling
introduced above, the underlying linear model $H^{h}$ is defined by using an
array of quantum wells of the form%
\begin{equation}
W^{h}=-\sum_{n=1}^{N}w_{n}\left(  \left(  x-x_{n}\right)  1/h\right)  \,,\quad
w_{n}\in\mathcal{C}^{0}\left(  \mathbb{R},\mathbb{R}_{+}\right)
\,,\quad\text{supp }w_{n}=\left[  -d,d\right]  \,.
\end{equation}
An accurate microlocal analysis of the tunnel effect as $h\rightarrow0$ then
shows that the estimates (\ref{semiclassical_est1}) still hold in the
stationary nonlinear framework and determine the limit occupation number of
resonant states. This analysis leads to a simplified equation for the Poisson
problem where the limit charge density is described by a superposition of
delta-shaped distributions centered in the points $\left\{  x_{n}\right\}  $.
In the time-dependent case, the non-linear evolution equation reads as%
\begin{equation}%
\begin{array}
[c]{ccc}%
i\partial_{t}u=\left(  H^{h}+V_{NL}^{h}\right)  u\,, &  & \partial_{x}%
^{2}V_{NL}^{h}=\left\vert u\right\vert ^{2}\,.
\end{array}
\label{SP_eq}%
\end{equation}
When the initial state is formed by a superposition of incoming waves with
energies close to the Fermi level ($E_{F}$ the resonant energy), this
interacts with resonant states which, as the estimates (\ref{Skib_est1})
suggest, are expected to evolve in time according to a quasi-stationary
dynamics. In this picture, $u$ behave as a metastable state and the charge
density $\left\vert u\right\vert ^{2}$ remains concentrated in a neighbourhood
of the wells for a large range of time fixed by the imaginary part of
the (nonlinear) resonances. Depending on the position of the wells, this
possibly induces a local charging process; then, the nonlinear coupling in
(\ref{SP_eq}) generates a positive response (depending on the charge in the
wells) which modifies the potential profile and reduces the tunnelling rate.%

%TCIMACRO{\FRAME{dtbpFU}{3.8821in}{2.1959in}{0pt}{\Qcb{Fig.2}}{\Qlb{Fig_2}%
%}{fig2.bmp}{\special{ language "Scientific Word";  type "GRAPHIC";
%width 3.8821in;  height 2.1959in;  depth 0pt;  original-width 15.3598in;
%original-height 8.6401in;  cropleft "0";  croptop "1";  cropright "1";
%cropbottom "0";  filename 'fig2.bmp';file-properties "XNPEU";}} }%
%BeginExpansion
%\begin{center}
%\includegraphics[
%natheight=2.195900in,
%natwidth=3.882100in,
%height=2.1959in,
%width=3.8821in
%]%
%{C:/Users/Andrea/Desktop/Johannes/graphics/fig2__2.pdf}%
%\\
%Fig.2
%\label{Fig_2}%
%\end{center}
%EndExpansion
\begin{figure}
   \caption{\label{fig2} The potential with charge accumulation}\begin{center}
   \includegraphics[width=12cm]{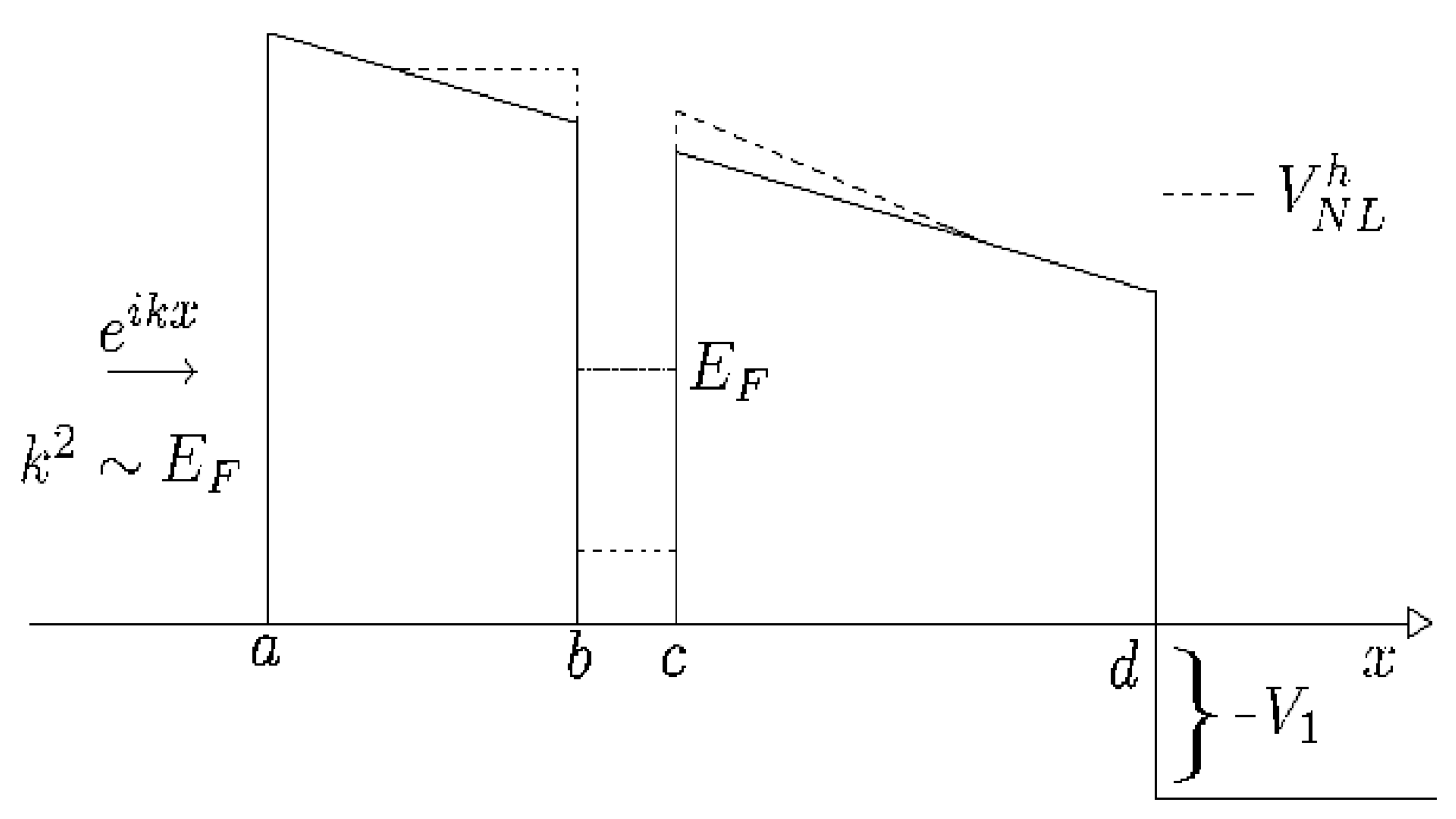}
\end{center}
\end{figure}

The above scheme outlines the behaviour of the nonlinear dynamics under
non-equilibrium initial conditions and assuming $h$ small. In particular,
$V_{NL}^{h}$ is expected to define an adiabatic process with variations in
time of size $\varepsilon=\Gamma^{h}$. The relevance of adiabatic
approximations in the small-$h$ asymptotic analysis of the nonlinear quantum
transport was pointed out in (\cite{JoPrSj95}, \cite{PrSj96}, \cite{PrSjJ97})
where this dynamics was considered within a simplified framework.

\subsection{The works of C. Presilla and J. Sj\"{o}strand}

Following the work \cite{JoPrSj95} (with G.\ Jona-Lasinio), in
\cite{PrSj96}, \cite{PrSjJ97} C.\ Presilla and J.\ Sj\"{o}strand
considered a non-linear evolution problem for a mesoscopic
semi-conductor device and did some heuristic work. It is assumed that
the incoming charged particles (entering from the left) have energies
$E\geq0$ distributed according to the density $g(E)dE$ supported on
$[0,E_{F}]$, where $E_{F}<V_{0}$ is the Fermi level. Moreover, these
particles interact only inside the device (i.e.\ in the region $]a,d[$
in Figure \ref{fig1}) through a modification of the common potential
due to charge accumulation there. After a rescaling, the model is
described by a nonlinear Hamiltonian%
\[
H_{NL}^{h}=H^{h}+s\left(  u^{h}(t,\cdot)\right)  W_{0}(x)\,,
\]
where the linear part $H^{h}$ is defined as in (\ref{H_lin}) and
Figure \ref{fig1}, while
the Poisson potential is replaced by an affine function $W_{0}(x)$, with a
fixed \textquotedblleft profile\textquotedblright\ and support in $]a,d[$,
multiplied by the charge accumulated inside the device. This is defined in
terms of the nonlinear evolution of generalized eigenfunctions $u\left(
t,x,E\right)  $ according to%
\begin{equation}
s(u^{h}(t,\cdot)):=\int%
%TCIMACRO{\dint _{(a+b)/2}^{(c+d)/2}}%
%BeginExpansion
{\displaystyle\int_{(a+b)/2}^{(c+d)/2}}
%EndExpansion
\left\vert u^{h}(t,x,E)\right\vert ^{2}g(E)\,dxdE\,,
\end{equation}
and the corresponding nonlinear evolution problem is{%
\begin{equation}
\left\{
\begin{array}
[c]{l}%
i\partial_{t}u^{h}\left(  t,x,E\right)  =\left(  H^{h}+s(u^{h}(t,\cdot
))W_{0}(x)\right)  u^{h}\left(  t,x,E\right)  \,,\\
\\
\left(  H^{h}+s(u^{h}(0,\cdot))W_{0}(x)-E\right)  u^{h}\left(  0,x,E\right)
=0\,.
\end{array}
\right.  \label{int.1}%
\end{equation}
}The heuristic analysis in \cite{PrSj96} was based on the 1-mode
approximation
\begin{equation}
u^{h}(t,x,E)\approx\mu(t,x,E)+e^{-iEt/h}z^{h}(t,E)e^{h}(x,s\left(
u^{h}(t,\cdot)\right)  )\,, \label{int.4}%
\end{equation}
where

\begin{itemize}
\item $\mu(t,x,E)$ is the solution of a linear evolution problem, obtained by
``filling'' the potential well $[b,c]$,

\item $e^{h}=e^{h}\left(  x,s\left(  u^{h}(t,\cdot)\right)  \right)  $ is a
resonant state ($\not \in L^{2}$) corresponding to a resonance $\lambda
^{h}(s\left(  u^{h}(t,\cdot)\right)  )$ in the lower half-plane.
\end{itemize}

From this the authors derived a simpler evolution equation%
\begin{equation}
\left.
\begin{array}
[c]{l}%
h\partial_{t}z^{h}(t,E)=i(E-\lambda^{h}(s))z^{h}(t,E)+\mathcal{B}%
^{h}(t,s,E)\,,\\
\\
s\left(  u^{h}(t,\cdot)\right)  \sim%
%TCIMACRO{\dint }%
%BeginExpansion
{\displaystyle\int}
%EndExpansion
\left\vert z^{h}\left(  t,E\right)  \right\vert ^{2}g\left(  E\right)  \,dE\,,
\end{array}
\right.  \label{int.5}%
\end{equation}
where $\mathcal{B}^{h}(t,x,E)$ is a \textquotedblleft driving
term\textquotedblright\ derived from $\mu$ and $e^{h}$. Then, using an
adiabatic approximation for the evolution of the nonlinear resonant state in
term of instantaneous resonances and WKB expansions, an even further
simplified differential equation for $s\left(  u^{h}(t,\cdot)\right)  $ was
obtained in the limit $h\rightarrow0$ (eq. 9.7 in \cite{PrSj96}). From this
one could describe fixed points of the vector field in (\ref{int.5}), and
hysteresis phenomena under slow variations of the exterior bias $V_{1}%
=V_{1}(t)$.

\subsection{Adiabatic evolution of resonant states}

A rigorous study of the model (\ref{SP_eq}) in the small $h$ regime is a very
vast program. In this connection, we remark that the lack of an error bound in
the adiabatic formulas used in \cite{PrSj96} prevents to control the possible
remainder terms in the asymptotic limit. Hence, {a strong adiabatic theorem
for resonant states, with adiabatic parameter $\varepsilon$ satisfying }%
\begin{equation}
{\ln\varepsilon\asymp-1/h\,,} \label{adiabatic _scaling}%
\end{equation}
{seems to be a key point of this program. The adiabatic problem for resonances
can be considered following different approaches. One consists in using the
unitary propagator associated to the physical (selfadjoint) Hamiltonian and
study the adiabatic evolution of }$L^{2}${ states spectrally localized near
the resonant energy (}or $L^{2}$ functions obtained by truncating resonant
states{). This point of view was adopted in }\cite{Pe00} where the condition
(\ref{adiabatic _scaling}), connecting the adiabatic parameter to the
resonance lifetime, was also taken into account; in the case of a single
time-dependent resonance $E_{res}\left(  t\right)  $ with: $\operatorname{Im}%
E_{res}\left(  t\right)  \sim\operatorname{Im}E_{res}\left(  0\right)
=e^{-\frac{c}{h}}$ an adiabatic formula was obtained on the specific time
range $\left(  \operatorname{Im}E_{res}\left(  0\right)  \right)  ^{-1}$.

A different approach consists in using a complex deformation to define
resonances as eigenvalues of a deformed (non-selfadjoint) Hamiltonian; in this
framework, the resonant states identify with $L^{2}$ eigenvectors and the
corresponding evolution problem is naturally formulated in terms of the deformed
operator. Then, an adiabatic approximation can be studied by adapting the
standard adiabatic theorem with gap condition to the non-selfadjoint case. (a
similar strategy was implemented in \cite{AbFr07}). It is worthwhile to
remark that this requires uniform-in-time bounds for the deformed dynamical
system (we refer to \cite{Ne93}): the lack of this condition, due to the
complex deformation, is the main obstruction to implement such an approach in
some relevant physical models (including those we are considering here). In
\cite{Jo07}, the time-adiabatic evolution in Banach spaces has been
considered, under a fixed gap condition, for semigroups exhibiting an
exponential growth in time. In this framework, which could be adapted to the
case of resonant states, the exponential growth of the dynamical system is
compensated by the small error of the adiabatic approximation under
analyticity-in-time assumptions and an adiabatic formula for the evolution of
spectral projectors is provided with an error which is small on a suitable
range of time.

More recently, an alternative approach to the adiabatic evolution of shape
resonances has been proposed \cite{FaMaNi11}. For a 1D Schr\"{o}dinger operator
describing the regime of quantum well in a semiclassical island, artificial
(non-selfadjoint) interface conditions are added at the boundary of the
potential's support. Depending on the deformation parameter $\theta$, these
may be chosen in such a way that the corresponding modified and complex
deformed operator is maximal accretive: hence, the deformed dynamical system
allows uniform-in-time estimates and the adiabatic theory for isolated
spectral sets can be developed in the deformed setting. In particular, the
authors show that artificial interface conditions introduce small
perturbations on shape resonances, preserving the relevant physical quantities
(the exponentially small scales). Using an exterior complex deformation
polynomially small w.r.t. the quantum scale $h$, an adiabatic theorem for the
resonant states associated to shape resonances is provided. In this framework,
the adiabatic parameter $\varepsilon=e^{-c/h}$ can be fixed with any $c>0$
independently of the resonance lifetime (see \cite[Th. 7.1]{FaMaNi11} for
the precise statement), while the loss due to the small spectral gap (induced
by the $h$-dependent deformation) is compensated by the small error of the
adiabatic expansion which is now given by $\varepsilon^{1-\delta}$ where
$\delta\in\left(  0,1\right)  $.

We study the adiabatic evolution problem for resonant states in connection
with models of mesoscopic transport. Our aim is to avoid the unphysical
modification of the selfadjoint operator introduced in \cite{FaMaNi11}. Our
approach consists in using adapted Hilbert spaces that contain the relevant
resonant states, with the goal of having an adiabatic approximation to all
orders in $\varepsilon$, over time intervals of length $\varepsilon^{-N}$ for
any fixed $N\geq0$. In our framework, the evolution is no longer unitary and
our first result says that we can arrange so that the generator of our
evolution has an imaginary part which is $\leq\varepsilon^{N}$ for any $N$. A
second result (for the moment limited to the case of one space dimension),
gives appropriate control on the resolvent \emph{in the same spaces}, and we
get adiabatic approximation over long time intervals for exact solutions of
linear adiabatic evolution equations. (The multidimensional case will
be attacked in a future work.)

\subsection{Aims and ideas}\label{aim_ideas}
The aim of the present paper is to establish asymptotic approximations
for solutions of adiabatic evolution equations of the form
\begin{equation}\label{int.6}
(\varepsilon D_t+P(t))u(t)=0
\end{equation}
that are valid over time intervals of length $\varepsilon ^{-N}$, with
errors ${\cal O}(\varepsilon ^N)$ for arbitrary $N>0$. Here
$P(t)=-h^2\Delta +V(t,x)$ is a self-adjoint Schr\"odinger operator
with a single well in a potential island which is assumed to generate
a shape resonance with real energy $E=E(t)$. Typically,
$\varepsilon $ should be comparable to the tunneling relaxation time
for $P(t)$ on a logarithmic scale. More precisely, we should have $\ln
(1/\varepsilon )\asymp 1/h$. The approximations should be of the form
\begin{equation}\label{int.7}
u_\mathrm{ad}(t)\sim (\nu _0(t)+\varepsilon \nu
_1(t)+...)e^{-\frac{i}{\varepsilon }\int^t \lambda (s,\varepsilon )ds}
\end{equation}
where $\nu _j(t)$ are well-behaved smooth functions of $t$ with values
in the (by assumption) common domain of the $P(t)$, and
\begin{equation}\label{int.8}
\lambda (t,\varepsilon )\sim \lambda _0(t)+\varepsilon \lambda _1(t)+...,
\end{equation}
where $\lambda _0(t)$ is a shape resonance of $P(t)$, satisfying
\begin{equation}\label{int.9}
0\le -\Im \lambda _0(t)\le e^{-2(1+o(1))S(t)/h},\ \Re \lambda _0(t)=E(t)+o(t),
\end{equation}
and $\nu _0(t)$ is a corresponding resonant state; $(P(t)-\lambda
_0(t))\nu _0(t)=0$. Here $S(t)>0$ is the Lithner-Agmon distance for $P(t)-E(t)$
from the potential well to the sea surrounding the potential island.

Since the non-linearity is concentrated to a bounded region, we
can choose an ambient Hilbert function space ${\cal H}$ quite freely
such that the space of restrictions of its elements to
some neighborhood $\Omega $ of the island is equal to $L^2(\Omega )$.

\par One such choice is ${\cal H}=L^2({\bf R}^n)$. A nice feature with
this choice is that the evolution (\ref{int.6}) is norm preserving:
$\| u(t)\|_{\cal H}=\| u(s)\|_{\cal H}$. A difficulty with this choice is
that the resonant state $\nu _0(t)$ does not belong to $L^2({\bf R}^n)$
(and $\lambda _0(t)$ does not belong to the $L^2$-spectrum of
$P(t)$), so the adiabatic expansion (\ref{int.7}) can hold only
locally in ${\bf R}^n$.

\par Rather, we choose ${\cal H}$ to be a Hilbert space that contains
the resonant state $\nu _0(t)$ and such that $\lambda _0(t)$
belongs to the ${\cal H}$-spectrum of $P(t)$.
When replacing $L^2$ with some other Hilbert space we lose (in
general) the self-adjointness of the operators $P(t)$ and the
corresponding unitarity of the evolution (\ref{int.6}). The original
non-linear problem is not time reversible, so we only wish to have a
good control of the solutions in the forward time direction. If we can
choose ${\cal H}$, depending on $\varepsilon $ but not on $t$, so that
\begin{equation}\label{int.10}
\Im P(t)\le \tau (\varepsilon ),
\end{equation}
for $P(t)$ as an unbounded operator ${\cal H}\to {\cal H}$, where
$\tau (\varepsilon )\ge 0$, then if $u(t)$ solves (\ref{int.6}) for
$t$ in some interval, we would have
\begin{equation}\label{int.11}
\| u(t)\|_{\cal H}\le e^{\tau (\varepsilon )(t-s)}   \| u(s)\|_{\cal
  H},\hbox{ for }t\ge s,
\end{equation}
so the solution will grow at most exponentially with rate $\tau
(\varepsilon )$ in the direction of increasing time. Correspondingly,
we can expect well-posedness for solutions of the initial value
problem,
\begin{equation}\label{int.12}
\begin{cases}
(\varepsilon D_t+P(t))u(t)=0,\ 0\le t\le T,\\
u(0)=u_0,
\end{cases}
\end{equation}
assuming, to fix the ideas, that $P(t)$ is defined for $t\in [0,T]$,
$T>0$. Then $\| u(t)\|_{\cal H}\le e^{\tau (\varepsilon )T}\|
u_0\|_{\cal  H}$ and in order to avoid exponential growth of the upper
bound we require
$$
\tau (\varepsilon )T\le {\cal O}(1).
$$
With $T=\varepsilon ^{-N_0}$ for some fixed $N_0>0$ we then need
$\tau (\varepsilon )\le {\cal O}(\varepsilon ^{N_0})$. Assume that we can
perform the adiabatic constructions in (\ref{int.7})--(\ref{int.9})
with $\Im \lambda (t,\varepsilon )\le \varepsilon ^{N+1}$, $N\ge N_0$, and that we have
(\ref{int.10}) with $\tau (\varepsilon )\le \varepsilon ^{N+1}$ for
$0\le t\le T$, $T\le \varepsilon ^{-N_0}$. Then taking a suitable
realization of $u_\mathrm{ad}$ we expect to have
$$
(\varepsilon D_t+P(t))u_\mathrm{ad}={\cal O}(\varepsilon ^{N+1})\hbox{
in }{\cal H},\ \| u(t)\|_{\cal  H}={\cal O}(1),
$$
for $0\le t\le T$ and using (\ref{int.11}) we would be able to solve
$$
(\varepsilon D_t+P(t))v=-(\varepsilon D_t+P(t))u_\mathrm{ad},\ \| v(0)\|_{\cal H}=0,
$$
with $\| v\|={\cal O}(\varepsilon ^{N-N_0})$. Then
$u(t)=u_\mathrm{ad}(t)+v(t)$ is an exact solution of $(\varepsilon
D_t+P(t))u=0$ on $[0,T]$ with $u-u_\mathrm{ad}={\cal O}(\varepsilon
^{N-N_0})$ in ${\cal H}$.

The easiest choice of ${\cal H}$, at first sight, would be to follow the
method of exterior complex distortions \cite{AgCo71, BaCo71, Si79,
  Hu86} in the spirit of \cite{SjZw91} so that
${\cal H}=L^2(\Gamma )$, where $\Gamma \subset {\bf C}^n$ is a totally
real manifold of real dimension $n$, obtained as a deformation of
${\bf R}^n$, coinciding with ${\bf R}^n$ along the island. We did non
succeed with this particular choice however (see further comments
below). A.\ Faraj, A.\ Mantile and F.\ Nier \cite{FaMaNi11} followed
this path. They defined an operator $\widetilde{P}(t)$ from $P(t)$
living on a distorted contour having an artificial interface condition
between real part of the contour near the island and the complex
distorted part. This way the discrete spectrum of $P(t)$ needs no
longer to consist of resonances of $P(t)$, so the exact link with the original
evolution problem is disrupted.

\par In this paper we use the spaces of \cite{HeSj86}. Such spaces
${\cal H}=H(\Lambda _{\upsilon G})$, $0\le \upsilon \ll 1$ are defined
with the help of a suitable real symbol $G(x,\xi )$ vanishing for
large $|\xi |$, and very roughly $H(\Lambda _{\upsilon G})$ is the
space of functions $u(x)$ on ${\bf R}^n$, such that
$\widetilde{u}(x,\xi )e^{-\upsilon G(x,\xi )/h}\in L^2({\bf R}^{2n})
$, where $\widetilde{u}$ denotes a suitable FBI transform. The
associated ``I-Lagrangian'' manifold $\Lambda _{\upsilon G}$ is given
by $\Im (x,\xi )=\upsilon H_G(\Re (x,\xi ))$, where
$H_G=\partial _\xi G\cdot \partial _x-\partial _xG\cdot \partial _\xi
\simeq (\partial _\xi G,-\partial _xG)$ is the Hamilton field of
$G$. $\Lambda _{\upsilon G}$ is then the natural classical phase space
associated with $H(\Lambda _{\upsilon G})$. Thus if we consider a
semi-classical Schr\"odinger operator $P=-h^2\Delta +V(x)$ with
leading symbol $p(x,\xi )=\xi ^2+V(x)$ where $V$ extends far enough in
the complex domain, the natural leading symbol of the unbounded
operator $P:{\cal H}\to {\cal H}$ is
${{p}_\vert}_{\Lambda _{\upsilon G}}$. Here, by Taylor expansion we
have
\[
\begin{split}
{{p}_\vert}_{\Lambda _{\upsilon G}}(x,\xi )&=p(\Re (x,\xi ))+i\upsilon H_G(p)(\Re
(x,\xi ))+{\cal O}(\upsilon ^2)\\
&=p(\Re (x,\xi ))-i\upsilon H_p(G)(\Re
(x,\xi ))+{\cal O}(\upsilon ^2),
\end{split}
\]
locally uniformly (and also globally after putting the right
order function into the remainder). If $H_pG\ge 0$, then (at least up
to ${\cal O}(\upsilon ^2)$) we have $\Im {{p}_\vert}_{_{\Lambda _{\upsilon G}}}\le 0$
and this is a first step towards having (\ref{int.10}) for $P$. In
order to define a resonance $\lambda_0$ for $P$ we need to choose $G$
and $\upsilon $ so that near infinity $\Im {{p}_\vert}_{\Lambda _{\upsilon G}}<\Im
\lambda _0$ when $\Re {{p}_\vert}_{\Lambda _{\upsilon G}} $ belongs to a
neighborhood of $\Re \lambda _0$. This can be obtained by choosing $G$
to be an escape function, meaning roughly that $H_pG>1/{\cal O}(1)$
near infinity on $p^{-1}(E_0)$ where $E_0$ is some fixed real energy
and we assume that $\Re \lambda _0\approx E_0$.

\par The method of (small) contour dis\-tor\-tions follows this scheme
with the symbol $G(x,\xi )$ chosen to be linear or possibly affine linear in
$\xi $. With such restrictions it is harder (maybe impossible) to find $G$
so that $H_pG\ge 0$ everywhere and $H_pG>0$ near infinity on
$p^{-1}(E_0)$.

\par For the construction of formal adiabatic solutions we will also
need a good control over $(P(t)-z)^{-1}$ for $z$ on some small closed
contour enclosing $\lambda _0(t)$.

\subsection{The main results}\label{mr} The results of this paper
concern the linear adiabatic theory for time dependent potentials with
a well in an island, in a fairly general setting. We hope that they will be useful
for non-linear problems of the type described above and also that they
are interesting in their own right. We study
\begin{itemize}
\item[1)]Semi-boundedness as in (\ref{int.10})
\item[2)]Resolvent estimates in $H(\Lambda _{\upsilon G})$-spaces
\item[3)]Adiabatic approximations over long time intervals.
\end{itemize}
Here 3) will be a fairly direct consequence of 1) and 2), using
general arguments from adiabatic constructions, that we shall review
in Section \ref{1eig}, see also Section \ref{fad}. 

\par In Sections \ref{opesc}, \ref{nore} we review some of the theory
in \cite{HeSj86}. Let $r(x),\, R(x)$ be positive smooth functions on
${\bf R}^n$ satisfying (\ref{opesc.1}):
$$
r\ge 1,\ rR\ge 1.
$$
Define $\widetilde{r}(x,\xi )=(r(x)^2+\xi ^2)^{1/2}$ as in
(\ref{opesc.2}) and the symbol spaces $S(m)$ as in (\ref{opesc.3}). We
assume (\ref{opesc.4}):
$$
m\in S(m),\ r\in S(r),\ R\in S(R).
$$
Let $1\le m_0\in S(m_0)$. We consider the formally self-adjoint
semi-classical differential operator in (\ref{opesc.5}):
$$
P=\sum_{|\alpha |\le N_0} a_\alpha (x;h)(hD)^\alpha ,\ a_\alpha \in
S(m_0r^{-|\alpha |}),
$$
where $a_\alpha $ is a finite sum in powers of $h$ as in
(\ref{opesc.6}) with leading term $a_{\alpha ,0}(x)$. The
full symbol of $P$ for the standard left quantization will also be
denoted by $P$ (see (\ref{opesc.7})) and the semi-classical principal
symbol will be denoted by $p(x,\xi )$ ((\ref{opesc.13})):
$$
p(x,\xi )=\sum_{|\alpha |\le N_0} a_{\alpha ,0}(x)\xi ^{\alpha }\in S(m)
$$
where $m=m_0(x)(\widetilde{r}(x,\xi )/r(x))^{N_0}$. We also have
$P(x,\xi ;h)\in S(m)$ (uniformly with respect to $h$). We make the
ellipticity assumption (\ref{opesc.14}), where $p_\mathrm{class}$ is
the classical (PDE) principal symbol in (\ref{opesc.12}). Then for
every fixed real $E_0$ the
energy surface $\Sigma _{E_0}=p^{-1}(E_0)$ has the property that
$$
|\xi |\le {\cal O}(r(x)),\hbox{ for }(x,\xi )\in \Sigma _{E_0}.
$$
We say (see Definition \ref{opesc1}) that the real-valued function
$G\in S(\widetilde{r}R)$ is an escape function if
\begin{equation}\label{int.3}
H_pG(\rho )\ge \frac{m(\rho )}{{\cal O}(1)}\hbox{ on }\Sigma
_{E_0}\setminus K,
\end{equation}
for some compact set $K$. We make the technical assumption
(\ref{opesc.22}) (there stated with $E_0$ replaced with $0$, a
reduction obtained by replacing $p$ with $p-E_0$):
\begin{equation}\label{opescint.22}\begin{split}
&\hbox{For every }r_0>0,\hbox{ there exists }\epsilon _0>0,\hbox{ such
  that}\\
&|p-E_0|\ge \epsilon _0m\hbox{ on }{\bf R}^{2n}\setminus \bigcup_{\rho \in
  \Sigma _0}B_{g(\rho )}(\rho ,r_0).
\end{split}\end{equation}
Here $g$ is the natural metric associated to the scales $\widetilde{r}$, $R$, see (\ref{opesc.4.5}). 

Proposition \ref{opesc2} (where again we took the case $E_0=0$)  states that if we have an escape function for
a given energy $E_0$, then we can modify it on a bounded set to get an
escape function $G$ which vanishes on any given compact set, such that
$H_pG\ge 0$.

\par The main example we have in mind is that of the Schr\"odinger
operator
\begin{equation}\label{int.14}
P=-h^2\Delta +V(x),
\end{equation}
where $V$ is real-valued and
\begin{equation}\label{int.15}
\partial _x^\alpha V=o(\langle x\rangle^{-|\alpha |}),\ |x|\to \infty .
\end{equation}
Then we take $r(x)=1$, $R(x)=\langle x\rangle$, $m=\xi ^2$ and
$P(x,\xi ;h)=p(x,\xi )=\xi ^2+V(x)$. When $E_0>0$, we have the
escape function $x\cdot \xi $ and after multiplication with a cutoff
$\chi (p(x,\xi )-E_0)$ we can also assume that $G$ has compact support
in $\xi $. In this case (\ref{opescint.22}) holds automatically.

\par Let $G=G(x,\xi )$ be real-valued and sufficiently small in
$S(\widetilde{r}R)$. Let $\Lambda _G$ be the corresponding
I-Lagrangian manifold $\Im (x,\xi )=H_G(\Re (x,\xi ))$, given in
(\ref{nore.1}), which is also symplectic for $\Re \sigma $, where
$\sigma =\sum d\xi _j\wedge dx_j$ is the complex symplectic form on
${\bf C}^{2n}$. We assume that $|\xi |\le {\cal O}(r(x))$ on the
support of $G$ and define the weight function $H$ on $\Lambda _G$ by
(\ref{nore.7}). It is also of class $S(\widetilde{r}R)$ when using the
natural parametrization of $\Lambda _G$ in (\ref{nore.1}).

Let $T$ be an FBI-transform defined as in
(\ref{nore.2})--(\ref{nore.6}) so that $T:{\cal E}'({\bf R}^n)\to
C^\infty (\Lambda _G;{\bf C}^{n+1})$. If $m$ is an order function on
$\Lambda _G$ ($m\in S(m)$), we define the Hilbert spaces $H(\Lambda
_G,m)$ as in Definition \ref{nore2} and put $H(\Lambda _G)=H(\Lambda
_G,1)$. When $G=0$ this gives $L^2({\bf R}^n)$ with equivalence of
norms. In Section \ref{pfops} we review pseudodifferential operators,
Fourier integral operators and Toeplitz operators, acting on these
spaces.

Let $r$, $R$, $m(x,\xi )=m_0(x)(\widetilde{r}(x,\xi )/r(x))^{N_0} $ be
as above and let $P$ be a formally self-adjoint $h$-differential
operator as in (\ref{opesc.5}), (\ref{opesc.6}), fulfilling (\ref{opesc.14}) as well as the technical
assumption (\ref{opesc.22}). We also make the exterior analyticity assumption
(\ref{opesc.11}). If $G\in S(\widetilde{r}R)$ with
$|\xi |\le {\cal O}(r(x))$ on $\mathrm{supp\,}G$, then (cf.\
(\ref{nore.29})) we can view $P$ as a bounded operator
$$
P:\, H(\Lambda _{\upsilon G},m)\to H(\Lambda _{\upsilon G}),
$$
for $0\le \upsilon \ll 1$, provided that the coefficients $a_{\alpha ,k}$ of $P$ are analytic
in a neighborhood of the $x$-space projection of
$\mathrm{supp\,}G$. In Section \ref{sbd} we prove a first
semiboundedness result:
\begin{theo}\label{int1}
  Under the above assumptions, assume in addition that $P$ has an
  escape function $G_0$ at energy $E_0\in {\bf R}$. Let $K\subset {\bf
    R}^n$ be compact, containing the analytic singular support of $P$,
  i.e. the smallest closed set $\widetilde{K}$ such that the
  coefficients of $P$ (more precisely all the $a_{\alpha ,k}$ in
  {\rm (\ref{opesc.6})}) are analytic in ${\bf R}^n\setminus
  \widetilde{K}$. Then we can find an escape function at energy $E_0$;
$$
G(x,\xi ;h)\sim G^0+hG^1+...\hbox{ in }S(\widetilde{r}R),
$$
supported in $|\xi |\le {\cal O}(r(x))$,
where $G^0=G_0$ near infinity on $\Sigma _{E_0}$, $\pi
_x\mathrm{supp\,}G^j$, $\pi _x\mathrm{supp\,}G$ are disjoint from a
fixed neighborhood of $K$, such that for $P$ as a closed unbounded
operator: $H(\Lambda _{\upsilon G})\to H(\Lambda _{\upsilon G})$, we have
\begin{equation}\label{int.16}
\Im (Pu|u)_{H(\Lambda _{\upsilon G})}\le \upsilon {\cal O}(h^\infty )\| u
\|^2_{H(\Lambda _{\upsilon G},m^{1/2})},
\end{equation}
for $\upsilon \ge 0$ and $h>0$ small enough. In the Schr\"odinger case
($m=\langle \xi \rangle ^2$), we can replace $H(\Lambda
_{\upsilon G},m^{1/2})$ with $H(\Lambda _{\upsilon G})$.
\end{theo}

\par As we shall see, we can arrange so that
$$
\| u\|_{H(\Lambda _{\upsilon G})}=\| u\|_{L^2},\hbox{ when
}u\in L^2_{\mathrm{comp}}(K).
$$

\par The half-estimate (\ref{int.16}) is of independent interest, when compared with
various versions of G\aa{}rding's inequality, but for adiabatic theory
in connection with a potential well in an island, we would like to
replace ${\cal O}(h^N)$ for every $N>0$ with ${\cal O}(e^{-N/h})$ for
every $N>0$. This improvement will be obtained with a scaling
argument.

Keeping the above assumptions, we also assume (\ref{far.1}),
(\ref{far.2}):
$$
r=1,\ R(x)=\langle x\rangle,\ m_0(x)=1,\ m=\langle \xi \rangle^{N_0},
$$
as well as (\ref{far.3}) which states that $p(x,\xi )$ converges to a
limiting polynomial $p_\infty (\xi )$ when $x\to \infty $ in the
natural sense for the semi-norms of $S(m)$. The main fact that we
exploit is now that in a region where $x=\mu \widetilde{x}$,
$|\widetilde{x}|\asymp 1$, $P(x,hD_x;h)$ can be viewed as an
$\widetilde{h}$-differential operator $P(\mu
\widetilde{x},\widetilde{h}D_{\widetilde{x}};h)$ with
$\widetilde{h}=h/\mu $. Combination of this observation with Theorem
\ref{int1} leads to (cf.\ (\ref{far.28})):
\begin{theo}\label{int2}
We make the assumptions of Theorem {\rm \ref{int1}} and the two additional
assumptions above. Let $\pi _x:\, {\bf R}^n\times {\bf R}^n\ni (x,\xi
)\mapsto x\in {\bf R}^n$. Then uniformly for $\mu \in [1,+\infty [$, we can
find an escape function
$$
G(x,\xi ,\mu ;h)\sim G^0+hG^1+...\hbox{ in }S(\widetilde{r}R),
$$
with support in $\pi _x^{-1}(({\bf R}^n\setminus B(0,\mu ))\cap
({\bf
  R}_x^n\times B(0,r_0))$ for some fixed $r_0>0$ and with $G^0=G_0$ on
$\Sigma _{E_0}\cap \pi _x^{-1}({\bf R}^n\setminus B(0,2\mu ))$,
independent of $\mu $ for $|x|\ge 2\mu $, such that for $P$ as an
unbounded closed operator $H(\Lambda _{\upsilon G})\to H(\Lambda _{\upsilon G})$, we
have
\begin{equation}
\label{int.17}
\Im (Pu|u)_{H(\Lambda_{\upsilon G})}
\le \upsilon {\cal O}((h/\mu )^\infty )\| u\|^2_{H(\Lambda _{\upsilon G},m^{1/2})},\
0\le \upsilon \ll 1.
\end{equation}

In the Schr\"odinger operator case, we can replace $H(\Lambda
_{\upsilon G},m^{1/2})$ in {\rm (\ref{int.17})} with
$H(\Lambda_{\upsilon G})$.
\end{theo}

\par In this result we use a decoupling property which can be obtained with
a suitable choice of norm in $H(\Lambda _{\upsilon G})$, namely that
$$
\| u\|_{H(\Lambda _{\upsilon G},1)}=\| u\|_{L^2},\hbox{ when
}\mathrm{supp\,}u\subset B(0,\mu /2).
$$

\par We next consider resolvent estimates. For simplicity we assume right
away that $P$ is a semi-classical Schr\"odinger operator as in
(\ref{int.14}), (\ref{int.15}). We will also assume that we are in the
1-dimensional case; $n=1$, even though we now think that the higher
dimensional case is within reach. With the higher dimensional case in
mind we will formulate certain statements as if we were in that general
case, even though (for the moment) $n=1$. In order to fit with
(\ref{int.15}), we take
$$
r(x)=1,\ R(x)=\langle x\rangle,\ m=\langle \xi \rangle^2,
$$
and note that $P(x,\xi ;h)=p(x,\xi )=\xi ^2+V(x)$. We keep the
exterior analyticity assumption (\ref{opesc.11}) which takes the form
(\ref{rest.3}). Let $E>0$. Let us first consider the non-trapping
case (cf.\ Proposition \ref{rest3}).
\begin{theo}\label{int3}
Assume that the $H_p$-flow on $\Sigma _E=p^{-1}(E)$ is
non-trapping (in the sense that no maximal $H_p$ trajectory in
$p^{-1}(E)$ is contained in a compact set). Let $G$ be as in Theorem {\rm \ref{int2}}, where we choose
$\mu =h/\epsilon $ where $0<\epsilon \ll h$ is a small parameter. Let
$\vartheta >0$ be small and fix $\upsilon >0$ sufficiently small. If $\delta
_0>0$, $C>0$ are respectively small and large enough, then for $z$ in
the range {\rm (\ref{rest.45})}:
$$
\Re z\in [E-\delta _0/2,E+\delta _0/2],\ -\epsilon _\vartheta /C \le
\Im z \le 1/C,
$$
we have that $P-z:H(\Lambda _{\upsilon \epsilon G},\langle \xi \rangle ^2)\to
H(\Lambda _{\upsilon \epsilon G})$ is bijective and
$$
m_\epsilon (x;h)^{\frac{1}{2}}(z-P)^{-1}m_\epsilon
(x;h)^{\frac{1}{2}}={\cal O}(1):\, H(\Lambda _{\upsilon \epsilon G})\to
H(\Lambda _{\upsilon \epsilon G},\langle \xi \rangle^2).
$$
Here we have put
$$m_\epsilon (x,\xi ):=\frac{h}{\langle x\rangle ^{1+\vartheta
  }}+\epsilon _\vartheta ,\ \ \epsilon _\vartheta =\left(\frac{\epsilon
  }{h} \right)^{\vartheta }\epsilon .$$
\end{theo}

We next consider a trapping case, namely that of a potential well in
an island. Let $E>0$ and keep the assumptions above, except the one
about non-trapping. Let $\ddot{\mathrm{O}}\Subset {\bf R}^n$ be open (still with
$n=1$) and let $U_0\subset \ddot{\mathrm{O}}$ be a compact
subset. Assume (\ref{rest.50}), (\ref{rest.51}):
$$
V-E<0\hbox{ in }{\bf R}^n\setminus \overline{{\ddot{\mathrm{O}}}},\ V-E>0 \hbox{ in
}{\ddot{\mathrm{O}}}\setminus U_0,\ V-E\le 0\hbox{ in }U_0,
$$
$$
\mathrm{diam}_dU_0=0,
$$
where $d$ is the Lithner-Agmon distance associated to the metric
$$
(V-E)_+dx^2,\ \ (V-E)_+=\max (V-E,0).
$$
Assume (\ref{rest.52}):
$$\hbox{The }H_p\hbox{-flow has no trapped trajectories in
}{{p^{-1}(E)}_\vert}_{{\bf R}^n\setminus \ddot{\mathrm{O}}}.
$$
Let $M_0\subset \ddot{\mathrm{O}}$ be a connected compact set with
smooth boundary (i.e. a compact interval in the present 1D case) such
that (\ref{rest.53}) holds:
$$
M_0\supset \{x\in {\ddot{\mathrm{O}}}; d(x,\partial {\ddot{\mathrm{O}}})\ge \epsilon _0 \},
$$
for some small $\epsilon _0>0$. Let $P_0$ denote the Dirichlet
realization of $P$ on $M_0$.

\par Let $J(h)\subset {\bf R}$ be an interval tending to $\{ E \}$ as
$h\to 0$. Assume (\ref{rest.65}):
$$
P_0\hbox{ has no spectrum in } \partial J(h)+[-\delta (h),\delta (h)] ,
$$
where the parameter $\delta (h)$ is small but not exponentially small;
$$
\ln \delta (h)\ge -o(1)/h.
$$
$\sigma (P_0)\cap J(h)$ is a discrete set of the form
$\{ \mu _1(h),...,\mu _m(h) \}$ where $m=m(h)={\cal O}(h^{-n})$ and we
repeat the eigenvalues according to their multiplicity. Let
$\Gamma (h)$ denote the set of resonances of $P$ in
$J(h)-i[0,\epsilon_\vartheta  /C]$, $C\gg 1$, also
repeated according to their (algebraic) multiplicity.
Assume (\ref{rest.65.5}):
$\epsilon \ge e^{-1/(Ch)},
$
for some $C\gg 1$, so that
$$
\left(\frac{\epsilon }{h} \right)^\vartheta \epsilon  \ge e^{\frac{1}{{\cal O}(h)}-\frac{2S_0}{h}}.
$$
Here $S_0>0$ denotes the Lithner-Agmon distance from $U_0$ to $\partial
\ddot{\mathrm{O}}$. In Proposition \ref{rest6} we recall a result from
\cite{HeSj86} when $V$ is analytic everywhere and due to Fujii\'e, Lahmar-Benbernou, Martinez \cite{FuLabeMa11}, for potentials that are merely smooth on
a bounded set, stating that there is a bijection $b:\,\{ \mu_1,...,\mu _m \}\to \Gamma (h)$ such that
$$
b(\mu )-\mu =\widetilde{{\cal O}}(e^{-2S_0/h}):={\cal O}(e^{\omega
  -2S_0)/h}),\ \omega =\omega (\epsilon _0)\to 0,\ \epsilon _0\to 0.
$$
We give a proof in Section \ref{rest}.
\begin{theo}\label{int4}
For $C\gg 1$ sufficiently large, let
\begin{equation}\label{int.18}
z\in \{z\in J(h)+i]-\epsilon _\vartheta /C,1/C[;\
\mathrm{dist\,}(z,\sigma (P_0)\cap J(h))=\mathrm{dist\,}(z,\sigma
(P_0))\} .
\end{equation}
Assume either that $m$ (the number of elements in $\Gamma (h))$ is
equal to {\rm 1}, or that $\mathrm{dist\,}(z, \sigma (P_0))\ge \widetilde{{\cal
    O}}(e^{-2S_0/h})$. Then we have {\rm (\ref{rest.106})}:
$$  (z-P)^{-1}={\cal O}(h/\delta )+{\cal O}(1)+{\cal
    O}(h/\mathrm{dist\,}(z,\Gamma )): \ m_\epsilon ^{\frac{1}{2}}{\cal H}_{\mathrm{sbd}}\to
  m_\epsilon ^{-\frac{1}{2}}{\cal D}_{\mathrm{sbd}},
$$ where the first two terms are holomorphic in the interior of the
set {\rm (\ref{int.18})}. Here
$$
{\cal H}_\mathrm{sbd}=H(\Lambda _{\upsilon \epsilon G}),\ {\cal
  D}_\mathrm{sbd}=H(\Lambda _{\upsilon \epsilon G},\langle \xi \rangle^2)
$$
and $\upsilon $, $G$ are as in Theorem {\rm \ref{int2}}.
\end{theo}
We next turn to adiabatic expansions. Let $I\subset {\bf R}$ be an interval and let
\begin{equation}\label{bicsint.1}
V_t=V(t,x)\in C_{b}^\infty (I\times {\bf R}^n;{\bf R}).
\end{equation}
Here $C_b^\infty (\Omega )$ denotes the space of smooth functions on
$\Omega $ that are bounded with all their derivatives.
We assume (cf.\ (\ref{rest.3}))
\begin{equation}\label{bicsint.2}
\begin{split}
&V_t\hbox{ has a holomorphic extension (also denoted } V_t \hbox{)
  to}\\
&\{x\in {\bf C}^n;\, |\Re x|>C,\ |\Im x|<|\Re x|/C \}\\&\hbox{such that
}
V_t(x)=o(1),\ x\to \infty .
\end{split}
\end{equation}
\begin{equation}\label{bicsint.3}\begin{split}
\partial _tV_t(x)=0\hbox{ for }|x|\ge C,\hbox{ for some constant }C>0
\end{split}
\end{equation}
It is tacitly assumed that $V(t,x)$ does not depend on $h$. However,
when considering a narrow potential well in an island, of diameter
$\asymp h$, we will have to make an exception and
allow such an $h$-dependence  in a small neighborhood of the well.

\par Let $0<E_-<E_-'<E_+'<E_+<\infty $ and let
\begin{equation}\label{bicsint.4}
E_0(t)\in C_b^\infty (I;[E_-',E_+']).
\end{equation}
We assume that $V_t-E_0(t)$ has a  potential well in an island as above.

\par Let $\ddot{\mathrm{O}}=\ddot{\mathrm{O}}(t)\Subset {\bf R}^n$ be a
connected open set and let $U_0(t)\subset \ddot{\mathrm{O}}(t)$ be
compact. Assume (cf.\ (\ref{rest.50})), still with $n=1$,
\begin{equation}\label{bicsint.5}
V_t-E_0(t)\begin{cases}
<0 \hbox{ in } {\bf R}^n\setminus \overline{\ddot{\mathrm{O}}}(t),\\
>0 \hbox{ in } \ddot{\mathrm{O}}(t)\setminus U_0(t),\\
\le 0\hbox{ in }U_0(t),
\end{cases}
\end{equation}
\begin{equation}\label{bicsint.6}
\mathrm{diam}_{d_t}(U_0(t))=0.
\end{equation}
Here $d_t$ is the Lithner-Agmon distance on $\ddot{\mathrm{O}}(t)$,
given by the metric $(V_t-E_0(t))_+dx^2$.

\par Also assume that with $p_t=\xi^2 + V_t(x)$,
\begin{equation}\label{bicsint.7}
\hbox{the }H_{p_t}\hbox{-flow has no trapped trajectories in
}{p_t^{-1}(E_0(t))_|}_{{\bf R}^n\setminus \ddot{\mathrm{O}}(t)}.
\end{equation}
It follows that
\begin{equation}\label{bicsint.8}
d_xV_t\ne 0 \hbox{ on }\partial \ddot{\mathrm{O}}(t),
\end{equation}
so $\partial \ddot{\mathrm{O}}(t)$ is smooth and depends smoothly on
$t$. Thus $\ddot{\mathrm{O}}(t)$ is a manifold with smooth boundary,
depending smoothly on $t$. Further, $U_0(t)$ depends continuously on
$t$.

\par For $\epsilon _0>0$ small, we define
\begin{equation}\label{bicsint.9}
M_0(t)=\{ x\in \ddot{\mathrm{O}}(t);\, d_t(x,\partial
\ddot{\mathrm{O}}(t))\ge \epsilon _0 \} ,
\end{equation}
so $M_0(t)\Subset \ddot{\mathrm{O}}(t)$ is a compact set with smooth
boundary, depending smoothly on $t$. (Here we use the structure of
$d_t(x,\ddot{\mathrm{O}}(t))$ that follows from (\ref{bicsint.8}), see
Sections 9, 10 of~\cite{HeSj86}.) More precisely, for every fixed $t$
(consequently suppressed from the notation) the function $\phi(x)=d(x,\partial \ddot{\mathrm{O}})$ in
$\overline{\ddot{\mathrm{O}}}$ is analytic in the interior, continuous up
to the boundary, where it vanishes, and solves the eikonal equation
$|\nabla \phi |^2=V(x)$. Over a neighborhood of the boundary, we have
the Lagrangian manifold $\Lambda $, defined as the flow out of
$\partial \ddot{\mathrm{O}}\times \{\xi =0 \}$ under the Hamilton flow
of $q=\xi ^2-V(x)$. The manifold $\Lambda$ has a simple fold over the boundary and
$\Lambda _\phi : \xi =\phi '(x)$ describes ``one of the two covering halves'' of $\Lambda
$. It is quite well known then that if we choose analytic local
coordinates $y=(y_1,...,y_{n-1},V)=(y',V)$ near a boundary point, then
$\phi $ is an analytic function of $y',V^{1/2}$ and has a convergent
expansion $\phi =a_3(y)V^{3/2}+a_4(y)V^2+...$ with $a_k$
analytic and $a_3>0$.

\par Since we would like to allow $I$ in (\ref{bicsint.1}) to be a very long interval, we introduce the following uniformity assumption:
\begin{equation}
\label{bicsint.9.5}
\begin{split}  &(V_t,E_0(t))\in {\cal K},\ \forall t\in I,\hbox{ where }{\cal K}
  \hbox{ is a compact subset of }\\
&\{V\in C_b^\infty ({\bf R}^n;{\bf R});\, V\hbox{ satisfies (\ref{bicsint.2}) with
  a fixed constant }C
\}\times [E_-',E_+']\\
&\hbox{such that }(V,E)\hbox{ satisfies the assumptions (\ref{bicsint.3}),
  with a fixed }C\\ &\hbox{as well as
(\ref{bicsint.5}), (\ref{bicsint.6}), (\ref{bicsint.7}).}
\end{split}
\end{equation}

\par Let $P_0(t)$ denote the Dirichlet realization of $P(t)=-h^2\Delta
+V_t(x)$ on $M_0(t)$. If we enumerate the eigenvalues of $P_0(t)$ in
$]E_-,E_+[$ in increasing order (repeated with multiplicities) we know
(as a general fact for 1-parameter families of self-adjoint
operators), that they are uniformly Lipschitz functions of $t$. Let
$\mu_0 (t)=\mu_0 (t;h)$ be such an eigenvalue and assume,
\begin{equation}\label{bicsint.10}
\mu_0 (t;h)=E_0(t)+o(1),\ h\to 0,\hbox{ uniformly in }t.
\end{equation}
\begin{equation}\label{bicsint.11}\begin{split}
&\mu_0 (t;h) \hbox{ is a simple eigenvalue and}\\ &\sigma (P_0(t))\cap
[E_0(t)-\delta (h),E_0(t)+\delta (h)]=\{ \mu_0 (t;h) \}.
\end{split}
\end{equation}
Here, as above, $\delta (h)>0$ is small but not
exponentially small,
\begin{equation}\label{bicsint.12}
\ln \delta (h)\ge -o(1)/h,\ h\to 0.
\end{equation}
In addition to (\ref{rest.65.5}), we assume (\ref{bicsny.54}), so we
have (\ref{bics.111.5}):
$$
e^{-1/(C_1h)}\le \epsilon \le \min (h/C_2,\delta ),\ C_1,\, C_2\gg 1.
$$

Let $\lambda _0(t)$ be the unique resonance of $P(t)$ in
$D(\mu_0 (t),\delta (h))$ (the open disc in ${\bf C}$ with center $\mu
_0(t)$ and radius $\delta (h)$), so that $\lambda _0(t)-\mu _0(t)=\widetilde{{\cal
    O}}(e^{-2S(t)/h})$. As we shall see in (\ref{bics'.16}):
$$
\partial _t^k\lambda _0(t)={\cal O}(\delta(h) ^{-k}),\
k\ge 1.
$$

\par Let $e^0(t)$ be the corresponding resonant state, $(P(t)-\lambda
_0(t))e^0(t)=0$, uniquely determined up to a factor $\pm 1$ (that we
take independent of $t$) by the condition
$$
\int_{{\bf R}^{n}}e^0(t,x)^2 dx=1.
$$
As we shall see in Section \ref{bics}, we can find an escape function
$G$ as in Theorem \ref{int2} which applies simultaneously to all
$(P(t),E(t))$. Moreover, we can choose $G$ such that $G(x,-\xi
)=-G(x,\xi )$ and with this choice the bilinear scalar product
$$
\langle u|v\rangle =\int u(x)v(x)dx
$$
is well-defined and bounded on $H(\Lambda _{\upsilon \epsilon G})\times H(\Lambda
_{\upsilon \epsilon G})$. Then  $e^0={\cal O}(1)$ in $H(\Lambda
_{\upsilon \epsilon
  G})$. Recall that ${\cal H}_{\mathrm{sbd}}=H(\Lambda _{\upsilon \epsilon G})$, ${\cal
  D}_{\mathrm{sbd}}=H(\Lambda _{\upsilon \epsilon G},\langle \xi \rangle^2)$
where $\upsilon >0$ is small and fixed. Then we have (\ref{bicsny.48}):
$$
\partial _t^ke^0(t)={\cal O}(1) \left(\frac{h}{\epsilon _\vartheta ^2}
\right)^k\hbox{ in } {\cal D}_\mathrm{sbd}\hbox{ for }k\ge 0.
$$
Using the resolvent estimates in Theorem \ref{int4} we will establish
(as Proposition \ref{bics2}) the following result:
\begin{theo}\label{int5}
Under the assumptions above there exist two formal asymptotic series,
\begin{equation}\label{bicsint.65}
\nu (t,\varepsilon )\sim \nu _0(t)+\varepsilon \nu _1(t)+\varepsilon
^2 \nu _2(t)+...\ \hbox{in }C^\infty (I;{\cal D}_{\mathrm{sbd}}),
\end{equation}
\begin{equation}\label{bicsint.66}
\lambda (t,\varepsilon )\sim \lambda _0(t)+\varepsilon \lambda _1(t)+\varepsilon
^2 \lambda _2(t)+...\ \hbox{in }C^\infty (I),
\end{equation}
such that
\begin{equation}\label{bicsint.67}
(\varepsilon D_t+P(t)-\lambda (t,\varepsilon ))\nu (t,\varepsilon
)\sim 0
\end{equation}
as a formal asymptotic series in $C^\infty (I;{\cal H}_\mathrm{sbd})$. Here,
\begin{equation}\label{bicsint.68}
  \partial _t^k\nu _j=   {\cal O}(1) (h/\hat{\epsilon }_\vartheta ^2)
  ^{2j+k}\hbox{ in }{\cal D}_\mathrm{sbd},\ j\ge 0,\ k\ge 0,
\end{equation}
\begin{equation}\label{bicsint.69}
\partial _t^k\lambda _j= {\cal O}(1)(h/\hat{\epsilon }_\vartheta ^2
)^{2j-1+k},\ j\ge 1, k\ge 0,
\end{equation}
where
\begin{equation}\label{int70.5}
\hat{\epsilon }_\vartheta :=\frac{\epsilon _\vartheta }{\max
  (1,\epsilon _\vartheta /(\delta h) )^{1/2}}=\min (\epsilon _\vartheta
,(\epsilon _\vartheta\delta h) ^{1/2}).
\end{equation}
\end{theo}

We continue the discussion under the assumptions of Theorem
\ref{int5}. Put for $N\ge 1$
\begin{equation}\label{bicsint.84}
\nu ^{(N)}=\nu _0+\varepsilon \nu _1+...+\varepsilon ^N\nu _N,
\end{equation}
\begin{equation}\label{bicsint.85}
\lambda ^{(N)}=\lambda _0+\varepsilon \lambda _1+...+\varepsilon ^N
\lambda _N,\ \ N\ge 1.
\end{equation}
Then by the proof of Theorem \ref{int5} (cf.\ (\ref{1eig.4})),
\begin{equation}\label{bicsint.86}
(\varepsilon D_t+P(t)-\lambda ^{(N)})\nu ^{(N)}=r^{(N+1)},
\end{equation}
where
\begin{equation}\label{bicsint.87}
r^{(N+1)}=\varepsilon ^{N+1}D_t\nu _N-\sum_{j,k\le N \atop j+k\ge N+1}\varepsilon
^{j+k}\lambda _j\nu _k.
\end{equation}

\par In the following, we assume that
\begin{equation}\label{bicsint.88}
    \frac{\varepsilon ^{\frac{1}{2}}h}{\hat{\epsilon }_\vartheta ^2}\ll 1 .
\end{equation}
Recall from (\ref{bicsint.12}) that $\delta =\delta (h)$ is small, but
not exponentially small and that $\epsilon _\vartheta =(\epsilon
/h)^\vartheta \epsilon $. Then (\ref{bicsint.88}) holds if we assume that
$\varepsilon $ is exponentially small:
\begin{equation}\label{bicsint.89}
0<\varepsilon \le {\cal O}(1)\exp\left(-1/(Ch)\right),\hbox{ for some }C>0,
\end{equation}
and choose
\begin{equation}\label{bicsint.90}
\epsilon \ge \varepsilon ^{\frac{1}{4(1+\vartheta ) }-\alpha },
\end{equation}
for some $\alpha \in ]0,1/(4(1+\vartheta ))[$. We also assume
(\ref{bics.92.5}), stating that $\epsilon \le {\cal O}(\varepsilon
^{1/N_0})$ for some $N_0>0$. 

\par Having assumed (\ref{bicsint.88}) we get, as we shall see:
\begin{equation}\label{bicsint.91}
  \partial _t^k r^{(N+1)}={\cal O}(1)
  \varepsilon ^{\frac{1}{2}}\left(\frac{\varepsilon
      ^{\frac{1}{2}}h}{\hat{\epsilon }_\vartheta ^2}\right)^{2N+1}
  \left(\frac{h}{\hat{\epsilon }_\vartheta ^2} \right)^k\hbox{ in }{\cal D}_\mathrm{sbd}.
\end{equation}
Also,
\begin{equation}\label{bicsint.91.2}
  \| \nu ^{(N)}(t)\|_{{\cal H}_\mathrm{sbd}}
  =(1+{\cal O}(\varepsilon h^2 /\hat{\epsilon }_\vartheta ^4))\| \nu
  _0(t)\|_{{\cal H}_\mathrm{sbd}}\asymp 1.
\end{equation}

\par From (\ref{int.16}) with $\mu $ in (\ref{rest.26}), $\mu
=h/\epsilon $, we get:
\begin{equation}\label{bicsint.92}
\Im (P(t)u|u)_{{\cal H}_\mathrm{sbd}}\le {\cal
  O}(\epsilon ^\infty )\| u\|^2_{{\cal H}_\mathrm{sbd}}.
\end{equation}

\par Let $I\ni t\mapsto u(t)\in H(\Lambda _{\upsilon \epsilon G},\langle \xi \rangle)$
be continuous such that $\partial _tu$ is continuous with values in
$H(\Lambda _{\upsilon \epsilon G},\langle \xi \rangle^{-1})$. Assume that $u$ is a solution of
$$
(\varepsilon D_t+P(t))u(t)=0.
$$
It then follows from (\ref{bicsint.92}) that
\begin{equation}\label{bicsint.93}
\| u(t)\|_{{\cal H}_\mathrm{sbd}}\le e^{{\cal O}(\epsilon
^\infty )(t-s)}\| u(s)\|_{{\cal H}_\mathrm{sbd}},\ t\ge s.
\end{equation}

From (\ref{bicsint.92}) we will derive a resolvent bound which leads
to the fact that for every $u_0\in {\cal D}_\mathrm{sbd}$ and every
$s\in I$, there exist
$u_0\in C(I\cap [s,\infty [;{\cal D}_\mathrm{sbd})\cap C^1(I\cap
[s,\infty [;{\cal H}_\mathrm{sbd})$ such that
\begin{equation}\label{bicsint.96}
(\varepsilon D_t+P(t))u(t)=0\hbox{ for }s\le t\in I,\ \ u(s)=u_0.
\end{equation}
Again the solution satisfies (\ref{bicsint.93}). 
When $P(t)=P(t_0)$ is independent of $t$, this follows from the
Hille--Yosida theorem. In the general case we can use~\cite[Theorem 6.1 and Remark 6.2]{Ka70}. 

\par This allows us to define the forward fundamental matrix $E(t,s)$,
$I\ni t\ge s\in I$ of
$\varepsilon D_t+P(t)$:
$$
\begin{cases}(\varepsilon D_t+P(t))E(t,s)=0,\ t\ge s,\\ E(t,t)=1
\end{cases}
$$
and we have
\begin{equation}\label{bicsint.97}
\| E(t,s)\|_{{\cal L}({\cal H}_\mathrm{sbd},{\cal H}_\mathrm{sbd})}\le
\exp ((t-s){\cal O}(\epsilon ^\infty )),\ t\ge s,\ t,s\in I.
\end{equation}
If $v\in C (I; {\cal H}_\mathrm{sbd})$ vanishes for $t$ near $\inf I$,
we can solve $(\varepsilon D_t+P(t))u=v$
on $I$ by
$$
u(t)=\frac{i}{\varepsilon }\int _{\inf I}^tE(t,s)v(s) ds.
$$

\par Now return to (\ref{bicsint.84})--(\ref{bicsint.86}) with $\lambda _j$,
$\nu _j$ as in Theorem \ref{int5} and $r^{(N+1)}$ satisfying
(\ref{bicsint.91}). We notice that
\begin{equation}\label{bicsint.98}
\lambda ^{(N)}=\lambda _0+{\cal O}(1)\varepsilon ^{\frac{1}{2}}\frac{\varepsilon
  ^{\frac{1}{2}}h}{\hat{\epsilon }_\vartheta ^2}.
\end{equation}
We can choose $\nu _0(t)=e^0(t)$ implying that $\lambda _1=0$ and (\ref{bicsint.98}) improves to
\begin{equation}\label{bicsint.99}
\lambda ^{(N)}=\lambda _0+{\cal O}(1)\varepsilon ^{\frac{1}{2}}\left(\frac{\varepsilon
  ^{\frac{1}{2}}h}{\hat{\epsilon }_\vartheta ^2}\right)^3.
\end{equation}
See Remark \ref{bics3}. We choose $\nu _0=e^0$ in the remainder of this introduction.

\par Assume, to fix the ideas, that $0\in I$, and restrict the
attention to $I_+=\{ t\in I;\, t\ge 0 \}$. From (\ref{bicsint.86}), we
get
\begin{equation}\label{bicsint.103}(\varepsilon D_t+P(t))u^{(N)}=\rho
  ^{(N+1)},\ t\in I_+,\end{equation}
where
\begin{equation}\label{bicsint.104}
u^{(N)}=e^{-i\int_0^t \lambda ^{(N)} ds/\varepsilon }\nu ^{(N)},\
\rho ^{(N+1)}=e^{-i\int_0^t \lambda ^{(N)} ds/\varepsilon }r ^{(N+1)}.
\end{equation}
By (\ref{bicsint.91.2}), (\ref{bicsint.91}), we have
\begin{equation}\label{bicsint.105}
\| \rho ^{(N+1)}\|_{{\cal H}_\mathrm{sbd}}=
{\cal O}(1) \varepsilon ^{\frac{1}{2}}\left(\frac{\varepsilon
  ^{\frac{1}{2}}h}{\hat{\epsilon }_\vartheta ^2} \right)^{2N+1}
\| u^{(N)}\|_{{\cal H}_\mathrm{sbd}} .
\end{equation}
Using again (\ref{bicsint.92}), we get as in Section \ref{bics},
\begin{equation}\label{bicsint.105.2}
\| u^{(N)}(t)\|\le e^{{\cal O}(1)t\varepsilon ^{-1/2}(\varepsilon
^{1/2}h/\hat{\epsilon }_\vartheta ^2)^{2N+1}}\| u^{(N)}(0)\|,\
t\in I_+.
\end{equation}
Assume (\ref{bics.105.4}):
\begin{equation}\label{bicsint.105.4}
(\sup I)\varepsilon ^{-\frac{1}{2}}\left( \frac{\varepsilon
    ^{\frac{1}{2}}h}{\hat{\epsilon }_\vartheta ^2} \right)^{2N+1}
\le {\cal O}(1).
\end{equation}
Then, for $t\in I_+$,
\begin{equation}\label{bicsint.105.6}
  \| u^{(N)}(t)\|_{{\cal H}_\mathrm{sbd}}\le {\cal O}(1),\
  \| \rho ^{(N+1)}(t)\|_{{\cal H}_\mathrm{sbd}}\le {\cal O}(1) \varepsilon ^{\frac{1}{2}}\left( \frac{\varepsilon
      ^{\frac{1}{2}}h}{\hat{\epsilon }_\vartheta ^2} \right)^{2N+1}.
\end{equation}

\par Using the fundamental matrix $E$ to correct the error $\rho
^{(N+1)}$ we have the exact solution $u=u^{(N)}_\mathrm{exact}$,
\begin{equation}\label{bicsint.106}
u=u^{(N)}-\frac{i}{\varepsilon }\int_0^t E(t,s)\rho ^{(N+1)}(s) ds
\end{equation}
of the equation
$$
(\varepsilon D_t+P(t))u=0\hbox{ on }I_+.
$$

\par From (\ref{bicsint.105.4}) we get
\begin{equation}\label{bicsint.109}
\sup I\le \varepsilon ^{-N_0},
\end{equation}
for some fixed finite $N_0$. Then by (\ref{bicsint.97})
\begin{equation}\label{bicsint.110}
\| E(t,s)\|_{{\cal L}({\cal H}_\mathrm{sbd},{\cal H}_\mathrm{sbd})}\le
e^{{\cal O}(\varepsilon ^\infty )}=1+{\cal O}(\varepsilon ^\infty ),
\end{equation}
and using this and (\ref{bicsint.105.6}) in (\ref{bicsint.106}), we get
\begin{equation}\label{bicsint.111}
\| u-u^{(N)}\|_{{\cal H}_\mathrm{sbd}}\le {\cal O}(1)\varepsilon
^{-1}(\sup I) \varepsilon ^{\frac{1}{2}}\left( \frac{\varepsilon
      ^{\frac{1}{2}}h}{\hat{\epsilon }_\vartheta ^2} \right)^{2N+1}.
\end{equation}
This estimate is the main result of the present work. Let us recollect
the assumptions and the general context in the following theorem
(same as Theorem \ref{bics5} below).
\begin{theo}\label{int6}
Let $V_t=V(t,x)\in C_b^\infty (I\times {\bf R}^n;{\bf R})$, where $n=1$, $0<E_-<E_-'<E_+'<E_+<\infty $,
$E_0(t)\in C^\infty (I;[E_-',E_+'])$, $\ddot{\mathrm{O}}(t)\Subset {\bf R}^n$, $U_0(t)\subset \ddot{\mathrm{O}}(t)$ be as in the discussion around and
including {\rm (\ref{bicsint.1})}--{\rm (\ref{bicsint.8})}, {\rm (\ref{bicsint.9.5})}. Let $\mu_0(t)$ be a Dirichlet eigenvalue of $-h^2\Delta +V(t,\cdot )$ on
$M_0$ as in {\rm (\ref{bicsint.9})}--{\rm (\ref{bicsint.12})}. The operator $P(t)$ has a unique resonance $\lambda _0(t)$ in the set $\Omega (t)$ in
{\rm (\ref{bics.26})}: $$  \Omega (t):=\{ z\in D(\mu _0(t;h),\delta (h)/2);\, \Im z\ge -\epsilon_\vartheta/C  \}.$$
It is simple and satisfies {\rm (\ref{bics.27})}:
$$
\lambda _0(t)-\mu _0(t)=\widetilde{{\cal O}}(e^{-2S_t/h}),\
S_t:=d_t(U_0(t),\partial \ddot{\mathrm{O}}(t)).
$$
Here $\epsilon _\vartheta
  =(\epsilon /h)^\vartheta \epsilon  $, for some $\epsilon \in
  [e^{-1/(Ch)},h/C]$ for some sufficiently large constant $C>0$ and
  $\vartheta >0$ is a fixed small constant. Assume {\rm (\ref{bicsny.54})},
  {\rm (\ref{rest.65.5})}:
\begin{equation}\label{bicsint.111.5}
e^{-1/(C_1h)}\le \epsilon \le \min (h/C_2,\delta ),\ C_1,\, C_2\gg 1,
\end{equation}
as well as {\rm (\ref{bics.92.5})}: $\epsilon \le {\cal O}(\varepsilon ^{1/N_0})$ for some $N_0>0$. 

\par Define the spaces ${\cal H}_\mathrm{sbd} = H(\Lambda _{\upsilon \epsilon G})$,
${\cal D}_\mathrm{sbd}=H(\Lambda _{\upsilon \epsilon G},\langle \xi \rangle^2)$ as earlier in this section, so that $\lambda
  _0(t)$ is the unique eigenvalue in $\Omega (t)$ of $P(t):\, {\cal H}_\mathrm{sbd}\to {\cal H}_\mathrm{sbd}$ with domain ${\cal D}_\mathrm{sbd}$.

  \par Then we have the formal asymptotic series
  $\nu (t,\varepsilon )$, $\lambda (t,\varepsilon )$ in Theorem
  {\rm \ref{int5}}, where we choose $\nu _0(t)=e^0(t)$. For $N\ge 1$, define
  the partial sums $\nu ^{(N)}$, $\lambda ^{(N)}$ as in
  {\rm (\ref{bicsint.84})}, {\rm (\ref{bicsint.85})}. Let $\varepsilon $ be small enough
  so that {\rm (\ref{bicsint.88})} holds (and notice that this would follow
  from {\rm (\ref{bicsint.89})}, {\rm (\ref{bicsint.90})}). Assume (to fix the ideas)
  that $0\in I$, and assume {\rm (\ref{bicsint.105.4})} so that $\sup I\le \varepsilon ^{-N_0}$ for some constant
  $N_0>0$ and put $I_+=I\cap [0,+\infty [$. Let
  $u(t)\in C^1(I_+;{\cal H}_\mathrm{sbd})\cap C^0(I_+;{\cal
    D}_\mathrm{sbd})$ be the solution of
\begin{equation}\label{bicsint.112}
(\varepsilon D_t+P(t))u=0 \hbox{ on }I_+,\ u(0)=u^{(N)}(0),
\end{equation}
where $u^{(N)}$ is defined in {\rm (\ref{bicsint.104})}. Then {\rm (\ref{bicsint.111})}
holds uniformly for $t\in I_+$.
\end{theo}

\par {\bf Acknowledgments}. Part of this projet was conducted when MH visited Universit\'e de Bourgogne in June-July of 2017. He is very grateful to its
Institut de Math\'ematiques for the generous hospitality. MH and JS would also like to thank Galina Perelman for a very helpful discussion.
JS acknowledges support from PRC CNRS/RFBR 2017-2019 No. 1556 "Multi-dimensional semi-classical problems of Condensed Matter Physics and Quantum Dynamics''.

\section{Formal adiabatic solutions for an isolated eigenvalue}\label{1eig}
\setcounter{equation}{0}

Let ${\cal H}$ be a separable complex Hilbert space, let $I\subset
{\bf R}$ be a compact interval and let $P=P(t):{\cal
  H}\to {\cal H}$ be a closed densely defined operator, depending on
$t\in I$ such that
\begin{itemize}
\item[(H1)] The domain ${\cal D}={\cal D}(t)$ is independent of $t\in
  I$ and the domain norms $\| u \|_{{\cal D}(t)}$ are uniformly
  equivalent to each other in the sense that there exists a constant
  $C\ge 1$ such that $C^{-1}\| u \|_{{\cal D}(s)}\le \| u \|_{{\cal
      D}(t)}\le C \|u\|_{{\cal D}(s)}$, $s,t\in I$, $u\in {\cal D}$
\item[(H2)] $P(t)\in C^\infty (I;{\cal L}({\cal D},{\cal H}))$ in the
  natural sense: the successive derivatives are uniformly bounded
  ${\cal D}\to {\cal H}$ and again differentialble in the uniform
  sense.
\end{itemize}
In this section we shall also assume
\begin{itemize}
\item[(H3)] $P(t)$ has a simple eigenvalue $\lambda _0(t)$ which
  depends continuously on $t$ and is isolated from the rest of the
  spectrum:
  $$
  \mathrm{dist\,}\left(\lambda_0 (t),\sigma (P(t))\setminus \{\lambda_0 (t) \}\right)\ge 1/C,
  $$
  where $C>0$ is independent of $t$.
\end{itemize}
It follows from these assumptions that $\lambda_0(t)$ is a smooth
function of $t$. In the following result we review the formal
adiabatic construction.
\begin{prop}\label{1eig1}
There exist two asymptotic series
\begin{equation}\label{1eig.1}
\nu (t,\varepsilon )\sim \nu _0(t)+\varepsilon \nu _1(t)+\varepsilon ^2\nu
_2(t)+...\hbox{ in }C^\infty (I;{\cal D}),
\end{equation}
\begin{equation}\label{1eig.2}
\lambda (t,\varepsilon )=\lambda _0(t)+\varepsilon \lambda _1(t)+...\hbox{
  in }C^\infty (I),
\end{equation}
where $\nu _0(t)$ is non-vanishing, such that
\begin{equation}\label{1eig.3}
(\varepsilon D_t+P(t)-\lambda (t,\varepsilon ))\nu (t,\varepsilon )\sim 0,
\end{equation}
as an asymptotic series in $C^\infty (I;{\cal H})$.
\end{prop}
\begin{proof}
We insert the developments for $\nu $ and $\lambda $ into
(\ref{1eig.3}) and try to cancel the successive powers of $\varepsilon $:
\begin{equation}\label{1eig.4}
\begin{split}
&\hskip -1cm (\varepsilon D_t+P(t)-\lambda (t,\varepsilon ))\nu (t,\varepsilon )=\\
&(P(t)-\lambda _0(t))\nu _0(t)+\varepsilon \left( (D_t-\lambda _1(t))\nu _0(t)+(P-\lambda _0)\nu _1(t)
\right)\\
&+ ...\\
&+\varepsilon ^k\left( (P(t)-\lambda _0(t))\nu _k+(D_t-\lambda _1)\nu
  _{k-1}-\lambda _2\nu _{k-2}-...-\lambda _k\nu _0 \right)\\
&+...
\end{split}
\end{equation}
In order to annihilate the $\varepsilon ^0$ term it is necessary and
sufficient to let $\nu _0(t)$ be a non-vanishing eigenvector
associated to $\lambda _0(t)$, which depends smoothly on $t$ and we
fix such a choice.

\par By assumption, $\nu _0(t)$ is unique up to a smooth non-vanishing
scalar factor. The corresponding spectral projection (independent of
the scalar factor) is
\begin{equation}\label{1eig.4.5}
\pi _0(t)=\frac{1}{2\pi i}\int _{\gamma (t)}(z-P(t))^{-1}dz,
\end{equation}
where $\gamma (t)$ is the oriented boundary of the disc $D(\lambda
_0(t),r)$ and $r>0$ is a sufficiently small constant. We know that
$\pi _0(t)$ is a projection of rank 1 and hence of the form
\begin{equation}\label{1eig.5}
\pi _0(t)u=(u|\delta _0(t))\nu _0(t).
\end{equation}
This projection and its adjoint, $\pi _0(t)^*v=(v|\nu _0(t))\delta
_0(t)$ depend smoothly on $t$ and we deduce that $\delta _0(t)$ (like
$\nu _0(t)$) depends smoothly on $t$. Since $\pi _0(t)(P(t)-\lambda
_0(t))=0$, we have $((P(t)-\lambda _0(t))u|\delta _0(t))=0$ for all $u\in
{\cal D}$ and it follows that $\delta _0(t)\in {\cal D}(P_0(t)^*)$ and
that $(P(t)^*-\overline{\lambda }_0(t))\delta _0(t)=0$.
From $\pi _0(t)^2=\pi _0(t)$ it follows that $(\nu _0(t)|\delta
_0(t))=1$.

\par By holomorphic functional calculus, we know that
\begin{equation}\label{1eig.6}
{\cal H}={\cal R}(1-\pi _0(t))\oplus {\cal R}(\pi _0(t))={\cal N}(\pi
_0(t))\oplus {\cal R}(\pi _0(t))=\delta _0(t)^\perp \oplus {\bf C}\nu
_0(t).
\end{equation}
Further, $P(t):\, \delta _0(t)^\perp \to \delta _0(t)^\perp  $ is a
closed densely defined operator with spectrum equal to $\sigma
(P(t))\setminus \{ \lambda _0(t) \}$. From this we conclude that the
equation
\begin{equation}\label{1eig.7}
(P-\lambda _0(t))u=v
\end{equation}
has a solution precisely when $v\perp \delta _0(t)$ and when this
condition is fulfilled the general solution is of the form $u=\widetilde{u}+z\nu
_0(t)$, where $\widetilde{u}$ is the unique solution in $\delta
_0(t)^\perp$ and $z\in {\bf C}$ is arbitrary.

In order to annihilate the $\varepsilon ^1$ term in (\ref{1eig.4}) we need to find
$\nu _1(t)\in {\cal D}$ depending smoothly on $t$ such that
\begin{equation}\label{1eig.8}
(P(t)-\lambda _0(t))\nu _1(t)=\lambda _1(t)\nu _0(t)-D_t\nu _0(t).
\end{equation}
As we have just seen, this equation can be solved precisely when
\begin{equation}\label{1eig.9}
0=(\lambda _1(t)\nu _0(t)-D _t\nu _0(t)|\delta _0(t))=\lambda
_1(t)-(D_t\nu _0|\nu _0),
\end{equation}
so we choose $\lambda _1=(D_t\nu _0|\nu _0)$.
Choose a smooth solution $\nu _1(t)$ of (\ref{1eig.8}) (which is unique
up to a term $z(t)\nu _0(t)$ where $z(t)$ is a smooth scalar
function). Then, to annihilate the $\varepsilon ^2$ term, we have the
equation,
$$(P(t)-\lambda _0(t))\nu _2(t)+(D_t-\lambda _1(t))\nu _1(t)-\lambda
_2\nu _0=0,$$
and we see that the solvabilty with respect to $\nu _2(t)$ imposes a
unique choice of $\lambda _2$. By iterating this argument we get the proposition.
\end{proof}
\begin{remark}\label{1eig2}\par (a) Let $\lambda ,\nu $ be as in the
  proposition and let
$$
\theta (t,\varepsilon )\sim \theta _0(t)+\varepsilon \theta _1(t)+...\hbox{
  in }C^\infty (I).
$$
Then,
$$
\widetilde{\lambda }=\lambda +\varepsilon \partial _t\theta ,\
\widetilde{\nu }=e^{i\theta }\nu
$$
is another pair as in the Proposition. Indeed,
$$
0\sim e^{i\theta }(\varepsilon D_t+P-\lambda )e^{-i\theta }\widetilde{\nu
}=(\varepsilon D_t+P-\widetilde{\lambda })\widetilde{\nu }.
$$
Now, any function $\widetilde{\lambda }\sim \lambda
_0(t)+\varepsilon \widetilde{\lambda }_1(t)+...$ in $C^\infty (I)$ is of the form
$\widetilde{\lambda }=\lambda +\varepsilon \partial _t\theta $ for a
suitable $\theta $ as above (which is unique up to a constant
$C(\varepsilon )\sim C_0+\varepsilon C_1+...$) and we conclude that $\lambda
(t,\varepsilon )$ in the proposition can be any asymptotic series as in
{\rm (\ref{1eig.2})}, with leading term $\lambda _0$ given in {\rm (H3)}.
\par\noindent (b) If $(\lambda ,\nu )$ and $(\lambda ,\widetilde{\nu
})$ are two pairs as in the proposition, then there exists $C(\varepsilon
)\sim C_0+\varepsilon C_1+...$, such that $\widetilde{\nu }=C(\varepsilon
)\nu $. In fact, writing $\widetilde{\nu }_0=C_0(t)\nu
_0$ and the corresponding equation for $\widetilde{\nu }_1$;
$$
(P-\lambda _0)\widetilde{\nu }_1=C_0(\lambda _1-D_t)\nu
_0-(D_tC_0)\nu _0,$$
we see that $D_tC_0$ has to vanish, so that $C_0$ is constant. Repeat
the argument for
$$
(\varepsilon D_t+P-\lambda )\left(\frac{\widetilde{\nu }-C_0\nu
  }{\varepsilon } \right)=0,
$$
to see that $\widetilde{\nu }=C_0\nu +\varepsilon C_1\nu _0+{\cal
  O}(\varepsilon ^2)$, where $C_1$ is a constant. By iteration we get the
statement.
\end{remark}
\section{Some further adiabatic results }\label{fad}
\setcounter{equation}{0}
\subsection{Formal adiabatic solutions for an isolated group of eigenvalues}\label{greig}
\setcounter{equation}{0}

Let $P(t)$ satisfy the assumptions of Section \ref{1eig} except for
the assumption (H3) that we generalize to
\begin{itemize}
\item[(H4)] For some integer $N_0\ge 1$, $P(t)$ has a  group of $N_0$
  eigenvalues $\lambda _1(t),...$, $\lambda_{N_0}(t)$ counted with their
  multiplicities, depending continuously on $t$ and isolated from the
  rest of the spectrum:
$$
\mathrm{dist\,}\left( \{\lambda_\cdot (t)\},\sigma (P(t))\setminus \{\lambda _\cdot (t) \}\right)\ge 1/C,
$$
where $C>0$ is independent of $t$. Here $\{ \lambda _\cdot (t)\} =\{\lambda _j(t);\, 1\le j\le N_0 \}$.
\end{itemize}
This assumption can be reformulated as follows:
\begin{itemize}
\item There exists a simple
closed $C^1$ loop, $\gamma =\gamma _t:\, S^1\to {\bf C}$, enclosing
some non-empty part of the spectrum, such that
$\mathrm{dist\,}(\gamma _t(S^1),\cap \sigma (P(t)))\ge 1/C $ and such that the
spectral projection
\begin{equation}\label{greig.1}
\pi _0(t)=\frac{1}{2\pi i}\int_{\gamma _t}(z-P(t))^{-1} dz
\end{equation}
has finite rank, necessarily equal to some constant $N_0\ge 1$.
\end{itemize}

\par The formal adiabatic problem is now to construct
\begin{equation}\label{greig.2}
U(t,\varepsilon )\sim U_0(t)+\varepsilon U_1(t)+...\in C^\infty (I;{\cal
  L}({\bf C}^{N_0},{\cal D}))
\end{equation}
and
\begin{equation}\label{greig.3}
\Lambda (t,\varepsilon )\sim \Lambda _0(t)+\varepsilon \Lambda _1(t)+...\in C^\infty (I;{\cal
  L}({\bf C}^{N_0},{\bf C}^{N_0})),
\end{equation}
such that
\begin{equation}\label{greig.4}
\left(\varepsilon D_t+P(t) \right) U(t,\varepsilon )-U(t,\varepsilon )\Lambda
(t,\varepsilon )=0
\end{equation}
as a formal powerseries with values in ${\cal L}({\bf C}^{N_0},{\cal
  H})$,
\begin{equation}\label{greig.5}
\sigma (\Lambda _0(t))=\{ \lambda _\cdot (t) \},
\end{equation}
\begin{equation}\label{greig.6}
U_0(t)\hbox{ is injective and }U_0(t)=\pi _0(t)U_0(t).
\end{equation}

\par As in the case of a single eigenvalue ($N_0=1$) we substitute
(\ref{greig.2}), (\ref{greig.3}) into (\ref{greig.4}) and try to
annihilate the successive powers of $\varepsilon $. This leads to the
equations,
\begin{equation}\label{greig.7}
PU_0-U_0\Lambda _0=0,
\end{equation}
\begin{equation}\label{greig.8}
PU_1-U_1\Lambda _0+D_tU_0-U_0\Lambda _1=0,
\end{equation}
...
\begin{equation}\label{greig.9}
PU_k-U_k\Lambda _0+D_tU_{k-1}-U_{k-1}\Lambda _1-U_{k-2}\Lambda
_2...-U_0\Lambda _k=0,
\end{equation}
...

\par As for (\ref{greig.7}), we let
\begin{equation}\label{greig.10}
U_0(t):{\bf C}^{N_0}\to {\cal R}(\pi _0(t))
\end{equation}
be any injective map which depends smothly on $t$ and then take
$\Lambda _0(t)=U_0(t)^{-1}P(t)U_0(t)$, where $U_0^{-1} $ denotes the
inverse of (\ref{greig.10}). Then
$$
\sigma (\Lambda _0 (t))=\sigma \left( {{P(t)}_\vert}_{{\cal R}(\pi _0(t))} \right),
$$
so (\ref{greig.5}) is fulfilled.

\par In order to solve (\ref{greig.8}) we first choose $\Lambda _1(t)$
so that
\begin{equation}\label{greig.11}
\pi _0\left(D_tU_0-U_0\Lambda _1 \right) =0,
\end{equation}
i.e.\ $\Lambda _1=U_0^{-1}\pi _0 D_tU_0$, where $U_0^{-1}$ denotes the
inverse of the operator in (\ref{greig.10}). Then we look for $U_1$
with
\begin{equation}\label{greig.12}
(1-\pi _0)U_1=U_1,
\end{equation}
i.e.\ $U_1:{\bf C}^{N_0}\to {\cal R}(1-\pi _0)\cap {\cal D}$, and it
suffices to find such a (smooth family of) map(s) such that
\begin{equation}\label{greig.13}
PU_1-U_1\Lambda _0+(1-\pi _0)D_tU_0 =0.
\end{equation}
Here $P$ is identified with ${{P}_\vert}_{{\cal R}(1-\pi _0)}:{\cal
  R}(1-\pi _0)\to {\cal R}(1-\pi _0)$, which is closed, densely
defined and satisfies
\begin{equation}\label{greig.14}
\sigma \left( {{P}_\vert}_{{\cal R}(1-\pi _0)} \right)\cap \sigma
(\Lambda _0)=\emptyset .
\end{equation}
\begin{lemma}\label{greig1}
  Let ${\cal H}$ be a complex separable Hilbert space and let $A:{\cal
    H}\to {\cal H}$ be closed and densely defined. Let $B\in {\cal
    L}({\bf C}^{N_0},{\bf C}^{N_0})$ and assume that
\begin{equation}\label{greig.15}
\sigma (A)\cap \sigma (B)=\emptyset .
\end{equation}
Then for every $V\in {\cal L}({\bf C}^{N_0},{\cal H})$, there is a
unique $U\in {\cal L}({\bf C}^{N_0},{\cal D}(A))$ such that
\begin{equation}\label{greig.16}
AU-UB=V.
\end{equation}
\end{lemma}
\begin{proof}
Decompose
$$
{\bf C}^{N_0}=\bigoplus_{\lambda \in \sigma (B)}E_\lambda ,
$$
where $E_\lambda $ is the spectral subspace, so that $B:E_\lambda \to
E_\lambda $ and ${{B}_\vert}_{E_\lambda }=\lambda +N$, where
$N=N_\lambda $ is nilpotent. It suffices to find, for every $\lambda
\in \sigma (B)$, a unique linear
operator $U=U_\lambda :E_\lambda \to {\cal D}(A)$ such that
$$
AU-U(\lambda +N)=V_\lambda ,\hbox{ where }V_\lambda
={{V}_\vert}_{E_\lambda }.
$$
We write this as
\begin{equation}\label{greig.17}
(A-\lambda )U-UN=V_\lambda
\end{equation}
and notice that when $N=0$, the unique solution is $U=(A-\lambda
)^{-1}V_\lambda $.

\par In the general case we look for $U$ of the form $U=(A-\lambda
)^{-1}\widetilde{U}$, $\widetilde{U}\in {\cal L}(E_\lambda ,{\cal H})$
and (\ref{greig.17}) becomes
\begin{equation}\label{greig.18}
\widetilde{U}-\widetilde{N}(\widetilde{U})=V_\lambda ,
\end{equation}
where
$$
\widetilde{N}(\widetilde{U}):=(A-\lambda )^{-1}\widetilde{U}N.
$$
It then suffices to observe that $\widetilde{N}$ is nilpotent, so that
(\ref{greig.18}) has a unique solution.
\end{proof}

By a simple Neumann series argument, if $A=A(t)$, $B=B(t)$, $V=V(t)$ depend
smoothly on a real parameter and ${\cal D}(A(t))$ is independent of
$t$, then the same holds for $U(t)$.

\par Applying Lemma \ref{greig1} and the above observation to
(\ref{greig.13}), we get a unique solution $U_1(t):{\bf C}^{N_0}\to
{\cal D}\cap {\cal R}(1-\pi _0)$ which is smooth in $t$. Thus, there
is a unique solution $(U_1,\Lambda _1)$ to
(\ref{greig.11})--(\ref{greig.13}).

\par Assuming that $U_0,...,U_{k-1}$, $\Lambda _1,...,\Lambda _{k-1}$
have been constructed, we solve (\ref{greig.9}) in the same way:
First, make the unique choice of $\Lambda _k$ for which
\begin{equation}\label{greig.19}
\pi _0(D_tU_{k-1}-U_{k-1}\Lambda _1-U_{k-2}\Lambda _2...-U_0\Lambda _k)=0.
\end{equation}
Then, let $U_k$ be the unique map: ${\bf C}^{N_0}\to {\cal R}(1-\pi
_0)\cap {\cal D}$, such that
$$
PU_k-U_k\Lambda _0+(1-\pi _0)\left(D_tU_{k-1}-\Lambda
  _1U_{k-1}...-\Lambda _kU_0 \right)=0.
$$

\par Summing up the discussion, we have
\begin{prop}\label{greig2}
  The problem {\rm (\ref{greig.2})}--{\rm (\ref{greig.6})} has a solution with
$$
U_k\in C^\infty (I;{\cal L}({\bf C}^{N_0},{\cal D})), \quad \Lambda_k\in C^\infty (I;{\cal L}({\bf C}^{N_0},{\bf C}^{N_0})).
$$
The solution is unique if we first choose $U_0,\Lambda _0$ as in
  {\rm (\ref{greig.7})} and then require that $(1-\pi _0(t))U_k(t)=U_k(t)$
  for $k\ge 1$.
\end{prop}

\subsection{Adiabatic projections}\label{apro}

We keep the assumption of the preceding section. Recall the notion of
adiabatic spectral projections, \cite{Ne81, Ne93} in the presentation
of \cite{Sj93}. Consider
\begin{equation}\label{apro.1}
{\cal P}(t,\varepsilon D_t;\varepsilon )=\varepsilon D_t+P(t)
\end{equation}
as a vector valued $\varepsilon $-pseudodifferential operator (see
e.g.\ \cite{DiSj99}) with symbol
\begin{equation}\label{apro.2}
{\cal P}(t,\tau ;\varepsilon )=\tau +P(t).
\end{equation}
Then for $z\in \mathrm{neigh\,}(\gamma _t)$, where $\gamma _t$ is as
in (\ref{greig.1}), we define the formal resolvent
\begin{equation}\label{apro.3}
(z-{\cal P})^{-1}=S(t,\varepsilon D_t,z;\varepsilon ),
\end{equation}
as a formal $\varepsilon $-pseudodifferential operator with symbol
{\begin{equation}\label{apro.4}
S(t,\tau ,z;\varepsilon  )\sim S_0(t,\tau ,z )+\varepsilon
S_1(t,\tau ,z)+...,
\end{equation}
defined for $t\in I$, $z-\tau \in \mathrm{neigh\,}(\gamma _t)$ and
obtained by the standard elliptic parametrix construction, so that
$S_0(t,\tau ,z)=(z-\tau -P(t))^{-1}$. The corresponding adiabatic
projection is the formal $\varepsilon $-pseudodifferential operator,
\begin{equation}\label{apro.5}
\pi (t,\varepsilon D_t;\varepsilon )=\frac{1}{2\pi i}\int_{\gamma
  _t}(z-{\cal P})^{-1}dz,
\end{equation}
defined on the symbol level, for $\tau \in \mathrm{neigh\,}(0, {\bf C})$.

\par Using the property,
$$
S(t,\tau , z ;\varepsilon )=S(t,\tau -z,0;\varepsilon )
$$
it is shown in \cite{Sj93} that the symbol $\pi (t,\tau ;\varepsilon )$
is independent of $\tau $, so that
\begin{equation}\label{apro.8}
\pi (t,\varepsilon D_t;\varepsilon )u=\pi (t,\varepsilon )u(t),
\end{equation}
where
\begin{equation}\label{apro.9}
\pi (t;\varepsilon )=\pi _0(t)+\varepsilon \pi _1(t)+...\in C^\infty
(I;{\cal L } ({\cal H},{\cal D})).
\end{equation}
$\pi _0(t)$ is the
spectral projection for $P(t)$ in (\ref{greig.1}). Moreover (cf.\
(16), (17) in \cite{Sj93}),
\begin{equation}\label{apro.10}
\pi (t;\varepsilon )^2=\pi (t;\varepsilon ),
\end{equation}
\begin{equation}\label{apro.11}
[ \varepsilon D_t+P(t),\pi (t;\varepsilon ) ]=0.
\end{equation}
\begin{prop}\label{apro1}
Let $U\sim U_0(t)+\varepsilon U_1(t)+...$, $\Lambda \sim \Lambda
_0(t)+\varepsilon \Lambda _1(t)+...$ be a solution of the problem
{\rm (\ref{greig.2})}--{\rm (\ref{greig.4})}. Put $\widetilde{U}=\pi (t,\varepsilon
)U$. Then $\widetilde{U}\sim U_0+\varepsilon \widetilde{U}_1+...$ (with
$\widetilde{U}_0=\pi _0U_0$), and $(\widetilde{U},\widetilde{\Lambda
}):=(\widetilde{U},\Lambda )$ is a solution of
{\rm (\ref{greig.2})}--{\rm (\ref{greig.4})} and we have
\begin{equation}\label{apro.12}
\pi (t,\varepsilon ) \widetilde{U}(t,\varepsilon )=\widetilde{U}(t,\varepsilon ).
\end{equation}
In view of {\rm (\ref{apro.12})} and the fact that $\pi $ is asymptotically
a projection, we shall say that ${\cal R}(\widetilde{U})\subset {\cal
  R}(\pi )$ (pointwise in $t$).
\end{prop}
\begin{proof}
  By (\ref{greig.4}), (\ref{apro.11}), we get
$$
(\varepsilon D_t+P(t))\widetilde{U}-\widetilde{U}\Lambda =\pi
((\varepsilon D_t+P(t)U-U\Lambda )=0
$$
\end{proof}

\begin{prop}\label{apro2}
Let $(U,\Lambda )$, $(\widetilde{U},\widetilde{\Lambda })$ be two
solutions of the problem {\rm (\ref{greig.2})}--{\rm (\ref{greig.4})}, with ${\cal
  R}(U), {\cal R}(\widetilde{U})\subset {\cal R}(\pi )$ in the sense
of {\rm (\ref{apro.12})}. Assume that $U$ satisfies {\rm (\ref{greig.6})}.
Then $\exists !$
  \begin{equation}\label{apro.13}M(t;\varepsilon )\sim M_0(t)+\varepsilon M_1(t)+...\in C^\infty (I;{\cal
    L}({\bf C}^{N_0},{\bf C}^{N_0}))\end{equation} such that
\begin{equation}\label{apro.14}
\widetilde{U}=UM.
\end{equation}

\par Conversely, if $U$ solves {\rm (\ref{greig.2})}-{\rm (\ref{greig.6})} and $M$
is of the form {\rm (\ref{apro.13})} with $M_0(t)$ invertible, then
$(\widetilde{U},\widetilde{\Lambda })$ solves
{\rm (\ref{greig.2})}--{\rm (\ref{greig.5})} where
\begin{equation}\label{apro.15}
\widetilde{U}:=UM,\ \widetilde{\Lambda }:=M^{-1}\varepsilon
D_t(M)+M^{-1}\Lambda M.
\end{equation}
\end{prop}
\begin{proof}
We first prove the converse part by direct calculation
\[\begin{split}
(\varepsilon D_t+P(t))\widetilde{U}&=((\varepsilon
D_t+P(t))U)M+U\varepsilon D_t(M)\\
&=U(\Lambda M+\varepsilon D_t(M))\\
&=UM(M^{-1}\Lambda M+M^{-1}\varepsilon D_t(M))\\
&=\widetilde{U}\widetilde{\Lambda }.
\end{split}
\]

\par We next prove the direct part. Let $(U,\Lambda )$,
$(\widetilde{U},\widetilde{\Lambda })$ be as in the beginning of the
proposition, both solving (\ref{greig.2})--(\ref{greig.4}) with $\pi
U=U$, $\pi \widetilde{U}=\widetilde{U}$ and such that $U$ satisfies (\ref{greig.6}).

\par Writing $\widetilde{U}=\widetilde{U}_0+\varepsilon
\widetilde{U}_1+...$, we conclude that $\widetilde{U}_0$ maps ${\bf
  C}^{N_0}\to {\cal R}(\pi _0)$ pointwise in $t$. $U_0(t):{\bf
  C}^{N_0}\to {\cal R}(\pi _0(t))$ has the same property and is
bijective. Hence there is a unique $M_0(t):{\bf C}^{N_0}\to {\bf
  C}^{N_0}$, smooth in $t$, such that
$\widetilde{U}_0(t)=U_0(t)M_0(t)$. From the proof of the ``converse''
part, we see that
$$
V_1(t):=\widetilde{U}(t)-U(t)M_0(t)\sim : \varepsilon
V_{1,1}(t)+\varepsilon ^2V_{2,1}(t)+...
$$
solves (\ref{greig.2})--(\ref{greig.4}) with $\Lambda $ replaced by a
new matrix $\varepsilon \Lambda _1(t;\varepsilon )$. We also have $\pi
V_1=0$, so $\pi _0(t)V_{1,1}(t)=0$ and hence $\exists$ $M_1(t):{\bf
  C}^{N_0}\to {\bf C}^{N_0}$, such that
$V_{1,1}(t)=U_0(t)M_1(t)$. Then
$$
V_2(t):=\widetilde{U}(t)-U(t)(M_0(t)+\varepsilon M_1(t))\sim :
\varepsilon ^2 V_{2,2}(t)+\varepsilon ^3V_{2,3}(t)+...
$$
and $\pi (t)V_2(t)=0$, so $\pi _0(t)V_{2,2}(t)=0$. Iterating this
procedure we get $M(t)\sim M_0(t)+\varepsilon M_1(t)+...$ with the
required properties.
\end{proof}

\subsection{Extension to the case of variable $\varepsilon$}
\label{veps}

Consider the evolution equation
\begin{equation}\label{veps.1}
(D_s+P(s))\nu (s)=0
\end{equation}
on some large interval $I$, where $P(s)$ are closed densely defined
operators with common domain ${\cal D}$. Assume that
\begin{equation}\label{veps.2}
\partial _s^kP(s)={\cal O}(\varepsilon (s)^k),\ k=0,1,2,...,
\end{equation}
as bounded operators from ${\cal D}$ to ${\cal H}$. Here the function
$\varepsilon (s)$ is assumed to satisfy
\begin{equation}\label{veps.3}
\varepsilon >0,\ \ \partial _s^k\varepsilon ={\cal O}(\varepsilon ^{k+1}),\
k\ge 0.
\end{equation}
\begin{remark}\label{veps1}
  Under the same assumptions, if $\widetilde{\varepsilon }\ge \varepsilon $
  is a second function on $I$ which satisfies (\ref{veps.3}), then
  (\ref{veps.2}) holds with $\varepsilon $ replaced by
  $\widetilde{\varepsilon }$.
\end{remark}

Let $f(s)$ be the strictly increasing function, uniquely determined up to a
constant, by
\begin{equation}\label{veps.4}
f'(s)=\varepsilon (s).
\end{equation}
Then, if $t=f(s)$, we have $D_s=f'(s)D_t=\varepsilon (s)D_t$ and
(\ref{veps.1}) takes the form,
\begin{equation}\label{veps.5}
\left( \varepsilon (g(t))D_t+P(g(t)) \right) u(t)=0,\ \ u(t)=\nu
(g(t)).
\end{equation}
Here, $g:=f^{-1}$.

\par Differentiating $f(g(t))=t$, we first get $f'(g(t))g'(t)=1$, so
\begin{equation}\label{veps.6}
g'(t)=\frac{1}{\varepsilon (g(t))}.
\end{equation}
Differentiating $m$ times, where $m\ge 2$, we get
$$
f'(g(t))\partial _t^mg(t)+\sum_{k=2}^m \sum_{m_j\ge 1,\atop
  m_1+...+m_k=m}C_{m_1,...,m_k} f^{(k)}(g(t))\partial _t^{m_1}g(t)\cdot
...\cdot \partial _t^{m_k}g(t)=0.
$$
Assuming by induction, that
$$
\partial _t^{\widetilde{m}}g={\cal O}(1/\varepsilon (g(t))),\ \widetilde{m}<m,
$$
we get
$$
\varepsilon (g(t))\partial _t^mg+\sum_{k=2}^m\sum_{m_1+..+m_k=m}{\cal
  O}(1)\varepsilon (g(t))^{k}\varepsilon (g(t))^{-k}=0,
$$
so we get
\begin{equation}\label{veps.7}
\partial _t^m g={\cal O}(1/\varepsilon (g(t))),\ m\ge 1.
\end{equation}
Now,
\begin{equation}\label{veps.8}
\partial _t^m(P(g(t)))=\sum_{k=1}^m
C_{m_1,..,m_k}P^{(k)}(g(t))\partial _t^{m_1}g...\partial
_t^{m_k}g=\sum {\cal O}(1)\varepsilon ^{k-k}={\cal O}(1).
\end{equation}
Similarly,
\begin{equation}\label{veps.9}
\partial _t^m\varepsilon (g(t))=\sum {\cal O}(1)\varepsilon ^{(k)}\partial
_t^{m_1}g...\partial _t^{m_k}g={\cal O}(\varepsilon (g(t))).
\end{equation}
This shows that (\ref{veps.5}) is a very nice semi-classical equation.

\section{General facts about operators and escape functions
}\label{opesc}
\setcounter{equation}{0}
This section is a review of some material from \cite{HeSj86, GeSj87}
and we add some remarks for later use.
We adopt the frame work of
\cite{HeSj86}: Choose two positive smooth scale functions
$r(x)$, $R(x)$ on ${\bf R}^n$ with
\begin{equation}\label{opesc.1}
r\ge 1,\ \ rR\ge 1,
\end{equation}
Let \begin{equation}\label{opesc.2}
\widetilde{r}(x,\xi )=(r(x)^2+\xi ^2)^{\frac{1}{2}}\in C^\infty ({\bf R}^{2n}).
\end{equation}
If $0<m\in C^\infty ({\bf R}^{2n})$ we say that $a\in C^\infty ({\bf
  R}^{2n})$ belongs to the space $S(m)$, if
\begin{equation}\label{opesc.3}|\partial _x^\alpha \partial _\xi ^\beta a|\le
  C_{\alpha ,\beta }m(x,\xi )R(x)^{-|\alpha
    |}\widetilde{r}(x,\xi )^{-|\beta |}.\end{equation}
We will always assume that the weight $m$ and the scale functions
belong to their own symbol classes:
\begin{equation}\label{opesc.4}
m\in S(m),\ r\in S(r),\ R\in S(R).
\end{equation}

\par
It follows that $\widetilde{r}\in S(\widetilde{r})$. The
naturally associated metric on ${\bf R}^{2n}$ in the spirit of
H\"ormander's Weyl calculus of pseudodifferential operators \cite{Ho79} is given
by
\begin{equation}\label{opesc.4.5}
g=\left(\frac{d\xi }{\widetilde{r}}\right)^2+\left(\frac{dx}{R} \right)^2.
\end{equation}
It is slowly varying, but another important assumption of that
calculus will not be satisfied in our case however, namely the
$\sigma $-temperance.

Often, $a$ and even $m$ will depend on the semi-classical parameter
$h$, and it will then be implicitly assumed that all estimates
involved in the statements $a\in S(m)$ and $m\in S(m)$ are uniform
with respect to $h$ (and possibly other parameters as
well). We define $h^kS(m):=S(h^km)$. When $a:{\bf
    R}^{2n}\to {\bf R}^{2n}$ is a smooth map and $g$ is a smooth
  metric on ${\bf R}^{2n}$, we say that $a$ is of
  class $S(m)$ for the metric $g$, if $g_{a(x,\xi )}(\partial
  _x^\alpha \partial _\xi ^\beta a(x,\xi ))$ satisfy the estimates in
  (\ref{opesc.3}) uniformly.

Let $1\le m_0(x)\in S(m_0)$ and let $P=P(x,hD;h)$ be a semi-classical
differential operator on ${\bf R}^n$ of the form
\begin{equation}\label{opesc.5}
P=\sum_{|\alpha |\le N_0}a_\alpha (x;h)(hD_x)^\alpha ,
\end{equation}
where $a_\alpha (x;h)\in S(m_0(x)r^{-|\alpha |})$ and
\begin{equation}\label{opesc.6}
a_\alpha (x;h)=
\sum_{k=0}^{N_0-|\alpha |} h^ka_{\alpha ,k}(x),\
a_{\alpha ,k}\in S(m_0r^{-|\alpha |}(rR)^{-k}).
\end{equation}
Such operators form an algebra in the natural way.
Then
$$
h^ka_{\alpha ,k}\xi ^\alpha \in S\left( m_0(\widetilde{r}/r)^{N_0}(h/(\widetilde{r}R))^k\right)
$$
and we have the full semi-classical symbol for the standard left quantization
\begin{equation}\label{opesc.7}P(x,\xi ;h)=\sum_{|\alpha |\le
    N_0}a_\alpha (x;h)\xi ^\alpha\in S(m),\hbox{ where
  }m(x,\xi )=m_0(x)(\widetilde{r}/r)^{N_0}
\end{equation} We can write
\begin{equation}\label{opesc.8}
P(x,\xi ;h)= p(x,\xi )+hp_1(x,\xi )+h^2p_2(x,\xi
)+...+h^{N_0}p_{N_0}(x,\xi ),
\end{equation}
where
\begin{equation}\label{opesc.9} p_j\in S(m(\widetilde{r}R)^{-j}),\end{equation}

\par We also assume analyticity in $x$ near infinity:
\begin{equation}\label{opesc.11}\begin{split}
&\exists C>0\hbox{ such that }P \hbox{ extends to a holomorphic
  function} \\
&\hbox{in }\{x\in {\bf C}^n;\, |\Re x|>C,\ |\Im
x|\le R(\Re x)/C \}\hbox{ and the}\\
& \hbox{symbol properties above extend in the natural sense.}
\end{split}
\end{equation}
This could be formulated more directly in terms of the coefficients
$a_{\alpha }$ in (\ref{opesc.5}).

We assume that $P$ is formally self-adjoint, so that the classical and
the semi-classical principal symbols, given respectively by
\begin{equation}\label{opesc.12}
p_\mathrm{class}(x,\xi )=\sum_{|\alpha |=N_0}a_{\alpha ,0}(x)\xi ^\alpha
\end{equation}
and
\begin{equation}\label{opesc.13}
p(x,\xi )=\sum_{|\alpha |\le N_0}a_{\alpha ,0}(x)\xi ^\alpha
\end{equation}
are real-valued. We make the ellipticity assumption
in the classical PDE sense,
\begin{equation}\label{opesc.14}
p_\mathrm{class}(x,\xi )\ge m_0(x)(|\xi |/r)^{N_0},\ (x,\xi )\in {\bf R}^{2n}.
\end{equation}
This implies for the zero energy surface,
\begin{equation}\label{opesc.15}
\Sigma _0=\{ (x,\xi )\in {\bf R}^{2n};\, p(x,\xi )=0 \},
\end{equation}
that
\begin{equation}\label{opesc.16}
|\xi |\le \mathrm{Const.\,}r(x)\hbox{ on }\Sigma _0.
\end{equation}
The same holds on $\Sigma _E:=p^{-1}(E)$ for every fixed $E$, but we
shall mainly concentrate on the case $E=0$ to simplify the notation
(observing that after replacing $p$ with $p-E$, we are reduced to that case).

\par We are particularly interested in the following situation:
\begin{equation}\label{opesc.17}
P=-h^2\Delta +V(x)-1
\end{equation}
with symbol
\begin{equation}\label{opesc.18}
P(x,\xi ;h)=p(x,\xi )=\xi ^2+V(x)-1,
\end{equation}
We will assume that $V$ is smooth, real-valued and extends
holomorphically to the set in (\ref{opesc.11}) and tends to 0 when $x\to
\infty $ in that set. This enters into the general framework with
\begin{equation}\label{opesc.19}
r=1,\ R=\langle x\rangle,\ m_0(x)=1,\ m=\langle \xi \rangle^2.
\end{equation}

\par We next discuss escape functions. If $a_j\in S(m_j)$, $j=1,2$,
then $a_1a_2\in S(m_1m_2)$ and
\begin{equation}\label{opesc.20}
H_{a_1}a_2=\{ a_1 ,a_2 \}\in S\left( \frac{m_1m_2}{\widetilde{r}R} \right).
\end{equation}
(\ref{opesc.20}) remains valid if we weaken the assumption on $a_1$,
$a_2$ to $a_j\in \dot{S}(m_j)$ for $j=1,2$ where we let $\dot{S}(m)$
denote the space of smooth functions $a$ on ${\bf R}^{2n}$ which
satisfy the estimates (\ref{opesc.3}) for all {\it non-vanishing}
$(\alpha ,\beta) \in {\bf N}^{2n}$.
\begin{dref}\label{opesc1}
A real-valued function $G\in \dot{S}(\widetilde{r}R)$ is called an
escape function if there exists a constant $C_0$ and a compact set
$K\subset {\bf R}^{2n}$ such that
\begin{equation}\label{opesc.21}
H_pG(\rho )\ge \frac{m(\rho )}{C_0},\hbox{ for all }\rho \in \Sigma
_0\setminus K.
\end{equation}
\end{dref}
When specifying the energy level we say that $G$ in the definition
above is an escape function at energy 0. More generally we can define
escape functions for $p$ at a real energy $E$, replacing $\Sigma _0$
by $\Sigma _E$.

As we have already noticed, $|\xi |\le r(x)$ on $\Sigma _0$, so
$m\asymp m_0$ there. Also, $H_pG\in S(m)$ for all $G\in
\dot{S}(\widetilde{r}R)$, so when $G$ is an escape function, we have
$H_pG\asymp m$ on $\Sigma _0$ near infinity.

\par We will also need to know that $|p|$ cannot be very small away
from $\Sigma _0$ and therefore make the following assumption:
\begin{equation}\label{opesc.22}\begin{split}
&\hbox{For every }r_0>0,\hbox{ there exists }\epsilon _0>0,\hbox{ such
  that}\\
&|p|\ge \epsilon _0m\hbox{ on }{\bf R}^{2n}\setminus \bigcup_{\rho \in
  \Sigma _0}B_{g(\rho )}(\rho ,r_0).
\end{split}\end{equation}
Here, $g$ is the metric in (\ref{opesc.4.5}) and $B_{g(\rho )}(\rho
,r_0)$ denotes the open ball with center $\rho $ and radius $r_0$ for
the constant metric $g(\rho )$.

\par For $\epsilon _0>0$ sufficiently small, we introduce the energy
shell
\begin{equation}\label{opesc.23}
\Sigma _{[-\epsilon _0,\epsilon_0]}=\{\rho \in {\bf R}^{2n};\, |p(\rho )|\le
\epsilon _0 \}.
\end{equation}
The assumption (\ref{opesc.22}) implies that for every $r_0>0$, there
exists $\epsilon _0>0$ such that
\begin{equation}\label{opesc.24}
\Sigma _{[-\epsilon _0,\epsilon_0]}\subset \bigcup_{\rho \in \Sigma _0}B_{g(\rho
  )}(\rho ,r_0).
\end{equation}
From (\ref{opesc.21}), (\ref{opesc.24}) and the fact that $H_pG\in
S(m)$ it follows that there exist $C_0,\epsilon _0>0$ and a compact
set $K\subset {\bf R}^{2n}$ such that
\begin{equation}\label{opesc.25}
H_pG(\rho )\ge \frac{m(\rho )}{C_0},\hbox{ for all }\rho \in \Sigma
_{[-\epsilon _0,\epsilon_0]}\setminus K.
\end{equation}

For the Hamilton field,
$$
H_G=\partial _\xi G(x,\xi )\cdot \partial _x-\partial _xG(x,\xi
)\cdot \partial _\xi ,
$$
we get when $G\in \dot{S}(\widetilde{r}R)$,
\begin{equation}\label{opesc.26}\| H_G\|_g\asymp \frac{|\partial _\xi
    G|}{R}+\frac{|\partial _x G|}{\widetilde{r}}={\cal O}(1).\end{equation}
Using that $p\in S(m)$, we also have
$$
H_pG=\widetilde{r} p'_\xi\cdot
\frac{G'_x}{\widetilde{r}}-R p'_x\cdot \frac{G'_\xi }{R}={\cal O}(m)\|H_G\|_g.
$$
Thus, if $G$ is an escape function we get with $\epsilon _0$, $K$ as
in (\ref{opesc.25}),
\begin{equation}\label{opesc.27}
\| H_G\|_g\asymp 1\hbox{ in }\Sigma _{[-\epsilon _0,\epsilon_0]}\setminus K
\end{equation}
and here, $H_pG\ge (m/C)\| H_G\|_g$. Also notice that
\begin{equation}\label{opesc.27.5}
\| H_p\|_g={\cal O}\left(\frac{m}{\widetilde{r}R} \right).
\end{equation}
Until further notice we restrict the attention to
  $\Sigma _{[-\epsilon _0,\epsilon_0]}$.
We next review the appendix in \cite{GeSj87} and especially how to
improve the escape function $G$ by modifying it on a bounded set. Let
$$
]-\tau _-(\rho ),\tau _+(\rho )[\ni t\mapsto \exp tH_p(\rho )
$$
be the maximal $H_p$-integral curve through the point $\rho \in \Sigma
_{[-\epsilon _0,\epsilon_0]}$, where $0<\tau _{\pm}(\rho )\le +\infty $ are lower
semi-continuous.

\par If $K\subset \Sigma _{[-\epsilon _0,\epsilon_0]}$ is a compact subset as in
(\ref{opesc.25}), then there exists a finite number $T=T(K)>0$ such
that
\begin{equation}\label{opesc.28}
-T(K)<G<T(K)\hbox{ on }K.
\end{equation}
The set $\{ \rho \in \Sigma _{[-\epsilon _0,\epsilon_0]};\, G(\rho )\ge T(K) \}$ is
invariant under the $H_p$-flow in the positive time direction:
$$G(\rho )\ge T(K)\Longrightarrow G(\exp tH_p(\rho ))\ge T(K), 0\le
t<\tau _+(\rho ).$$
If $G(\rho )\ge T(K)$, $\epsilon _0>0$ we have
$\exp tH_p(\rho )\in B_{g(\rho )}(\rho ,\epsilon _0)$ for
$0\le t\le t_0(\epsilon _0)\widetilde{r}R/m$, when
$t_0(\epsilon _0)>0$ is sufficiently small and $G(\exp tH_p(\rho ))$
will increase by $\asymp \widetilde{r}R(\rho )\ge 1$ during such a
time interval. Then repeat the same consideration with $\rho $
replaced by $\exp t_0(\epsilon _0)H_p(\rho )$ and so on. The
trajectory will have to go through infinitely many balls as above and we
conclude that $G(\exp tH_p(\rho ))\to +\infty $ when $\tau \to \tau
_+(\rho )$, for every
$\rho \in \Sigma _{[-\epsilon _0,\epsilon_0]}\cap G^{-1}([T(K),+\infty
[)$.
Similarly, $G(\exp tH_p(\rho ))\to -\infty $ when
$0\ge t\to -\tau _-(\rho )$ for every
$\rho \in G^{-1}(]-\infty ,-T(K)])$.

\par By a similar argument, if $K_1\subset \Sigma _{[-\epsilon _0,\epsilon_0]}$ is a
sufficiently large compact set containing $K$, then for every $\rho
\in G^{-1}([-T(K),T(K)])\setminus K_1$, we have $\exp tH_p(\rho
)\not\in K$, $t\in {\bf R}$, and $G(\exp tH_p(\rho ))\to \pm \infty $
when $t\to \pm \tau _{\pm}(\rho )$.

\par Define the outgoing and incoming tails $\Gamma _+,\, \Gamma
_-\subset \Sigma _{[-\epsilon _0,\epsilon_0]}$ respectively, by
\begin{equation}\label{opesc.29}
\Gamma _\pm =\{ \rho \in \Sigma _{[-\epsilon _0,\epsilon_0]};\, \exp tH_p(\rho
)\not\to \infty ,\ t\to \mp\tau _\mp (\rho ) \} .
\end{equation}

In \cite{GeSj87} it was shown that $\Gamma _\pm $ are closed sets,
$$
\Gamma _+\subset G^{-1}(]-T(K),+\infty [),\ \Gamma _-\subset
G^{-1}(]-\infty ,T(K)[),
$$
and that
$$\Gamma _+\cap G^{-1}(]-\infty ,T]),\ \Gamma _-\cap
G^{-1}([-T,+\infty [)$$
are compact for every $T\in {\bf R}$. In particular {\it the trapped
  set} $\Gamma _+\cap \Gamma _-$ is a compact subset of
$G^{-1}(]-T(K),T(K)[)$ and
\begin{equation}\label{opesc.30}
\Gamma _\pm =\{ \rho \in \Sigma _{[-\epsilon _0,\epsilon_0]};\, \exp tH_p(\rho )\to
\Gamma _+\cap \Gamma _-,\ t\to\mp \tau _\mp(\rho ) \} .
\end{equation}

\par Having fixed $T=T(K)$ above, let $\widetilde{K}\subset
G^{-1}(]-T(K),T(K)[)$ be a compact set containing the trapped set
$\Gamma _+\cap \Gamma _-$. For $\rho \in G^{-1}(T(K))$, define
\begin{equation}\label{opesc.31}
\sigma _+(\rho )=\sup \{t\in ]0,\tau _-(\rho )[;\, \exp(-[0,t]H_p)(\rho
)\subset G^{-1}(]-T(K),T(K)[)\setminus \widetilde{K} \} .
\end{equation}
When $\rho $ is outside the set $K_1$ above, assumed to
  be large enough, the $(-H_p)$-trajectory
through $\rho $ will hit $G^{-1}(-T(K))$ without reaching
$\widetilde{K}$ or get trapped and $\sigma _+(\rho )$ is the
corresponding hitting time which depends locally smoothly on $\rho
$.
For $\rho \in K_1\cap G^{-1}(T(K))$ it may also happen that the
trajectory hits $\widetilde{K}$ at the finite time $\sigma _+(\rho )$
or converges to $\Gamma _+\cap \Gamma _-$ without hitting
$\widetilde{K}$, in which case
$\sigma _+(\rho )=\tau _-(\rho )=+\infty $.

Notice that $\sigma _+$ is a lower semi-continuous function. Define
the open subset $\Omega _+$ of $G^{-1}(T(K))\times [0,+\infty [$ by
\begin{equation}\label{opesc.32}
\Omega _+:=\{ (\rho ,t)\in G^{-1}(T(K))\times [0,+\infty [;\, 0\le
t<\sigma _+(\rho ) \}.
\end{equation}
Then
$$
\widetilde{\Omega }_+=\{ \exp (-tH_p)(\rho );\, (\rho ,t)\in \Omega _+ \}
$$
is an open subset of $G^{-1}(]-T(K),T(K)])$ and the map
\begin{equation}\label{opesc.33}
\Omega _+\ni (\rho ,t)\mapsto \exp (-tH_p)(\rho )\in \widetilde{\Omega
}_+
\end{equation}
is a diffeomorphism. We have
\begin{equation}\label{opesc.34}
\widetilde{\Omega }_+=G^{-1}(]-T(K),T(K)])\setminus \Gamma _-(\widetilde{K}),
\end{equation}
where $\Gamma _-(\widetilde{K})$ denotes the {\it incoming
  $\widetilde{K}$-tail}, defined as
\begin{equation}\label{opesc.35}
\Gamma _-(\widetilde{K})=\Gamma _-\cup \{ \exp (-tH_p)(\rho );\, \rho
\in \widetilde{K},\ 0\le t<\tau _-(\rho ) \}.
\end{equation}
The intersection of $\Gamma _-(\widetilde{K})$ with
$G^{-1}([-T,+\infty [)$ is compact for every $T\in {\bf R}$.

\par Let $f_+\in C^{\infty }(\widetilde{\Omega }_+;]0,+\infty [)$ be
equal to $H_pG$ near $G^{-1}(T(K))$ and outside some bounded
set. Define $G_+\in C^\infty (\widetilde{\Omega }_+)$ by
\begin{equation}\label{opesc.36}
H_pG_+=f_+,\ G_+=G \hbox{ on } G^{-1}(T(K))
\end{equation}
and observe first that $G_+=G$ near $G^{-1}(T(K))$. By choosing $f_+$
large enough we may arrange so that
\begin{equation}\label{opesc.37}\begin{split}
&\limsup _{\widetilde{\Omega }_+\ni \nu \to \partial \widetilde{\Omega
  }_+}G_+(\nu )\le -T(K)-\frac{1}{C},\ C>0,\\
&\limsup _{\widetilde{\Omega }_+\ni \nu \to G^{-1}(-T(K))}G_+(\nu )\le -T(K),
\end{split}
\end{equation}
where $\partial \widetilde{\Omega }_+$ denotes the boundary of
$\widetilde{\Omega }_+$ as a subset $G^{-1}(]-T(K),T(K)])$, so that
$\partial \widetilde{\Omega }_+\subset \Gamma _-(\widetilde{K})$.
Since $G_+=G$ near $G^{-1}(T(K))$, we can extend $G_+$ by $G$ to a smooth
function
$G_+\in C^\infty (\widetilde{\Omega }_+ \cup G^{-1}(]T(K),\infty[)) $.
\par By construction, if $\chi \in C^\infty ({\bf R})$ and
$\mathrm{supp\,}\chi \subset [-T(K),+\infty [$, then $\chi \circ G_+$
is well-defined in $C^\infty ({\bf R}^{2n})$.

Next we briefly introduce the analogous quantities, $\Omega _-$, $G_-$:
For $\rho \in G^{-1}(-T(K))$, define
\begin{equation}
\label{opesc.38}
\sigma _-(\rho )=\sup \{t\in ]0,\tau _+(\rho )[;\, \exp([0,t]H_p)(\rho
)\subset G^{-1}(]-T(K),T(K)[)\setminus \widetilde{K} \} .
\end{equation}
$\sigma _-$ is a lower semi-continuous function. Let us also define
the open subset $\Omega _-$ of $G^{-1}(-T(K))\times [0,+\infty [$ by
\begin{equation}\label{opesc.39}
\Omega _-:=\{ (\rho ,t)\in G^{-1}(-T(K))\times [0,+\infty [;\, 0\le
t<\sigma _-(\rho ) \}.
\end{equation}
Then
$$
\widetilde{\Omega }_-:=\{ \exp (tH_p)(\rho );\, (\rho ,t)\in \Omega _- \}
$$
is an open subset of $G^{-1}([-T(K),T(K)[)$ and the map
\begin{equation}\label{opesc.40}
\Omega _-\ni (\rho ,t)\mapsto \exp (tH_p)(\rho )\in \widetilde{\Omega
}_-
\end{equation}
is a diffeomorphism. We have
\begin{equation}\label{opesc.41}
\widetilde{\Omega } _-=G^{-1}([-T(K),T(K)[)\setminus \Gamma _+(\widetilde{K}),
\end{equation}
where $\Gamma _+(\widetilde{K})$ denotes the {\it outgoing
  $\widetilde{K}$-tail}, defined as
\begin{equation}\label{opesc.42}
\Gamma _+(\widetilde{K})=\Gamma _+\cup \{ \exp (tH_p)(\rho );\, \rho
\in \widetilde{K},\ 0\le t<\tau _+(\rho ) \}.
\end{equation}
The intersection of $\Gamma _+(\widetilde{K})$ with
$G^{-1}(]-\infty ,T])$ is compact for every $T\in {\bf R}$.

\par Let $f_-\in C^{\infty }(\widetilde{\Omega }_+;]0,+\infty [)$ be
equal to $H_pG$ near $G^{-1}(-T(K))$ and outside some bounded
set. Define $G_-\in C^\infty (\widetilde{\Omega }_-)$ by
\begin{equation}\label{opesc.43}
H_pG_-=f_-,\ G_-=G \hbox{ on } G^{-1}(-T(K))
\end{equation}
and observe that $G_-=G$ near $G^{-1}(-T(K))$. By choosing $f_-$
large enough we may arrange so that
\begin{equation}\label{opesc.44}\begin{split}
&\liminf_{\widetilde{\Omega }_-\ni \nu \to \partial \widetilde{\Omega
  }_-}G_-(\nu )\ge T(K)+\frac{1}{C},\ C>0,\\
&\liminf_{\widetilde{\Omega }_-\ni \nu \to G^{-1}(T(K))}G_-(\nu )\ge T(K),
\end{split}
\end{equation}
where $\partial \widetilde{\Omega }_-$ denotes the boundary of
$\widetilde{\Omega }_+$ as a subset $G^{-1}([-T(K),T(K)[)$, so that
$\partial \widetilde{\Omega }_-\subset \Gamma _+(\widetilde{K})$.

\par Since $G_-=G$ near $G^{-1}(-T(K))$, we can extend $G_-$ by $G$ to a smooth
function
$G_-\in C^\infty (\widetilde{\Omega }_- \cup G^{-1}(]-\infty ,-T(K)[)) $. By construction, if $\chi \in C^\infty ({\bf R})$ and
$\mathrm{supp\,}\chi \subset ]-\infty ,T(K)]$, then $\chi \circ G_-$
is well-defined in $C^\infty ({\bf R}^{2n})$.

With $T=T(\widetilde{K})$, we define
\begin{equation}\label{opesc.45}
\widetilde{G}=\chi _+\circ G_++\chi _-\circ G_-\in C^\infty (\Sigma
_{[-\epsilon _0,\epsilon_0]}),
\end{equation}
where
\begin{itemize}
\item $\chi _\pm \in C^\infty ({\bf R};{\bf R})$, $\pm \chi _\pm \ge 0$,
\item $\chi _+(t)+\chi _-(t)=t$,
\item $\chi '_+>0$ on $]-T,+\infty [$,
\item $\chi '_->0$ on $[-\infty ,T[$,
\item $\mathrm{supp\,}\chi _+=[-T,+\infty [$, $\mathrm{supp\,}\chi
  _-=]-\infty ,T ]$,
\end{itemize}

We notice that $\widetilde{G} = G$ outside a bounded
subset of $G^{-1}([-T(K),T(K)])$ and that
\begin{equation}\label{opesc.46}
\widetilde{G}^{-1}(0)\supset\{\rho ;\, G_+(\rho )\le -T(K),\ G_-(\rho )\ge
T(K) \}\supset \Gamma _+(\widetilde{K})\cap \Gamma _-(\widetilde{K}).
\end{equation}
It is also clear that $\widetilde{G}^{-1}(0)\subset
G^{-1}(]-T,T[)$. Moreover, we can choose $f_+$, $f_-$ so that
\begin{itemize}
\item[] the set $G_+(\rho )\le -T(K)$ is contained in an arbitrarily
  small neighborhood of $G^{-1}(-T(K))\cup \Gamma _-(\widetilde{K})$,
\item[] the set $G_-(\rho )\ge T(K)$ is contained in an arbitrarily
  small neighborhood of $G^{-1}(T(K))\cup \Gamma _+(\widetilde{K})$.
\end{itemize}

Outside a bounded set, we have $\widetilde{G}=G$, $\| H_G\|_g\asymp 1$ and
$H_pG\ge m/{\cal O}(1)$ by (\ref{opesc.27}), (\ref{opesc.21}). In a
bounded set, we use (\ref{opesc.45}) and get
\begin{equation}\label{opesc.47}
H_p\widetilde{G}=(\chi _+'\circ G_+)f_++(\chi _-'\circ G_-)f_-\asymp
\chi _+'\circ G_++\chi _-'\circ G_-,
\end{equation}
\begin{equation}\label{opesc.48}
  \| H_{\widetilde{G}}\|_g=(\chi _+'\circ G_+)\| H_{G_+}\|_g+(\chi
  _-'\circ G_-) \| H_{G_-}\|_g
  \asymp
  \chi _+'\circ G_++\chi _-'\circ G_-,
  \end{equation}
so
\begin{equation}\label{opesc.49}
H_p\widetilde{G}\asymp m\|H_{\widetilde{G}}\| _g\hbox{ uniformly on }\Sigma
_{[-\epsilon _0,\epsilon _0]} ,
\end{equation}
in addition to the fact that $H_p\widetilde{G}\in S(m)$ and $\|
H_{\widetilde{G}}\|_g\asymp 1$ away from a bounded set.

\par We can arrange so that $H_p\widetilde{G}>0$ outside an
arbitrarily small neighborhood of $\Gamma _+(\widetilde{K})\cap \Gamma
_-(\widetilde{K})$.

\par We now strengthen the assumption $G\in \dot{S}(\widetilde{r}R)$
to
\begin{equation}\label{opesc.50}
G\in S(\widetilde{r}R).
\end{equation}
Then $\widetilde{G}\in S(\widetilde{r}R)$
and outside a bounded set we have
$$
\widetilde{G}={\cal O}(\widetilde{r}R)\| H_{\widetilde{G}}\|_g.
$$
Choose $\chi _\pm$ so that
$$
\chi _\pm ={\cal O}(\chi _\pm ')\hbox{ uniformly on any bounded set.}
$$
Then from (\ref{opesc.45}), (\ref{opesc.48}), we conclude that
\begin{equation}\label{opesc.51}
\widetilde{G}={\cal O}(\widetilde{r}R)\| H_{\widetilde{G}}\|_g,\hbox{ uniformly on }\Sigma _{[-\epsilon _0,\epsilon_0]}.
\end{equation}

\par Assume
\begin{equation}\label{opesc.52}
m\asymp 1 \hbox{ on }\Sigma _{[-\epsilon _0,\epsilon_0]}.
\end{equation}
Let $\chi \in C_0^\infty (]-{\epsilon _0} ,{\epsilon _0} [;[0,1])$ be equal
to 1 on $[-{\epsilon _0} /2,{\epsilon _0} /2]$ and define globally,
\begin{equation}\label{opesc.53}
G^0(\rho )=\chi (p(\rho ))\widetilde{G}(\rho )\in C^\infty ({\bf R}^{2n}),
\end{equation}
with the convention that $G^0=0$ outside $\Sigma _{[-\epsilon _0,\epsilon_0]} $.
By (\ref{opesc.52}) we have $\chi (p(\rho ))\in S(1)$ and hence
$G^0\in S(\widetilde{r}R)$ by (\ref{opesc.50}) and the subsequent observation.
Then
\begin{equation}\label{opesc.54}
H_pG^0=\chi (p)H_p\widetilde{G},
\end{equation}
\begin{equation}\label{opesc.55}
H_{G^0}=\chi (p)H_{\widetilde{G}}+\widetilde{G}\chi '(p)H_p.
\end{equation}
From the last equation, (\ref{opesc.27.5}) and (\ref{opesc.52}), we get
\[
\begin{split}
\| H_{G^0}\|_g
&
\le \chi (p)\| H_{\widetilde{G}}\|_g+{\cal O}(1)|\chi
'(p)||\widetilde{G}|\frac{m}{\widetilde{r}R}\\
&\le \chi (p)\| H_{\widetilde{G}} \|_g+{\cal O}(1)\frac{|\chi '(p)
  \widetilde{G}|}
{\widetilde{r}R}={\cal O}(1),
\end{split}
\]
leading to
$$
\|H_{G^0}\|_g^2\le {\cal O}(1)\left(\chi (p)^2\|H_{\widetilde{G}}\|_g^2 +\frac{\chi (p)|{\widetilde{G}}|^2}{(\widetilde{r}R)^2} \right),
$$
in view of the standard estimate, $\chi '={\cal O}(\chi ^{1/2})$ for
non-negative smooth functions. Now apply (\ref{opesc.51}) to get
\begin{equation}\label{opesc.56}
\|H_{G^0}\|_g^2\le {\cal O}(1)\chi (p)\|H_{\widetilde{G}}\|_g^2\le {\cal O}(1)\chi (p)\|H_{\widetilde{G}}\|_g,
\end{equation}
where the last inequality follows from (\ref{opesc.26}) which also
holds for $\widetilde{G}$. Combining this
with (\ref{opesc.49}), (\ref{opesc.54}), we get
\begin{equation}\label{opesc.57}
H_pG^0\ge \frac{m}{{\cal O}(1)}\| H_{G^0}\|_g^2.
\end{equation}

\par We sum up the constructions in
\begin{prop}\label{opesc2}
Let $r,R,\widetilde{r}$ be as in {\rm (\ref{opesc.1})}--{\rm (\ref{opesc.4})},
define the metric $g$ by {\rm (\ref{opesc.4.5})}. Let $P$, $p$ $m$ be as in {\rm (\ref{opesc.5})}--{\rm (\ref{opesc.9})},
where $1\le m_0\in S(m_0)$. Assume {\rm (\ref{opesc.14})} with $p_\mathrm{class}$ as in {\rm (\ref{opesc.12})}. Define the energy
 slice $\Sigma _{[-\epsilon _0,\epsilon_0]}$ by {\rm (\ref{opesc.23})} for some $\epsilon
 _0>0$ and let $G\in S(\widetilde{r}R)$ be an escape function in the
 sense of Definition {\rm \ref{opesc1}} and assume {\rm (\ref{opesc.22})}, so that {\rm (\ref{opesc.25})} holds if
 $\epsilon _0>0$ is small enough, and fix such a choice of $\epsilon
 _0$. Let $\widetilde{K}\subset \Sigma _{[-\epsilon _0,\epsilon_0]}$ be a compact
 set which contains the trapped set $\Gamma _+\cap \Gamma _-$ (cf.\
 {\rm (\ref{opesc.29})}). Define the outgoing and incoming
 $\widetilde{K}$-tails $\Gamma _+(\widetilde{K})$, $\Gamma
 _-(\widetilde{K})$ by {\rm (\ref{opesc.42})}, {\rm (\ref{opesc.35})}, so that
 $\widehat{\widetilde{K}}:=\Gamma _+(\widetilde{K})\cap \Gamma
 _-(\widetilde{K})\subset \Sigma _{[-\epsilon _0,\epsilon_0]}$ is a compact set; ``the
 $H_p$-convex hull'' of $\widetilde{K}$.

\par Then, after modifying $G$ on a bounded set to a new function $\widetilde{G}$, we can achieve that
\begin{itemize}
\item $H_p{\widetilde{G}}\asymp m\|H_{\widetilde{G}}\|_g$ uniformly on $\Sigma _{\epsilon_0} $,
\item $H_p{\widetilde{G}}>0$ outside any fixed given neighborhood of $\widehat{\widetilde{K}}$,
\item ${\widetilde{G}}=0$ in a neighborhood of $\widehat{\widetilde{K}}$.
\end{itemize}

\par If we also assume {\rm (\ref{opesc.52})} and define $G^0\in
S(\widetilde{r}R)$ as in {\rm (\ref{opesc.53})}, then we have {\rm (\ref{opesc.57})}.
\end{prop}

\section{Microlocal approach to resonances (\cite{HeSj86})}\label{nore}
\setcounter{equation}{0}

In this section we review some basic notions developed in the first
half of \cite{HeSj86}.
\subsection{I-Lagrangian manifolds}
Let $G\in \dot{S}(\widetilde{r}R)$ be real-valued. Then the manifold
\begin{equation}\label{nore.1}
\Lambda _G=\{ (x,\xi )\in {\bf C}^{2n};\, \Im (x,\xi )=H_G (\Re (x,\xi
))\}
\end{equation}
is I-Lagrangian, i.e. Lagrangian in ${\bf C}^{2n}$ for the real
symplectic form $-\Im \sigma $, where $\sigma =\sum d\xi _j\wedge
dx_j$ is the complex symplectic form.
Since $\Lambda _G$ is I-Lagrangian, $d(-{{\Im (\xi \cdot
    dx)}_\vert}_{\Lambda _G})=-\Im {{\sigma }_\vert}_{\Lambda _G}=0$
and since $\Lambda _G$ is topologically trivial, we know that $-{{\Im
  (  \xi \cdot dx)}_\vert}_{\Lambda _G}$ is exact and hence $=dH$ for
some smooth function $H\in C^\infty (\Lambda _G)$. The primitive
$H$ is unique up to a constant and we can choose
\begin{equation}\label{nore.7}
H=-\Re \xi \cdot \Im x +G(\Re (x,\xi ))=G(\Re (x,\xi ))-\Re \xi \cdot
G'_\xi (\Re (x,\xi )).
\end{equation}

\par If we also assume that $G$ is
small in $\dot{S}(\widetilde{r}R)$, then $\Lambda _G$ is ${\bf
  R}$-symplectic, i.e. a symplectic submanifold of ${\bf C}^{2n}$,
equipped with the symplectic form $\Re \sigma $. In other words,
${{\sigma }_\vert}_{\Lambda _G}$ is a (real) symplectic form on
$\Lambda _G$ and we have the volume element
$$
d\alpha =\frac{1}{n!}{{\sigma ^n}_\vert}_{\Lambda _G}.
$$

\subsection{FBI-transforms}
(\ref{nore.1}) gives a parametrization ${\bf R}^{2n}\ni \rho
\mapsto \rho +iH_G(\rho )$ of $\Lambda _G$ and we can then define
symbol spaces $S(m)=S(m,\Lambda _G)$ of functions on $\Lambda _G$ by
pulling back functions and weights to ${\bf R}^{2n}$. In particular, we
define the scales $\widetilde{r}$ and $R$ by this pull back.
Let $\lambda =\lambda (\alpha )\in S(\widetilde{r}R^{-1},\Lambda _G )$
be positive, elliptic (in the sense that $\lambda $ is
  non-vanishing and $1/\lambda \in S((\widetilde{r}R^{-1})^{-1},\Lambda _G )$) and put
\begin{equation}\label{nore.2}
\phi (\alpha ,y)=(\alpha _x-y)\alpha _\xi +i\lambda (\alpha )(\alpha
_x-y)^2,\ \alpha =(\alpha _x,\alpha _\xi )\in \Lambda _G,\ y\in {\bf C}^n.
\end{equation}
This will be the phase in our FBI-transform.

\par The amplitude will be a ${\bf C}^{n+1}$-valued smooth function
${\bf t}(\alpha ,y;h)$ on $\Lambda_G \times {\bf C}^n_y$
which is affine linear in $y$. When discussing symbol properties of such
functions we restrict the attention to a region
\begin{equation}\label{nore.3}
|y-\alpha _x|<{\cal O}(1)R(\alpha _x),
\end{equation}
and with this convention, we require that ${\bf t}\in
h^{-3n/4}S(\widetilde{r}^{n/4}R^{-n/4})$ and that
${\bf t}, \partial _{y_1}{\bf t},...,\partial
_{y_n}{\bf t}$ are maximally linearly independent in the
sense that with ${\bf t}$ treated as a column vector,
\begin{equation}\label{nore.4}
\left|\det \begin{pmatrix}{\bf t} &\partial
  _{y_1}{\bf t}
  &... &\partial _{y_n}{\bf t}\end{pmatrix}\right|\asymp R^{-n}\left( h^{-\frac{3n}{4}}\widetilde{r}^{\frac{n}{4}}R^{-\frac{n}{4}} \right)^{n+1}.
\end{equation}
Notice that the determinant is independent of $y$. If
$e_0,e_1,...,e_n$ is the canonical basis in ${\bf C}^{n+1}$, we can
choose
$$
{\bf t}(\alpha ,y)={\bf t}_0(\alpha
;h)+\sum_1^n(\alpha _{x_j}-y_j){\bf t}_j(\alpha ;h),
$$
where,
$$
{\bf t}_j=t_je_j, \hbox{ and } t_j(\alpha
;h)=\frac{t_0(\alpha ;h)}{R} \hbox{ for }j>0
$$
and $t_0\in h^{-3n/4}S(\widetilde{r}^{n/4}R^{-n/4})$ is elliptic.

\begin{remark}\label{nore1}
  If ${\bf s}(\alpha ,y;h)$ is a second amplitude with the
  same properties as ${\bf t}$, then it is not hard to show
  that there exists $U(\alpha ;h):{\bf C}^{n+1}\to {\bf C}^{n+1}$,
  independent of $y$ and invertible, such that
$$U,\, U^{-1}\in S(1)\hbox{ and }{\bf s}(\alpha
,y;h)=U(\alpha ;h){\bf t}(\alpha ,y;h). $$

\end{remark}

\par Let $\chi \in C_0^\infty (B(0,1/C))$ be equal to one in
$B(0,1/(2C))$, where $C>0$ is large enough. We define the
FBI-transform $T:{\cal D}'({\bf R}^n)\to C^\infty (\Lambda _G;{\bf
  C}^{n+1})$ by
\begin{equation}\label{nore.6}
Tu(\alpha ;h)=\int e^{\frac{i}{h}\phi (\alpha
  ,y)}{\bf t}(\alpha ,y;h) \chi _\alpha (y)u(y)dy,
\end{equation}
where $\chi _\alpha (y)=\chi ((y-\Re \alpha _x)/R(\Re \alpha
_x))$. Here the domain of integration is equal to ${\bf
    R}^n$ and the integral is defined as the bilinear scalar
  product of $u\in {\cal D}'({\bf R}^n)$ and a test function in
  $C_0^\infty ({\bf R}^n)$.

\par We assume from now on that $G$
belongs to $S(\widetilde{r}R)$. We also assume:
\begin{equation}\label{nore.7.5}\begin{split}
&\exists\, g_0=g_0(x)\in S(rR), \hbox{ such that }G(x,\xi )-g_0(x)\\
&\hbox{has
its support in a region where }|\xi |\le {\cal O}(r(x))\\
&\hbox{and }G(x,\xi )-g_0(x) \hbox{ is sufficiently small in }S(rR).
\end{split}\end{equation}

\par We will also consider the more special situation,
when $G\in S(m_G)$:
\begin{equation}\label{nore.7.7}\begin{split}
&\exists\, g=g(x)\in S(m_G^0), \hbox{ such that }G(x,\xi )-g(x)\\
&\hbox{has
its support in a region where }|\xi |\le {\cal O}(r(x)).
\\
&\hbox{and }G(x,\xi )-g_0(x) \hbox{ is sufficiently small in }S(m_G).
\end{split}\end{equation}
Here $m_G\le \widetilde{r}R$ is an order function, and we have put
$m_G^0(x)=m_G(x,0)$.

\par Let $H$ be given in (\ref{nore.7}). Then $H\in
S(\widetilde{r}R)$. Under the more restrictive assumption
(\ref{nore.7.7}), we have
\begin{equation}\label{nore.7.8}
H\in S(m_G).
\end{equation}

Using $T$ we shall define the function spaces $H(\Lambda _G,m)$,
essentially by requiring that
$$
Tu\in L^2(\Lambda _G,m^2e^{-2H/h}d\alpha )=:L^2(\Lambda _G,m).
$$
Here, $m$ is an order function: $0<m\in S(m)$. An intuitive reason for
the appearance of $H$ here is the following: The function
\begin{equation}\label{nore.8}
f(y,\theta )=-\Im y\cdot \theta +G(\Re y,\theta ),\ \theta \in {\bf R}^n
\end{equation}
is a nondegenerate phase function on ${\bf C}^n\times {\bf R}^n$ in
the sense of H\"ormander's theory of Fourier integral operators
(apart from a homogeneity condition) with $\theta $ as the fiber
variables. The corresponding critical manifold $C_f$ is given by
$$
f'_\theta (y,\theta )=0:\ \Im y=G'_\eta (\Re y,\theta )
$$
and the associated I-Lagrangian manifold is
$$
\{ (y,\frac{2}{i}\partial _yf(y,\theta ));\, (y,\theta )\in C_f \} =
\Lambda _G .
$$
We are beyond the scope of
H\"ormander's theory, but from  this it is natural to  define the
space $H(\Lambda _G,m)$ by saying that a distribution $u$ should
belong to it when $Tu\in L^2(\Lambda
_G,m^2e^{-2\widetilde{H}/h}d\alpha )$, where
$$
\widetilde{H}(\alpha )=\mathrm{v.c.}_{(y,\theta )}\left(-\Im \phi (\alpha,y)+f(y,\theta )\right).
$$
Here $\mathrm{v.c.}_{(y,\theta )}$ indicates that we take the critical
value with respect to the variables $(y,\theta )$.
The critical point is nondegenerate and given by $(y,\theta
)=(\alpha _x,\Re \alpha _\xi )$ and we get
$$
\widetilde{H}(\alpha )=H(\alpha ).
$$

\par Letting $\Lambda _0={\bf R}^{2n}$, we can find an FBI-transform
$$
T_0:{\cal D}'({\bf R}^{n})\to C^\infty (\Lambda _0;{\bf C}^{n+1})
$$
given by
$$
T_0u(\beta ;h)=\int e^{\frac{i}{h}\phi_0(\beta,y)}{\bf s}(\beta ,y ;h)\widetilde{\chi }_\beta (y)u(y) dy,
$$
which is equivalent to $T$ in the sense of (\ref{nore.11})
below, provided that
$$
\phi _0(\beta ,y)=(\beta _x-y)\cdot \beta _\xi +i\lambda _0(\beta
)(\beta _x-y)^2
$$
and ${\bf s}$, $\widetilde{\chi }_\beta $ are chosen suitably. First we need a
bijection $\Lambda _G\ni \alpha \mapsto \beta \in \Lambda _0$ and we
define $\beta =\beta (\alpha )$ by imposing the condition
$$
\{ \beta  \}=\Lambda _0\cap \{ (y,-\partial _y\phi (\alpha ,y));\,
y\in {\bf C}^n \},
$$
which gives
\begin{equation}\label{nore.9}
\beta =(\beta _x,\beta _\xi )=\left( \Re \alpha _x+\frac{\Im \alpha
    _\xi }{2\lambda (\alpha )},\Re \alpha _\xi -2\lambda (\alpha )\Im
  \alpha _x\right) .
\end{equation}
This gives a bijection $\Lambda _G\to \Lambda _0$ with inverse $\beta
\to \alpha (\beta )$, both having the natural symbol properties. We
define the elliptic element $0<\lambda _0\in S(\widetilde{r}R^{-1})$ by
\begin{equation}\label{nore.10}
\lambda _0(\beta )=\lambda (\alpha (\beta )).
\end{equation}
By construction the two quadratic polynomials $\phi (\alpha ,\cdot )$ and $\phi
_0(\beta ,\cdot )$ have the same gradients and Hessians at the point
$y=\beta_x$, so they differ by a constant (independent of $y$). More
explicitly,
$$
\phi (\alpha ,y)=\phi (\alpha ,\beta _x)+\phi _0(\beta ,y).
$$

\par Finally, choose
$$
{\bf s}(\beta ,y)={\bf t}(\alpha ,y),\ \widetilde{\chi
  }_\beta (y)=\chi _\alpha (y).
$$
Then
\begin{equation}\label{nore.11}
Tu(\alpha ;h)=e^{\frac{i}{h}\phi (\alpha ,\beta _x)}T_0u(\beta ;h),
\end{equation}
which expresses the equivalence of $T$ and $T_0$.

\par It follows that if we identify order functions on $\Lambda _G$
and on $\Lambda _0$ in a natural way, then we have the equivalence
$$
Tu\in L^2(\Lambda _G,m^2e^{-2H/h}d\alpha ) \Leftrightarrow T_0u\in
L^2(\Lambda _0,m^2e^{-2F/h}d\beta  ),
$$
where
\begin{equation}\label{nore.12}
F=H+\Im \phi (\alpha ,\beta _x) = \mathrm{v.c.}_{y,\theta }\left(-\Im\phi
_0(\beta ,y)+f(y,\theta )\right).
\end{equation}

\par Let $G_1,\,G_2\in S(\widetilde{r}R)$ be as above and let $f_1,\,
f_2$ and $F_1,\, F_2$ be the corresponding functions. In
\cite{HeSj86} it was shown, using (\ref{nore.12}) and a corresponding
inverse ``Legendre'' formula, that we have the equivalence,
\begin{equation}\label{nore.13}
G_1\le G_2 \Leftrightarrow F_1\le F_2.
\end{equation}
From this and the description with the help of $T_0$ it will follow
that we have the inclusion $H(\Lambda _{G_1},m)\subset H(\Lambda _{G_2},m)$, when $G_1\le G_2$.

\subsection{Sobolev spaces with expo\-nen\-tial pha\-se spa\-ce
  wei\-ghts}\label{hlg}
Let $G$ satisfy (\ref{nore.7.5}) and be
sufficiently small in $S(\widetilde{r}R)$. Define $H$ as in
(\ref{nore.7}), let $m$ be an order function on $\Lambda _G$ and let
$T$ be an associated FBI-transform as in (\ref{nore.6}). In
\cite{HeSj86} it is shown that $T$ is injective on
  $C_0^\infty ({\bf R}^n)$ and also on more general Sobolev spaces with
  exponential weights,  by the construction of
an approximate left inverse of $T$ which works with exponentially
small errors.
\begin{dref}\label{nore2}
$H(\Lambda _G,m)$ is the completion of $C_0^\infty ({\bf R}^n)$ for the
norm
\begin{equation}\label{nore.14}
\| u\|_{H(\Lambda _G,m)}=\| Tu\|_{L^2(\Lambda _G,m^2e^{-2H/h}d\alpha )}.
\end{equation}
\end{dref}
The following facts were established in \cite{HeSj86}:
\begin{itemize}
\item $H(\Lambda _G,m)$ is a Hilbert space
\item If we modify the choice of $\lambda $ and ${\bf t}$
  in the definition of $T$, we get the same space $H(\Lambda _G,m)$
  and the new norm is uniformly equivalent to the earlier one, when
  $h\to 0$.
\item When $G=g(x)$ is independent of $\xi $ and $m=m_0(x)$, we get
$$
H(\Lambda _G,m)=L^2({\bf R}^n;m_0^2e^{-2g(x)/h}dx)
$$
with uniform equivalence of norms. More generally, when $m(x,\xi
)=m_0(x)(\widetilde{r}(x,\xi )/r(x))^{N_0}$, $N_0\in {\bf R}$, then
$H(\Lambda _G,m)$ is the naturally defined exponentially weighted
Sobolev space.
\end{itemize}
\begin{remark}\label{nore3}
From the last point, we know that $L^2({\bf R}^n)=H(\Lambda _0,1)$
(when $G=0$) with uniformly equivalent norms. As in {\rm \cite{HeSj86}},
this can be improved:

\par There exists a positive weight $1\asymp M_0(\alpha ;h)\in S(1)$
such that if $L_0(d\alpha )=M_0(\alpha ;h)^2L(d\alpha )$ ($L$ being
the Lebesgue measure), then
\begin{equation}
\label{nore3.1}
(u|v)_{L^2({\bf R}^n)}=\int_{\Lambda _0} Tu \overline{Tv}L_0(d\alpha
)+(Ku|v)_{L^2({\bf R}^n)},
\end{equation}
where $K$ is negligible of order 1 (as defined in the beginning of
Subsection {\rm \ref{pfops}}) so that for every $N\in
{\bf N}$,
$$
K={\cal O}(1):\, H(\Lambda _0, (\widetilde{r}R/h)^{-N})\to H(\Lambda _0, (\widetilde{r}R/h)^{N}).
$$
Notice that the weight $H$ is zero when $G=0$.
\end{remark}

\subsection{Pseudodifferential- and Fourier integral
  operators}\label{pfops}
Such operators can be defined directly (cf.\ (6.3), (7.8) in
\cite{HeSj86}). We will only need their descriptions on the FBI-side,
somewhat in the spirit of Toeplitz operators.

Let $m$ be an order function on $\Lambda _G$. We say that $R:\,H(\Lambda
_G,m)\to H(\Lambda _G,1)$ is negligible of order $m$ if for every
order function $\widetilde{m}$ and every $N_0\in {\bf N}$, $R$ is a
well defined operator $H(\Lambda _G,m\widetilde{m})\to H(\Lambda
_G,\widetilde{m}(\widetilde{r}R/h)^{N_0})$ which is uniformly
bounded in the limit $h\to 0$. (``Well defined'' here refers to the
existence of a unique extension from the dense subspace $C_0^\infty
({\bf R}^n)$.) We have a completely analogous notion of negligible
operators of order $m$: $L^2(\Lambda _G,m)\to L^2(\Lambda _G,1)$.
Here, we write $L^2(\Lambda
  _G,m)=L^2(\Lambda _G,m^2e^{-2H/h}d\alpha )$ for short.
We will use the abbreviation\begin{itemize}
\item[]nop $=$ negligible operator,
\item[]pop $=$ pseudodifferential operator,
\item[]top $=$ Toeplitz operator.
\end{itemize}

\par Let $\Pi $ be the orthogonal projection $L^2(\Lambda _G,m)\to
TH(\Lambda _G,m)$. Then (see \cite{HeSj86}, (7.24) and the adjacent
discussion)
\begin{equation}\label{nore.15}\begin{split}
&\Pi =\widetilde{\Pi }+\Pi _{-\infty },\ \Pi _{-\infty }\hbox{ is
}L^2\hbox{-negligible of order }1,\\
&\widetilde{\Pi }u(\alpha )=\int p(\alpha ,\beta ;h)e^{\frac{i}{h}\psi
(\alpha ,\beta )}u(\beta )m(\beta )^2e^{-2H(\beta )/h}d\beta ,\end{split}
\end{equation}
where $\psi $ is independent of $m$ and of class $S(\widetilde{r}R)$
in a region
$\{ (\alpha ,\beta );\, d_{g}(\alpha ,\beta )\le 1/{\cal O}(1) \}$
and satisfies,
\begin{equation}\label{nore.16}
-\Im \psi (\alpha ,\beta )-H(\alpha )-H(\beta )\asymp
-\left(\frac{\widetilde{r}}{R}|\alpha _x-\beta
_x|^2+\frac{R}{\widetilde{r}}|\alpha _\xi -\beta _\xi |^2  \right).
\end{equation}
Moreover,
\begin{equation}\label{nore.17}
p\in S(m^{-2}h^{-n})\hbox{ is supported in a region }
d_{g}(\alpha ,\beta )\le 1/{\cal O}(1)
\end{equation}
and we have
\begin{equation}\label{nore.18}
\overline{\psi (\alpha ,\beta ) }=-\psi (\beta ,\alpha ),\
\overline{p(\alpha ,\beta ;h)}=p(\beta ,\alpha ;h).
\end{equation}
We refrain from recalling the characterization of $TH(\Lambda _G,m)$
as the approximate null space of a left ideal of pseudodifferential
operators.

We also have a class of pseudodifferential operators of order $m$ (\cite{HeSj86})
$A:\, H(\Lambda _G,\widetilde{m})\to H(\Lambda _G,\widetilde{m}/m)$,
$\forall \widetilde{m}$. Such
an operator has an associated principal symbol $\sigma (A)\in
S(m,\Lambda _G)/S(mh/(\widetilde{r}R),\Lambda _G)$ which determines
the operator $A$ up to an operator of order $mh/(\widetilde{r}R)$ and the
principal symbol map is a bijection between the corresponding quotient
spaces of operators and of symbols. We also have the usual result for
the composition modulo negligible operators.

\par When $P$ is an $h$-differential operator as in
(\ref{opesc.5})--(\ref{opesc.13}) with coefficients that are
holomorphic near $\pi _x\mathrm{supp\,}G$, then $P$ is an
$h$-pseudodifferential operator of order $m=m_0(x)(\widetilde{r}(x,\xi
)/r(x))^{N_0}$, associated to $\Lambda _G$ and the corresponding
principal symbol is
\begin{equation}\label{nore.19}
{{p}_\vert}_{\Lambda _G}.
\end{equation}

According to Proposition 7.3 in \cite{HeSj86} the classes $\{ \Pi b\Pi
;\, b\in S(m)\}$ and $\{ TAT^{-1}\Pi ;\, A\hbox{ is an }h\hbox{-pseudo
  of order }m\hbox{ associated to }\Lambda _G
 \}$ coincide modulo negligible operators of order $m$. Moreover, $b$
 and $A$ are related by
\begin{equation}\label{nore.20}
b\equiv \sigma _A\hbox{ mod }S\left(\frac{mh}{\widetilde{r}R} \right) .
\end{equation}

\par Now, let $G_0$, $G_1$ be two functions with the properties
of $G$ above. Then (see the beginning of Chapter 7 in \cite{HeSj86})
there exists a smooth real bijective canonical transformation $\kappa :\Lambda
_{G_0}\to \Lambda _{G_1}$ such that, writing $(x,\xi )=\kappa (y,\eta
)$, we have
$$
x-y\in S(R),\ \xi -\eta \in S(\widetilde{r}),
$$
either as functions of $(y,\eta )\in \Lambda _{G_0}$ or of $(x,\xi
)\in \Lambda _{G_1}$. We can then define Fourier integral operators of
order $m$, associated to $\kappa $; $A={\cal O}(1)$: $H(\Lambda
_{G_0},\widetilde{m})\to H(\Lambda _{G_1},\widetilde{m}/m),$ $\forall
\widetilde{m}$. Such operators have the usual composition result up to
negligible operators. Moreover, we have the usual notion of elliptic
operators: If $U:H(\Lambda _{G_0},m)\to H(\Lambda _{G_1},1)$ is an
elliptic Fourier integral operator of order $m$, then (for $h$ small
enough) $U$ is bijective and the inverse is an elliptic Fourier
integral operator of order $m^{-1}$ associated to $\kappa ^{-1}$ up to
a negligible operator of order $m^{-1}$. We also have a corresponding
Egorov's theorem: With $U$ as above, let $A$ be a pseudodifferential
operator of order $\widehat{m}$ associated to $\Lambda _{G_1}$. Then
$B=U^{-1}AU$ is a pseudodifferential operator of order $\widehat{m}$
, associated to
$\Lambda _{G_0}$
(up to a negligible operator of the same order), and the principal
symbols are related by
\begin{equation}\label{nore.21}
\sigma _B=\sigma _A\circ \kappa .
\end{equation}

We now specify the above in the case when
\begin{equation}\label{nore.21.2}G_0=0,\ G_1=G\end{equation}
 and in
doing so we go slightly beyond \cite{HeSj86}. Since there will be
several different symplectic frameworks, let us denote the standard
real Hamilton field of $G$ on ${\bf R}^{2n}$, by $H_G^{{\bf
    R}^{2n}}$. Recall that
\begin{equation}\label{nore.21.5}
\Lambda _{\upsilon G}=\{ \rho \in {\bf C}^{2n};\, \Im\rho = \upsilon H_G^{{\bf
    R}^{2n}}(\Re\rho ) \}.
\end{equation}

\par
We let $\sigma =\sum_1^n d\xi _j\wedge dx_j$ denote the
  complex symplectic form on ${\bf C}^n_x\times {\bf C}^n_\xi $. The
  real and imaginarty parts $\Re \sigma $ and $\Im \sigma $ are real
  symplectic forms. When $f$ is a real $C^1$ function on some open
  subset of ${\bf C}^{2n}$, we let $H_f^{\Re \sigma }$ and $H_f^{\Im
    \sigma }$ denote the corresponding Hamilton fields. In general, if $r=p+iq$ is
differentiable with complex-linear differential at some point, then at
that point (cf.\ \cite{Sj82}, (11.5), (11.6)),
\begin{equation}\label{nore.21.7}
\widehat{H}_r=H_q^{\Im \sigma },\ \ J\widehat{H}_r=H^{\Im \sigma }_p.
\end{equation}
Here, $J=\hbox{multiplication of
  tangent vectors with }i$,
$H_r$ denotes the complex Hamilton field for $\sigma $ (of type 1,0)
and the hat indicates that we take the corresponding real vector field;
$\widehat{H}_r=H_r+\overline{H}_r$,
$H_r=r'_\zeta \cdot \partial _z-r'_{z}\cdot \partial _{\zeta }$.

\par Returning to (\ref{nore.21.5}), if $G(\rho )=G(\Re \rho )$ is considered as a function on ${\bf
  C}^{2n}$, we have
$$
H_G^{\Im \sigma }=JH_G^{{\bf R}^{2n}}.
$$
Then we can view the family $\Lambda _{\upsilon G}$ as obtained from $\Lambda
_0$ by deformation with the field
$$
\nu _\upsilon ={{H_G^{\Im \sigma }}_\vert}_{\Lambda _{\upsilon G}}.
$$
Since $\Lambda _{\upsilon G}$ is $I$-Lagrangian and we get the same
deformation is we modify $\nu _\upsilon $ by adding a field tangent to
$\Lambda _{\upsilon G}$, we can replace $\nu _\upsilon $ with $\widetilde{\nu }_\upsilon ={{H^{\Im \sigma
    }_{F_\upsilon }}_\vert}_{\Lambda _{\upsilon G}}$, if $F_\upsilon $ is real, smooth and
$F_\upsilon =G$ on $\Lambda _{\upsilon G}$.

\par Let $\widetilde{G}_\upsilon $ be an almost holomorphic extension from $\Lambda
_{\upsilon G}$ of ${{G}_\vert}_{\Lambda _{\upsilon G}}$. Then
at $\Lambda _{\upsilon G}$,
$$J\widehat{H}_{\widetilde{G}_\upsilon }=H^{\Im \sigma }_{\Re
  \widetilde{G}_\upsilon }\equiv H^{\Im \sigma }_G\ \mathrm{mod}\ T\Lambda _{\upsilon G},
$$
by (\ref{nore.21.7}),
so $J\widehat{H}_{\widetilde{G}_\upsilon }$ generates the family $\Lambda
_{\upsilon G}$ by deformation from $\Lambda _0$.

\par $\widetilde{G}_\upsilon $ can be constructed in the following way:
Consider the map
$$
\theta =\theta _\upsilon :\, {\bf R}^{2n}\ni \rho \mapsto \rho +i\upsilon H_G^{{\bf
    R}^{2n}}(\rho )=:\rho +i\gamma _\upsilon (\rho )\in {\bf C}^{2n}.
$$
For $k\in {\bf N}$, $\partial _\upsilon ^k\gamma _\upsilon $ is of class $S(1)$ for
the metric $g$. (Here we use that $H_G$ is of class
  $S(1)$ for the metric $g$.) Thus $\partial _\upsilon ^k\theta_\upsilon  $ is of class
$\dot{S}(1)$ when $k=0$ and of class $S(1)$ when $k\ge 1$.

\par Let $$\widetilde{\theta }_\upsilon :{\bf C}^{2n}\to {\bf C}^{2n}$$ be an
almost holomorphic extension of $\theta _\upsilon $ with the same symbol properties and let
$\widetilde{G}\in S(\widetilde{r}R)$ be an almost holomorphic
extension of $G$. We notice that $\widetilde{\theta }_\upsilon $ is a local
diffeomorphism and that
$\widetilde{\theta }_\upsilon ^{-1}:\, \mathrm{neigh\,}(\Lambda _{\upsilon G},{\bf
  C}^{2n})\to {\bf C}^{2n}$
is almost holomorphic at $\Lambda _{\upsilon G}$ with the same symbol
properties. Then
$\widetilde{G}_\upsilon :=\widetilde{G}\circ \widetilde{\theta }_\upsilon ^{-1}$ has
the required properties. One can also see that it can be defined in a
$1/{\cal O}(1)$-neighborhood of $\Lambda _{\upsilon G}$ for $g$, and be of class
$S(\widetilde{r}R)$ there with all its $t$-derivatives. Using
$\widetilde{G}_\upsilon $, we get a smooth family of canonical transformations
$\kappa _\upsilon :\Lambda _0\to \Lambda _{\upsilon G}$ by integration of
$$
\dot{\kappa }_\upsilon (\rho )=H_{i\widetilde{G}_\upsilon }(\kappa _\upsilon (\rho )),\ \rho
\in \Lambda _0\hbox{ (identifying }H_{i\widetilde{G}_\upsilon }\simeq J\widehat{H}_{\widetilde{G}_\upsilon }).
$$
In this way $\kappa _\upsilon $ is defined in a $1/{\cal O}(1)$-neighborhood of
$\Lambda _0$ for $g$ and almost holomorphic at
$\Lambda _0$. $\kappa _\upsilon \in \dot{S}(1)$, $\partial _\upsilon ^k\kappa _\upsilon \in
S(1)$ for $k\ge 1$.

Write $G_\upsilon ^\ell =\widetilde{G}_\upsilon $ and let $G_\upsilon ^r$ be the almost
holomorphic function at $\Lambda _0$ which is given by
\begin{equation}\label{nore.22}
G_\upsilon ^\ell \circ \kappa _\upsilon =G_\upsilon ^r.
\end{equation}
We have
$$
\partial _\upsilon ^kG_\upsilon ^\ell,\, \partial _\upsilon ^kG_\upsilon ^r\in S(\widetilde{r}R)
\hbox{ for }k\ge 0.
$$

Then on $\Lambda _{\upsilon G}$:
$$
H_{G_\upsilon ^\ell}=\left(\kappa _\upsilon  \right)_*H_{G_\upsilon ^r},
$$
where $(\kappa _\upsilon )_*$ denotes the operation of push forward of vector fields.

\par Let ${\cal G}_\upsilon ^r$, ${\cal G}_\upsilon ^\ell$ be
pseudodifferential operators of order $\widetilde{r}R$ associated to
$\Lambda _0$, $\Lambda _{\upsilon G}$ with principal symbols
$G_\upsilon ^r$ and $G_\upsilon ^\ell$ respectively. We can also
assume that $\partial _\upsilon ^k {\cal G}_\upsilon ^r$ is a
pseudodifferential operator of order $\widetilde{r}R$ for all $k$.
Then we have elliptic Fourier integral operators $U_\upsilon $,
$\widetilde{U}_\upsilon $ of order $1$ associated to
$\kappa _\upsilon $, such that
\begin{equation}\label{nore.23}
hD_\upsilon U_\upsilon +iU_\upsilon {\cal G}_\upsilon ^r=K_\upsilon ^r,\ U_0=1,
\end{equation}
\begin{equation}\label{nore.24}
hD_\upsilon \widetilde{U}_\upsilon +i{\cal G}_\upsilon ^\ell\widetilde{U}_\upsilon =K_\upsilon ^\ell,\ \widetilde{U}_0=1,
\end{equation}
where $K_\upsilon ^r$, $\partial _\upsilon ^kK_\upsilon ^r$,
$K_\upsilon ^\ell$, $\partial _\upsilon ^kK_\upsilon ^\ell$ are negligible operators of order
$\widetilde{r}R$. This is a straight forward WKB-solution of Cauchy
problems within the framework of \cite{HeSj86}. Now replace
${\cal G}_\upsilon ^r$ with
${\cal G}_\upsilon ^r+iU_\upsilon ^{-1}K_\upsilon ^r$ and
${\cal G}_\upsilon ^\ell$ with
${\cal G}_\upsilon ^\ell+iK_\upsilon ^\ell \widetilde{U}_\upsilon
^{-1}$ and notice that $U_\upsilon ^{-1}K_\upsilon ^r$ and $K_\upsilon
^\ell \widetilde{U}_\upsilon ^{-1}$ are negligible
of order 1 with all their $\upsilon $-derivatives. Then we get,
\begin{equation}\label{nore.25}
hD_\upsilon U_\upsilon +iU_\upsilon {\cal G}_\upsilon ^r= 0,\ U_0=1,
\end{equation}
\begin{equation}\label{nore.26}
hD_\upsilon \widetilde{U}_\upsilon +i{\cal G}_\upsilon ^\ell\widetilde{U}_\upsilon = 0,\ \widetilde{U}_0=1.
\end{equation}
If we choose first ${\cal G}_\upsilon ^r$, $U_\upsilon $ in
(\ref{nore.25}) and then determine ${\cal G}_\upsilon ^\ell$ by
\begin{equation}\label{nore.27}
{\cal G}_\upsilon ^\ell U_\upsilon =U_\upsilon {\cal G}_\upsilon ^r
\end{equation}
(in formal agreement with Egorov's theorem and (\ref{nore.22})),
we get $\widetilde{U}_\upsilon =U_\upsilon $ in (\ref{nore.26}):
\begin{equation}\label{nore.28}
hD_\upsilon U_\upsilon +i{\cal G}_\upsilon ^\ell U_\upsilon =0,\ U_0=1.
\end{equation}

\par Using also (\ref{nore.31}) below, we get
\begin{equation}\label{nore.28.5}
(hD_\upsilon )^k{\cal G}_\upsilon ^\ell =U_\upsilon \left(hD_\upsilon  -i\mathrm{ad\,}_{{\cal G}_\upsilon ^r}
\right)^k({\cal G}_\upsilon ^r)U_\upsilon ^{-1},
\end{equation}
which shows that for every $k\ge 0$, $\partial _\upsilon ^k{\cal G}_\upsilon ^\ell$ is
the sum of a pseudodifferential operator and a negligible operator of
order $\widetilde{r}R$. Here $\mathrm{ad}_A(B)$ denotes the commutator $[A,B]$.
\par Let $P$ be an $h$-differential operator of order
  $m=m_0(x)(\widetilde{r}/r)^{N_0}$ as in (\ref{opesc.5})--(\ref{opesc.13}), so that $P$ is also
an $h$-pseudodifferential operator
\begin{equation}\label{nore.29}
P:\, H(\Lambda _{\upsilon G},m)\to H(\Lambda _{\upsilon G},1)
\end{equation}
with principal symbol ${{p}_\vert}_{\Lambda _{\upsilon G}}$ as in
(\ref{nore.19}). Here we also assume that the
  coefficients of $P$ are analytic in a neighborhood of the $x$-space
  projection of $\mathrm{supp\,}G$.
The study of $P$ in (\ref{nore.29}) is equivalent to
that of
\begin{equation}\label{nore.30}
V_\upsilon PU_\upsilon =:P_\upsilon :H(\Lambda _0,m)\to H(\Lambda _0,1),\hbox{ where }V_\upsilon =U_\upsilon ^{-1}.
\end{equation}
We will often write $H(\Lambda _{\upsilon G})=H(\Lambda _{\upsilon G},1)$. Notice that
$$
(Pu|v)_{H(\Lambda _{\upsilon G})}=(P_\upsilon V_\upsilon u|V_\upsilon v),
$$
if we define the norm on $H(\Lambda _{\upsilon G})$ by
\begin{equation}
\label{nore'.31}
\|v\|_{H(\Lambda _{\upsilon G})}=\|V_\upsilon v\|_{L^2},
\end{equation}
making the operators $U_\upsilon :L^2\to H(\Lambda _{\upsilon G})$ and $V_\upsilon :
H(\Lambda _{\upsilon G}) \to L^2$ unitary. This norm is uniformly
equivalent to the one in (\ref{nore.14}).
\begin{remark}\label{nore'3} Let $\Omega \Subset {\bf R}^n$ be open
  and assume that $G(x,\xi )=0$ whenever $x\in \widetilde{\Omega }$,
  where $\widetilde{\Omega }$ is a neighborhood of $\overline{\Omega
  }$. We can choose first the formal pseudodifferential operator
  part of ${\cal G}^r_\upsilon $ with symbol equal to zero over
  $\widetilde{\Omega }$. Then formally, $U_\upsilon $ is a Fourier
  integral operator equal to 1 on $L^2(\widetilde{\Omega })$. It
  follows from the way Fourier integral operators are defined in
  {\rm \cite{HeSj86}}, that we can choose a realization of $U_\upsilon $
  (that we denote with the same symbol) such that
\begin{equation}
\label{nore'.32}
U_\upsilon u=u,\hbox{ when }u\in C_0^\infty (\Omega ).
\end{equation}
As before, let $V_\upsilon =U_\upsilon ^{-1}$. Applying $V_\upsilon $
to {\rm (\ref{nore'.32})}, we get
\begin{equation}
\label{nore'.33}
V_\upsilon u=u,\hbox{ when }u\in C_0^\infty (\Omega ).
\end{equation}
After that we modify ${\cal G}^\ell_\upsilon $, ${\cal G}^r_\upsilon $
with negligible terms as above, so that {\rm (\ref{nore.25})},
{\rm (\ref{nore.28})} hold. From {\rm (\ref{nore'.31})}, we now get
\begin{equation}
\label{nore'.34}
\| v \|_{H(\Lambda _{\upsilon G})}=\| v\|_{L^2},\ v\in C_0^\infty
(\Omega ).
\end{equation}

\end{remark}

\par
From (\ref{nore.25}), we first notice that
\begin{equation}\label{nore.31}
hD_\upsilon V_\upsilon -i{\cal G}_\upsilon ^r V_\upsilon =0,
\end{equation}
and then that
\begin{equation}\label{nore.32}
h\partial _\upsilon P_\upsilon =[P_\upsilon ,{\cal G}_\upsilon ^r].
\end{equation}

We already know that $P_\upsilon =P_\upsilon '+N_\upsilon $ where $P'_\upsilon $, $N_\upsilon $ are continuous
in $\upsilon $ with values in the pseudodifferential and negligible operators
respectively, of order $m$. See the statements 1--3 after
Theorem 7.2 in \cite{HeSj86}. Write (\ref{nore.32}) as
$$
\left( \partial _\upsilon +\frac{1}{h}\mathrm{ad}_{{\cal G}_\upsilon ^r} \right) P_\upsilon =0,
$$
which implies,
$$
\left( \partial _\upsilon +\frac{1}{h}\mathrm{ad}_{{\cal G}_\upsilon ^r} \right)^k
P_\upsilon =0,\ k=1,2,...
$$
From this we deduce that $\partial _\upsilon ^kP_\upsilon $ has the same
structure. From Taylor's formula with integral remainder, we get
$$
P_\upsilon =P_{\upsilon ,k}+N_{\upsilon ,k}
$$
for every $k\in {\bf N}$, where $\upsilon \mapsto P_{\upsilon ,k}$, $\upsilon \mapsto N_{\upsilon ,k}$
are of class $C^k$ with values in the pops of order $m$ and nops of
order $m$ respectively.

\par On the other hand, since the machineries are based on the
(complex) method of stationary phase, we also know that the Weyl
symbols of $P_\upsilon $ and $P_{\upsilon ,k}$ are of the form
\begin{equation}\label{nore.33}
\sim \sum_{0}^\infty h^jp_j(\upsilon ,x,\xi ),
\end{equation}
where $p_j\in S(m/(\widetilde{r}R)^j)$ are independent of $k$ and
therefore smooth in $\upsilon $. We conclude that
$P_\upsilon =P_\upsilon '+N_\upsilon $, where $P'_\upsilon $,
$N_\upsilon $ are smooth in $\upsilon $ with values in the pops and
nops respectively, of order $m$.

The equation for $p_0(\upsilon ,x,\xi )=p_\upsilon (x,\xi )$ is
$$
\partial _\upsilon p_\upsilon =iH_{G_\upsilon ^r}p_\upsilon ,\ p_{\upsilon =0}=\hbox{the principal symbol of }P.
$$
We recover the fact (already known by Egorov's theorem) that
\begin{equation}\label{nore.34}
p_\upsilon (\rho )=p(\kappa _\upsilon (\rho ))=:\widetilde{p}_\upsilon .
\end{equation}
Indeed, the two symbols are equal when $\upsilon =0$ and
\[
\begin{split}
\partial _\upsilon \widetilde{p}_\upsilon (\rho )&=\langle \dot{\kappa }_\upsilon (\rho
),dp(\kappa _\upsilon (\rho ))\rangle = i\langle (\kappa _\upsilon )_*
H_{G_\upsilon ^r},dp(\kappa _\upsilon (\rho ))\rangle\\
&=\langle H_{G_\upsilon ^r},\kappa _\upsilon ^*(dp(\kappa _\upsilon (\rho )))\rangle=
i\langle H_{G_\upsilon ^r},d(p\circ \kappa _\upsilon (\rho ))\rangle =iH_{G^r_\upsilon }\widetilde{p}_\upsilon .
\end{split}
\]

\par From the construction of $\kappa _\upsilon $ prior to (\ref{nore.22}), we
see that
\begin{equation}\label{nore.35}
\kappa _\upsilon (\rho )=\rho +i\upsilon H_G(\rho )+{\cal O}(\upsilon ^2),
\end{equation}
where the remainder is ${\cal O}(\upsilon ^2)$ as a smooth function of $\upsilon $
with values in $S(1)$ (with respect to the metric $g$).
Using this in (\ref{nore.34}), we get
\begin{equation}\label{nore.36}
p_\upsilon (\rho )=p(\rho )-i\upsilon H_pG+{\cal O}(\upsilon ^2m)
\end{equation}
in the sense of smooth functions $\mathrm{neigh\,}(0,{\bf R})\ni
\upsilon \mapsto S(m)$.

\par From (\ref{nore.35}) we get
\begin{equation}\label{nore.37}
\widetilde{\rho }:=\Re \kappa _\upsilon (\rho )=\rho +{\cal O}(\upsilon ^2),\hbox{ so
  that by (\ref{nore.21.5}) }\kappa _\upsilon (\rho )=\widetilde{\rho }+i\upsilon H_G(\widetilde{\rho })
\end{equation}
and hence,
\begin{equation}\label{nore.38}
p_\upsilon (\rho )=p(\widetilde{\rho }+i\upsilon H_G(\widetilde{\rho })).
\end{equation}

\par If $G=G_s\in S(\widetilde{r}R)$ is real and depends smoothly on
$s\in \mathrm{neigh\,}(0,{\bf R})$, then the smooth dependence on $s$
diffuses into the whole construction above and we get (with the
obvious notation) that
$P_{\upsilon ,s}:=V_{\upsilon ,s}PU_{\upsilon ,s}$ in (\ref{nore.30}) is a smooth function of
$(\upsilon ,s)$ with values in the pops+nops of order $m$. (Recall
that we sometimes abbreviate: pop$=$pseudodifferential operator, nop$=$negli\-gible operator.)

\par Let $\Pi $ be the orthogonal projection $L^2(\Lambda _0,M_0)\to
TL^2({\bf R}^n)$ (cf.\ Remark \ref{nore3}) whose properties were recalled in
(\ref{nore.15})--(\ref{nore.18}). Combining the above properties of
$P_{\upsilon ,s}$ with Proposition 7.3 in \cite{HeSj86}, we get
\begin{equation}\label{nore.39}
TP_{\upsilon ,s}T^{-1}\Pi =\Pi P^\mathrm{top}_{\upsilon ,s}\Pi +N_{\upsilon ,s}
\end{equation}
where $N_{\upsilon ,s}$ is smooth in $(\upsilon ,s)$ with values in the nops of order
$m$ and
\begin{equation}\label{nore.40}
P_{\upsilon ,s}^\mathrm{top}(\rho ;h)\sim \sum_0^\infty  h^k p^k_{\upsilon ,s}(\rho
),\ \rho \in \Lambda _0,
\end{equation}
in $C^\infty (\mathrm{neigh\,}(0,0),S(m))$ and with the
general term in the sum belonging to $C^\infty
(\mathrm{neigh\,}(0,0),h^kS(m/(\widetilde{r}R)^k))$. Here, as already
recalled in (\ref{nore.20}),
\begin{equation}\label{nore.41}
p_{\upsilon ,s}^0=p_{\upsilon ,s}
\end{equation}
is the principal symbol of $P_{\upsilon ,s}$.

\par From (\ref{nore.39}) we infer that
\begin{multline}
\label{nore.42}
(PU_{\upsilon ,s}u|U_{\upsilon ,s}v)_{H(\Lambda _{\upsilon G_s})} =
(P_{\upsilon ,s}u|v) \\ =\int_{\Lambda _0}P^\mathrm{top}_{\upsilon ,s}(\rho ;h)Tu(\rho
)\cdot \overline{Tv(\rho )}L_0(d\rho )+(N_{\upsilon ,s}u|v),
\end{multline}
for $u,v\in H(\Lambda _0,m^{1/2})$ (cf.\ Remark \ref{nore3} and
(\ref{nore'.31})). This can be expressed in the coordinates
$\widetilde{\rho }$ in (\ref{nore.37}). Here the scalar product in the
middle is the one of $L^2({\bf R}^n)$. The Jacobian satisfies
\begin{equation}\label{nore.43}
J_{\upsilon ,s}(\widetilde{\rho }):=\frac{d\rho }{d\widetilde{\rho }}=1+{\cal
  O}(\upsilon ^2) \hbox{ in }S(1)
\end{equation}
and is a smooth function of $\upsilon ,s$. We can write
\begin{equation}\label{nore.44}
P_{\upsilon ,s}^\mathrm{top}(\rho
;h)=\widetilde{P}_{\upsilon ,s}^\mathrm{top}(\widetilde{\rho };h),
\end{equation}
so (\ref{nore.42}) becomes
\begin{equation}\label{nore.45}
(P_{\upsilon ,s}u|v)=\int_{\Lambda
  _0}\widetilde{P}^\mathrm{top}_{\upsilon ,s}(\widetilde{\rho
};h)\widetilde{T}u(\widetilde{\rho
})\cdot \overline{\widetilde{T}v(\widetilde{\rho })}J_{\upsilon ,s}(\widetilde{\rho
})L_0(d\widetilde{\rho })+(N_{\upsilon ,s}u|v),
\end{equation}
where $\widetilde{T}u(\widetilde{\rho }):=Tu(\rho
)$. $\widetilde{P}^\mathrm{top}_{\upsilon ,s}(\widetilde{\rho };h)$ has an
asymptotic expansion as in (\ref{nore.40}) with
$\widetilde{p}^k_{\upsilon ,s}(\widetilde{\rho })=p^k_{\upsilon ,s}(\rho )$ and the
advantage with (\ref{nore.45}) is that
$\widetilde{p}_{\upsilon ,s}=\widetilde{p}^0_{\upsilon ,s}$ satisfies
\begin{equation}\label{nore.46}
\widetilde{p}_{\upsilon ,s}(\widetilde{\rho })=p(\widetilde{\rho
}+i\upsilon H_{G_s}(\widetilde{\rho })).
\end{equation}
All this remains valid if we replace the single parameter $s$ by
$s=(s_1,...,s_k)\in \mathrm{neigh}(0,{\bf R}^k)$.

\par If $p$ is real-valued on $\Lambda _0$, we get
\begin{equation}\label{nore.47}
\begin{split}
\Im \widetilde{p}_{\upsilon ,0}(\widetilde{\rho })&=\upsilon H_{G_0}(p)+{\cal O}(m\upsilon ^3\|
H_{G_0}\|_g^3)\\
&=-\upsilon H_p(G_0)+{\cal O}(m\upsilon ^3\| H_{G_0}\|_g^3).
\end{split}
\end{equation}

We summarize the results in this section.
\begin{prop}\label{nore4}
Let $P$ be an $h$-differential operator of order $m(x,\xi
)=m_0(x)(\widetilde{r}/r)^{N_0}$ as in
{\rm (\ref{opesc.5})}--{\rm (\ref{opesc.13})}. Let $G\in S(\widetilde{r}R)$
satisfy {\rm (\ref{nore.7.5})} and assume that the coefficients of $P$ are
analytic in a neighborhood of the $x$-space projection of
$\mathrm{supp\,}(G-g_0)$. Then for $0\le \upsilon \le 1$, $P:H(\Lambda
_{\upsilon G},m)\to H(\Lambda _{\upsilon G},1)$ is the sum of an
$h$-pop and a nop both of order $m$, depending smoothly on $\upsilon $. The
principal symbol is equal to ${{p}_\vert}_{\Lambda _{\upsilon G}}$.

\par We can find a canonical transformation
$\kappa _\upsilon :\Lambda _0\to \Lambda _{\upsilon G}$ of class $\dot{S}(1)$ for the
metric $g$, depending smoothly on $\upsilon \in [0,1]$ in that class,
satisfying {\rm (\ref{nore.35})} and an operator
$U_\upsilon :H(\Lambda _0,1)\to H(\Lambda _{\upsilon G},1)$ of the form $U'_\upsilon +N_\upsilon $,
where $U'_\upsilon $ is an elliptic Fourier integral operator of order 1
associated to $\kappa _\upsilon $ and $N_\upsilon $ is a nop of order 1, with
$U_0=\mathrm{id}$, such that $U_\upsilon ^{-1}=:V_\upsilon =V_\upsilon '+M_\upsilon $ has the
analogous properties (with $\kappa _\upsilon $ replaced with
$\kappa _\upsilon ^{-1}$, such that $P_\upsilon :=V_\upsilon PU_\upsilon $ has the following
properties:
\begin{itemize}
\item $P_\upsilon $ is the sum of a pop and a nop of order $m$, both
  depending smoothly on $\upsilon $ in the corresponding spaces of operators.
\item The principal symbol of $P_\upsilon $ is given by {\rm (\ref{nore.38})},
  {\rm (\ref{nore.37})}.
\item Writing the Weyl symbol of $P_\upsilon $ as $\sim \sum_0^\infty
  h^jp_j(\upsilon ,x,\xi )$, we have
  $$
  \mathrm{supp\,}(p_j(\upsilon ,\cdot )-p_j(0,\cdot ))\subset \mathrm{supp\,}G,\quad j\ge 0
  $$.
\item We have the Toeplitz representation {\rm (\ref{nore.39})},
  {\rm (\ref{nore.40})}, {\rm (\ref{nore.42})} (without the parameters $s$ for the
  moment), where the leading symbol in {\rm (\ref{nore.40})} is equal to the
  one of $P_\upsilon $ as a pseudodifferential operator, i.e. $p_0(\upsilon ,x,\xi )$.
\end{itemize}

\par When $G$ depends smoothly on additional parameters $s\in
\mathrm{neigh\,}(0,{\bf R}^k)$ we have the corresponding smooth
dependence of all terms above.

\par When $G\in S(m_G)$ satisfies the more special condition
{\rm (\ref{nore.7.7})}, we can choose $U_\upsilon $ so that $P_\upsilon -P$ is of order
$mm_G/(\widetilde{r}R)$.
More precisely, $\partial
  _\upsilon P_\upsilon \sim \sum_0^\infty h^j\partial _\upsilon p_j$ in
  $S(mm_G/(\widetilde{r}R))$, $\partial _\upsilon p_j\in
  S(mm_G/(\widetilde{r}R)^{j+1})$. A similar statement holds for
  $P^\mathrm{top}_{\upsilon ,s}$, $\widetilde{P}^\mathrm{top}_{\upsilon ,s}$ and we
  here retain that $\partial _t(P^{\mathrm{top}}_{\upsilon ,s}-p_{\upsilon ,s})\in S(hmm_G/(\widetilde{r}R)^2)$.

\end{prop}

The extension in the last paragraph of the proposition follows from an
inspection of the proofs.

\section{Semi-boundedness in $H(\Lambda_{\upsilon G})$ spaces}\label{sbd}
\setcounter{equation}{0}

We continue the discussion from the preceding section and work with
the representations (\ref{nore.45}), (\ref{nore.46}), where we drop
the tildes until further notice. From (\ref{nore.46}), we get
$$
\partial _sp_{\upsilon ,s}(\rho )=i\upsilon \langle H_{\partial _sG_s}(\rho ),dp(\rho
+i\upsilon H_{G_s}(\rho ))\rangle ,
$$
which can also be written
\begin{equation}\label{sbd.1}
\partial _sp_{\upsilon ,s}(\rho )=-i\upsilon \langle H_p(\rho +i\upsilon H_{G_s}(\rho ) ),
d\partial _sG_s(\rho )\rangle ,
\end{equation}
If $p\in S(m)$ is real-valued on $\Lambda _0$, we get by Taylor
expansion ``in the $g$-metric'',
\begin{equation}\label{sbd.2}
\partial _s\Im p_{\upsilon ,s}(\rho )=-\upsilon \langle H_p(\rho ),d\partial
_sG_s(\rho )\rangle +{\cal O}\left( \upsilon ^3 \frac{m}{\widetilde{r}R}\|
  H_{G_s}\|_g^2 \|d\partial _sG_s\|_{g^*} \right).
\end{equation}
Here the factor $m/(\widetilde{r}R)$ corresponds to the estimate
(\ref{opesc.27.5}) and $g^*$ is the dual metric to $g$:
\begin{equation}\label{sbd.3}
g^* =(\widetilde{r}d\xi )^2+(Rdx)^2,
\end{equation}
so
$$
\| df\|^2_{g^*}=\widetilde{r}^2|\partial _\xi f |^2+R^2 |\partial
_xf|^2={\cal O}(\widetilde{r}^2R^2),
$$
when $f\in \dot{S}(\widetilde{r}R)$. Hence, if we assume a uniform bound on
$\partial _sG_s$ in $S(\widetilde{r}R)$, (\ref{sbd.2}) simplifies
to
\begin{equation}\label{sbd.4}
\partial _s\Im p_{\upsilon ,s}(\rho )=-\upsilon H_p(\partial _sG_s)(\rho )+{\cal
  O}(m\upsilon ^3\| H_{G_s}\|_g^2).
\end{equation}
In this formula we can take $s=(s_1,...,s_k)$ close to $0$ in ${\bf
  R}^k$ and replace $\partial _s$ with $\partial _{s_k}$:
$$
\partial _{s_k}\Im p_{\upsilon ,s}(\rho )=-\upsilon H_p(\partial _{s_k}G_s)(\rho )+{\cal
  O}(m\upsilon ^3\| H_{G_s}\|_g^2).
$$
Taylor expansion at $s=0$ gives,
\begin{equation}\label{sbd.5}
\partial _{s_k}\Im p_{\upsilon ,s}(\rho )=-\upsilon H_p(\partial
_{s_k}G_s)_{s_k=0}(\rho )+{\cal O}(m\upsilon s_k)+{\cal O}(m\upsilon ^3|s|)+{\cal O}(m\upsilon ^3\|H_{G_0}\|_g^2).
\end{equation}
Writing $s=(s',s_k)$ and integrating (\ref{sbd.5}) from $0$ to $s_k$,
gives
\begin{equation}\label{sbd.6}
\begin{split}
&-\Im p_{\upsilon ,s}(\rho )+\Im p_{\upsilon ,(s',0)}(\rho )=\\
&\upsilon s_k H_p{(\partial _{s_k})}_{s_k=0}G_s+{\cal O}(m\upsilon s_k^2)+{\cal
  O}(m\upsilon ^3|s|s_k)+
{\cal O}(m\upsilon ^3s_k\|H_{G_0}(\rho )\|_g^2).
\end{split}
\end{equation}

We now consider the situation in Proposition \ref{opesc2}. Let
$$
K_0\supset K_1\supset K_2\supset ...
$$
be a sequence of compact $H_p$-convex sets in $\Sigma _{[-\epsilon _0,\epsilon_0]}$ that
contain the trapped set. We choose $K_j$ so that $K_{j+1}$ is
contained in the interior of $K_j$. For $j\in {\bf N}$, let $\chi _j\in C_0^\infty
(]-\epsilon _0,\epsilon _0[;[0,1])$ be equal to 1 on $[-\epsilon
_0/2,\epsilon _0/2]$ and such that $\chi _{j+1}=1$ on
$\mathrm{supp\,}\chi _j$.

\par Let $G_0$ be a modification of $G$ on a bounded subset of $\Sigma
_{[-\epsilon _0,\epsilon_0]}$ as ``$\widetilde{G}$'' in Proposition \ref{opesc2} with $\widetilde{K}$
there equal to $K_0$. Let $G_j$ be constructed similarly with
$\widetilde{K}$ equal to $K_j$ in such a way that $H_pG_{j+1}>0$ in
$\Sigma _{[-\epsilon _0,\epsilon_0]}\setminus K_j$. This implies that $H_pG_{j+1}>0$
on $\mathrm{supp\,}G_j$.

\par Let $G^0_j=\chi _j(p)G_j$. Since $H_pG_j^0=\chi _j(p)H_pG_j$, we
get
\begin{equation}\label{sbd.7}
H_pG_{j+1}^0\ge \frac{m}{{\cal O}(1)}\hbox{ on }\mathrm{supp\,}G_j^0.
\end{equation}
Let
$$
G^{(N)}=G_0^0+hG_1^0+...+h^NG_N^0
$$
and notice that this enters into the framework of Section \ref{nore}:
$$
G^{(N)}=G_s=G_0^0+s_1G_1^0+...+s_NG_N^0,\ s=(h,h^2,...,h^N).
$$
We apply (\ref{nore.45}) (with the tildes dropped since the beginning
of this section) with $\upsilon >0$ small, and recall that we have
(\ref{nore.40}), (\ref{nore.46}) (with the tildes dropped). When going
from $G^{(k-1)}$ to $G^{(k)}$, the leading term $p_{\upsilon ,s}$ changes
according to (\ref{sbd.6}). Writing
$p_\upsilon ^{(k)}=p_{\upsilon ,(h,..,h^k,0,..,0)}$, we get
\begin{equation}\label{sbd.8}
\begin{split}
&-\Im p^{(k)}_\upsilon +\Im p^{(k-1)}_\upsilon \\
&= \upsilon h^kH_p(G_k^0)+{\cal O}(m\upsilon h^{2k})+{\cal O}(m\upsilon ^3h^{k+1})+{\cal
  O}(m\upsilon ^3h^k\| H_{G_0^0}\|_g^2)\\
&=\upsilon h^kH_p(G_k^0)+{\cal O}(m\upsilon h^{k+1})+{\cal O}(m\upsilon ^3h^k\|H_{G_0^0}\|_g^2).
\end{split}
\end{equation}
Also, by (\ref{nore.47}),
\begin{equation}\label{sbd.9}
\begin{split}
-\Im p_{\upsilon ,0}&=\upsilon H_p(G_0^0)+{\cal O}(m\upsilon ^3\| H_{G_0^0}\|_g^3)\\
&=\upsilon (1+{\cal O}(\upsilon ^2))H_p(G_0^0),
\end{split}
\end{equation}
where we also used (\ref{opesc.57}) for $G_0=G^0_0$. Using that
estimate also for the last term in (\ref{sbd.8}) and summing over $k$, we get
\begin{equation}\label{sbd.10}\begin{split}
-\Im p^{(N)}_\upsilon =&\upsilon \left( 1+{\cal O}^{(0)}(\upsilon ^2)+{\cal
  O}^{(1)}(\upsilon ^2h)+...+{\cal O}^{(N)}(\upsilon ^2h^N)\right) H_p(G_0^0)\\
&+ \upsilon \left( hH_pG_1^0+h^2H_pG_2^0+...+h^NH_pG_N^0\right)\\
&+\left( m\upsilon {\cal O}^{(1)}(h^2)+m\upsilon {\cal O}^{(2)}(h^3)+...+m\upsilon {\cal
  O}^{(N)}(h^{N+1}) \right) .
\end{split}
\end{equation}
Here ${\cal O}^{(k)}(\cdot )$ denotes a term which depends on
$G_0^0,...,G_k^0$ but not on $G_{k+1}^0,...$ and whose support is
contained in that of $G_k^0$. We see that after successive
replacements, $G_j^0\mapsto \alpha _jG_j^0$ with $\alpha _j>0$ large
enough, we can achieve that
\[
\begin{split}
&h^2H_pG_2^0+m\upsilon {\cal O}^{(1)}(h^2)\ge 0,\\
&h^3H_pG_3^0+m\upsilon {\cal O}^{(2)}(h^3)\ge 0,\\
&\ \ \ ....
\end{split}
\]
\begin{equation}\label{sbd.11}
-\Im p^{(N)}_\upsilon \ge \upsilon (1+{\cal O}(\upsilon ^2))H_p(G_0^0)-{\cal O}(\upsilon mh^{N+1}).
\end{equation}

\par Now recall (\ref{nore.40}) (after adding and removing the
tildes), where $p^k_{\upsilon ,s}$ are real for $\upsilon =0$. Taylor
  expand each $p_{\upsilon ,s}^k$ to sufficiently high order at $s=0$ and take
$s=(h,h^2,...,h^N)$. Then we get with
$P_\upsilon ^{\mathrm{top},(N)}=P^{\mathrm{top}}_{\upsilon ,(h,..,h^N)}$,
\begin{multline}\label{sbd.12}
-\Im P^{\mathrm{top},(N)}_\upsilon (\rho ;h)=\\
-\Im p^{(N)}_\upsilon +{h\cal
  O}^{(0)}(m\upsilon )+h^2{\cal O}^{(1)}(m\upsilon )+...+h^{N+1}{\cal O}^{(N)}(m\upsilon ).
\end{multline}
Here the factors ${\cal O}^{(j)}$ belong to
  $S(m\upsilon )$. They are independent of $h$ for $j\le N-1$.
By successive replacements $G_j^0\mapsto \alpha _jG_j^0$, we can
achieve, using (\ref{sbd.10}), that
\begin{equation}\label{sbd.13}
-\Im P^{\mathrm{top},(N)}_\upsilon \ge \upsilon (1+{\cal O}(\upsilon ^2))H_p(G_0^0)-{\cal O}(\upsilon mh^{N+1}).
\end{equation}
Hence, by (\ref{nore.45}),
\begin{equation}\label{sbd.14}
\begin{split}
-\Im (P^{(N)}_\upsilon u|u)\ge & \int_{\Lambda _0} \upsilon (1+{\cal
  O}(\upsilon ^2))H_p(G^0_0)\| Tu(\rho )\|^2 J_{\upsilon ,(h,..,h^N)}L_0(d\rho )\\
&-{\cal O}(\upsilon )h^{N+1}\| u\|^2_{H(\Lambda _0,m^{1/2} )}.
\end{split}\end{equation}

\par The replacement $G_j^0\mapsto \alpha _jG_j^0$ does not depend on the value of $N\ge
j$, so we get a full sequence $G_0^0,G_1^0,...$. Consider an asymptotic sum
\begin{equation}\label{sbd.15}
G^0\sim \sum_0^\infty  G_j^0h^j \hbox{ in }S(\widetilde{r}R).
\end{equation}
Then for every $N\ge 1$, $G^0=G^{(N)}+h^{N+1}\widetilde{G}^0_{N+1}$,
$\widetilde{G}^0_{N+1}\in S(\widetilde{r}R)$
and if $P_\upsilon =V_\upsilon PU_\upsilon $ (with the natural definitions of $U_\upsilon $ and
$V_\upsilon =U_\upsilon ^{-1}$) we have the analogue of (\ref{nore.45}) (now with the
tildes dropped),
\begin{equation}\label{sbd.16}
(P_\upsilon u|v)=\int_{\Lambda _0}P_\upsilon ^{\mathrm{top}}(\rho ;h)Tu(\rho ;h)\cdot
\overline{Tv(\rho ;h)}J_\upsilon (\rho ;h)L_0(d\rho )+(N_\upsilon u|v).
\end{equation} Here $N_\upsilon $ is negligible of order $m$, $\Im P^{\mathrm{top}}_\upsilon $ and
$\Im N_\upsilon $ vanish for $\upsilon =0$. We can replace $P_\upsilon ^{(N)}$ with $P_\upsilon $ in
(\ref{sbd.14}) and the discussion leading to that estimate shows that
\begin{equation}\label{sbd.16.5}
0<J_\upsilon (\rho ;h)=1+{\cal O}(\upsilon ^2),\ \ -\Im P_\upsilon ^{\mathrm{top}}\ge
\upsilon (1+{\cal O}(\upsilon ^2))H_pG_0^0-{\cal O}(m\upsilon h^\infty ).
\end{equation}
In particular,
\begin{equation}\label{sbd.18}
-\Im (P_\upsilon u|u)\ge -\upsilon {\cal O}(h^\infty )\| u\|^2_{H(\Lambda _0,m^{1/2})}.
\end{equation}
Recall that $P_\upsilon $ is just a reduction to $H(\Lambda _0)$ of the
restriction to $H(\Lambda _{\upsilon G^0},m)$ of $P$, so with the norm and
scalar product on
$H(\Lambda _{\upsilon G^0})$ induced by $U_\upsilon $, we get
\begin{equation}\label{sbd.19}
-\Im (Pu|u)_{H(\Lambda _{\upsilon G^0},1)}\ge -\upsilon {\cal O}(h^\infty )\|u\|^2_{H(\Lambda _{\upsilon G^0},m^{1/2})},
\end{equation}
for $u\in H(\Lambda _{\upsilon G^0},m)$.
\begin{remark}\label{sbd1}
Only
${m_|}_{\Sigma _{[-\epsilon _0,\epsilon_0]}}$ matters in the
calculations. Especially, in the Schr\"odinger case (discussed in
Section {\rm \ref{int}}), we have $m=\langle \xi \rangle^2$, so we can replace $m$ with $1$.
\end{remark}

We end this section with an observation about decoupling of the
exterior and the interior part in certain situations. Let $G^0$ be as in
(\ref{sbd.15}). Let $\Omega \subset {\bf
  R}^n$ be a bounded open set such that
\begin{equation}\label{sbd.20}
\overline{\Omega }\cap \pi _x(\mathrm{supp\,}G^0)=\emptyset ,\ \pi
_x(x,\xi )=x.
\end{equation}
We have seen in Remark \ref{nore'3} that we can choose the Fourier
integral operators $U_\upsilon $, $V_\upsilon $ in Section \ref{nore}
so that (\ref{nore'.32})--(\ref{nore'.34}) hold for $G=G^0$:
\begin{equation}\label{sbd.21}
\|v\|_{H(\Lambda _{\upsilon G^0})} = \|v\|_{L^2},\ v\in C_0^\infty
(\Omega ).
\end{equation}
It follows that
$$
-\Im (Pu|u)_{H(\Lambda _{\upsilon G^0})}=-\Im (\widetilde{P}u|u) _{H(\Lambda _{\upsilon G^0})}
$$
if $\widetilde{P}$ is a new formally self-adjoint operator (with
respect to $L^2({\bf R}^n)$) such that
$\mathrm{supp\,}(\widetilde{P}-P)\subset \Omega $, where
$\mathrm{supp\,}(\widetilde{P}-P)$ is defined to be the union of the
supports of the coefficients of $\widetilde{P}-P$. In particular, we
may then replace $P$ with $\widetilde{P}$ in (\ref{sbd.19}).
}

\section{Far away improvement}\label{far}
\setcounter{equation}{0}
In this section we discuss improvements in the semi-bound estimates,
when the escape function is supported far away. We let $P$, $m$, $r$,
$\widetilde{r}$, $R$ be as in Section \ref{opesc} with the following
special choices,
\begin{equation}\label{far.1}
r=1,\ R(x)=\langle x\rangle,\ m_0(x)=1,
\end{equation}
implying
\begin{equation}\label{far.2}
m(x,\xi )=\langle \xi \rangle^{N_0},\ \widetilde{r}(x,\xi )=\langle
\xi \rangle .
\end{equation}
With $p(x,\xi )$ still denoting the semi-classical principal symbol,
we assume that $p(x,\xi )\to p_\infty (\xi )\in S(m)$ as $x\to \infty
$ in the following  sense: For all $\alpha ,\beta \in {\bf
  N}^n$,
\begin{equation}\label{far.3}
\partial _x^\alpha \partial _\xi ^\beta (p(x,\xi )-p_\infty (\xi
))=o(1)m(\xi )R(x)^{-|\alpha |}\widetilde{r}(x,\xi )^{-|\beta |},\
x\to \infty ,
\end{equation}
uniformly with respect to $\xi $.

\par If we assume the existence of an escape function in $\Sigma
_{[-\epsilon _0,\epsilon_0]}$ for $\epsilon _0>0$ small enough, then
\begin{equation}\label{far.4}
p_\infty (\xi )=0\ \Rightarrow\ \partial _\xi p_\infty (\xi )\ne 0.
\end{equation}
From the ellipticity assumption (\ref{opesc.14}) we know in addition
to (\ref{opesc.15}), (\ref{opesc.16}) that $p_\infty ^{-1}(0)$ is
bounded. Conversely, if we assume (\ref{far.4}), then
$$
G(x,\xi )=x\cdot \frac{\partial _\xi p_\infty (\xi )}{\langle \xi
  \rangle^{N_0-2}}
\in S(R\widetilde{r})
$$
has the required properties. Indeed, when $x\to \infty $,
$$
H_pG(x,\xi )\to H_{p_\infty }G=\frac{(\partial _\xi p_\infty
  )^2}{\langle \xi \rangle ^{N_0-2}}\asymp m \hbox{ on }p_\infty ^{-1}(0).
$$
In the classical Schr\"odinger operator case, this gives $G(x,\xi
)=2x\cdot \xi $, which up to the factor $2$ is the escape function
appearing in standard complex scaling.

\par We study the situation in a domain $|x|>\mu /{\cal O}(1)$, for
$\mu \gg 1$, and eventually we will choose our escape function $G^0$
with its support contained in such domains. It is
natural to make the change of variables, $x=\mu \widetilde{x}$, so
that $|\widetilde{x}|>1/{\cal O}(1)$.

\par Consider first the principal symbol. Put
\begin{equation}\label{far.5}
p_\mu (\widetilde{x},\widetilde{\xi })=p(\mu
\widetilde{x},\widetilde{\xi })=p\circ \kappa _\mu
(\widetilde{x},\widetilde{\xi }),\end{equation}
where
\begin{equation}\label{far5.5}
\kappa _\mu
(\widetilde{x},\widetilde{\xi })=(\mu \widetilde{x},\widetilde{\xi
}),\end{equation}
\begin{equation}\label{far.5.5}
\kappa _\mu ^*\sigma =\mu \sigma .
\end{equation}
Then in any region, $|\widetilde{x}|>1/{\cal O}(1)$ we have
\begin{equation}\label{far.5.7}
\partial _{\widetilde{x}}^\alpha \partial _{\widetilde{\xi }}^\beta
p_\mu (\widetilde{x},\widetilde{\xi })={\cal O}(1)m(\widetilde{\xi })\widetilde{r}^{-|\beta
  |}\widehat{R}(\widetilde{x})^{-|\alpha |},
\end{equation}
\begin{equation}\label{far.6}
\partial _{\widetilde{x}}^\alpha \partial _{\widetilde{\xi }}^\beta
(p_\mu (\widetilde{x},\widetilde{\xi })-p_\infty (\widetilde{\xi
}))=o(1)m(\widetilde{\xi })\widetilde{r}^{-|\beta
  |}\widehat{R}(\widetilde{x})^{-|\alpha |},\ \mu \to \infty
\end{equation}
Here $\widehat{R}=\widehat{R}_\mu $
is given by
\begin{equation}\label{far.7}
\widehat{R}(\widetilde{x})=\frac{R(\mu \widetilde{x})}{\mu },
\end{equation}
so that
$$
\widehat{R}(\widetilde{x})\asymp R(\widetilde{x}),\
|\widetilde{x}|>1/{\cal O}(1).
$$

\par We restrict the attention to a region,
\begin{equation}\label{far.8}
\Sigma _{\mu ,\epsilon _0}=p_\mu ^{-1}([-\epsilon _0,\epsilon _0]).
\end{equation}
\begin{prop}\label{far1}
The ``balls'' $\pi _{\widetilde{x}}^{-1}B(0,r_0)\cap \Sigma _{\mu
  ,\epsilon _0}$ are $H_{p_\mu }$-convex for $r_0\ge 1/{\cal O}(1)$
when $\mu$ is large enough. More precisely, every $H_{p_\mu }$-trajectory in
$\Sigma _{\mu ,\epsilon _0}$ can visit such a ball only during at most one time interval
which can be finite or infinite.
\end{prop}
\begin{proof}
It suffices to check that
$$
H_{p_\mu }^2(\widetilde{x}^2/2)>0,\ \ (\widetilde{x},\widetilde{\xi
})\in \Sigma _{\mu ,\epsilon _0},\ |\widetilde{x}|\ge 1/{\cal O}(1).
$$
With a somewhat simplified notation, we get
\[
\begin{split}
H_{p_\mu }^2\left(\frac{\widetilde{x}^2}{2}\right)
&= \left( \frac{\partial p_\mu }{\partial \widetilde{\xi }}\cdot
  \frac{\partial }{\partial \widetilde{x}}-
\frac{\partial p_\mu }
{\partial \widetilde{x}}\cdot \frac{\partial }{\partial
  \widetilde{\xi }} \right)\left(\frac{\partial p_\mu }{\partial
  \widetilde{\xi }}\cdot \widetilde{x} \right)
\\
&= \left(\frac{\partial p_\mu }{\partial \widetilde{\xi }}
\right)^2+\frac{\partial p_\mu }{\partial \widetilde{\xi }}\cdot
\frac{\partial ^2p_\mu }{\partial \widetilde{x}\partial \widetilde{\xi
  }}\cdot \widetilde{x}-\frac{\partial p_\mu }{\partial
  \widetilde{x}}\cdot \frac{\partial ^2p_\mu }{\partial \widetilde{\xi
  }^2}\cdot \widetilde{x}\\
&\to \left(\frac{\partial p_\infty  }{\partial \widetilde{\xi }}
\right)^2>0,\ \ \mu \to \infty .
\end{split}
\]
\end{proof}

\par From this proposition and (\ref{far.17}) below, it will follow that the ``balls'' $\pi
_x^{-1}(B(0,\mu ))\cap \Sigma _{\epsilon_0} $ are $H_p$ convex for
$\mu $ large enough.

\par We next apply the change of variables $x=\mu \widetilde{x}$ to
the operator $P$ in (\ref{opesc.5}). We get,
\begin{equation}\label{far.9}
P(x,hD_x;h)=P(\mu
\widetilde{x},\widetilde{h}D_{\widetilde{x}};h)=:P_\mu
(\widetilde{x},\widetilde{h}D_{\widetilde{x}};\widetilde{h}),\ \
\widetilde{h}=\frac{h}{\mu }.
\end{equation}
More explicitly, in view of (\ref{opesc.5}), (\ref{opesc.6}):
\begin{equation}\label{far.10}
  P_\mu (\widetilde{x},\widetilde{h}D_{\widetilde{x}};\widetilde{h})=
  \sum_{|\alpha |\le N_0}a_\alpha ^\mu
  (\widetilde{x};\widetilde{h})(\widetilde{h}D_{\widetilde{x}})^\alpha ,
\end{equation}
where,
\begin{equation}\label{far.11}
a_\alpha ^\mu (\widetilde{x};\widetilde{h})=a_{\alpha }(\mu
\widetilde{x};h)
=\sum_{k=0}^{N_0-|\alpha |}h^ka_{\alpha ,k}(\mu
\widetilde{x})=\sum_{k=0}^{N_0-|\alpha |}\widetilde{h}^ka^\mu _{\alpha
,k}(\widetilde{x}).
\end{equation}
Here
\begin{equation}\label{far.12}
a_{\alpha ,k}^\mu =\mu ^ka_{\alpha ,k}(\mu \widetilde{x})\in
S(\widehat{R}^{-k}),\ |\widetilde{x}|\ge 1/{\cal O}(1),
\end{equation}
and $\widehat{R}\asymp R(\widetilde{x})$ as in (\ref{far.7}). This
means that $P_\mu $ satisfies the general assumptions for $P$ in the
region, $|\widetilde{x}|\ge 1/{\cal O}(1)$ and we have the analogue of
(\ref{opesc.8}):
\begin{equation}\label{far.13}
P_\mu (\widetilde{x},\widetilde{\xi };\widetilde{h})=p_{0,\mu
}(\widetilde{x},\widetilde{\xi })+\widetilde{h}p_{1,\mu
}(\widetilde{x},\widetilde{\xi })+...+h^{N_0}p_{N_0,\mu
}(\widetilde{x},\widetilde{\xi }),
\end{equation}
$p_{0,\mu }(\widetilde{x},\widetilde{\xi })=p_\mu
(\widetilde{x},\widetilde{\xi })$,
\begin{equation}\label{far.14}
p_{j,\mu }(\widetilde{x},\widetilde{\xi })\in
S(m(\widetilde{r}R)^{-j}),\ \ |\widetilde{x}|\ge 1/{\cal O}(1),
\end{equation}

\par We next check that $\Lambda _{\upsilon G}$ scales naturally when $G$ is
an escape function. We expect the scaled weight $G_\mu $ to satisfy,
$$
e^{\upsilon G(x,\xi )/h}=e^{\upsilon G_\mu (\widetilde{x},\widetilde{\xi
  })/\widetilde{h}},\ (x,\xi )=\kappa _\mu
(\widetilde{x},\widetilde{\xi }),
$$
i.e.\ $G_\mu (\widetilde{x},\widetilde{\xi })=G(x,\xi )/\mu $, so we
define:
\begin{equation}\label{far.15}
G_\mu (\widetilde{x},\widetilde{\xi })=\frac{1}{\mu }G(\mu
\widetilde{x},\widetilde{\xi })=\frac{1}{\mu }(G\circ \kappa _\mu
)(\widetilde{x},\widetilde{\xi }).
\end{equation}
We have,
\begin{equation}\label{far.16}
\Lambda _{\upsilon G_\mu }=\kappa _\mu ^{-1}(\Lambda _{\upsilon G}).
\end{equation}
Indeed, for $(\widetilde{x},\widetilde{\xi })\in \Lambda _{\upsilon G_\mu }$,
we have
\[
\begin{split}
&\mu \Im \widetilde{x}=\mu \upsilon \partial _{\widetilde{\xi }}G_\mu (\Re
\widetilde{x},\Re \widetilde{\xi })=\upsilon \partial _{\xi
}G(\mu \Re \widetilde{x},\Re \widetilde{\xi }),\\
&\Im \widetilde{\xi }=-\upsilon \partial _{\widetilde{x}}G_\mu (\Re
\widetilde{x},\Re \widetilde{\xi })=-\upsilon \partial _xG(\mu \Re
\widetilde{x},\Re \widetilde{\xi }),
\end{split}
\]
which shows that $(x,\xi )\in \Lambda _{\upsilon G}$ if $(x,\xi )=\kappa _\mu
(\widetilde{x},\widetilde{\xi })$.

\par In the same spirit, we observe that
\begin{equation}\label{far.17}
(\kappa _\mu )_*H_{p_\mu }=\mu H_p.
\end{equation}

\par We finally apply the natural scaling
$$
\Lambda _{\upsilon G}\ni \alpha \mapsto \widetilde{\alpha }=\kappa _\mu
^{-1}(\alpha )\in \Lambda _{\upsilon G_\mu }
$$
to $Tu$ in (\ref{nore.4}), (\ref{nore.6}). Starting from
(\ref{nore.6}), we put $u(y)=\widetilde{u}(\widetilde{y})$, where
$y=\mu \widetilde{y}$. Again, with $\widetilde{h}=h/\mu $, we get from
(\ref{nore.2}):
\begin{equation}\label{far.18}
\begin{split}
  \frac{1}{h}\phi (\alpha ,y)&=\frac{1}{\widetilde{h}}\widetilde{\phi
  }(\widetilde{\alpha },\widetilde{y}),\hbox{ where }\\
  \widetilde{\phi }(\widetilde{\alpha
  },\widetilde{y})&=(\widetilde{\alpha }_x-\widetilde{y})\cdot
  \widetilde{\alpha }_\xi +i\lambda _\mu (\widetilde{\alpha}
  )(\widetilde{\alpha }_x-\widetilde{y})^2
\end{split}
\end{equation}
and
\begin{equation}\label{far.19}
\lambda _\mu (\widetilde{\alpha })=\mu \lambda (\mu \widetilde{\alpha
}_x,\widetilde{\alpha }_\xi )\in S(\widetilde{r}(\widetilde{\alpha
})R(\widetilde{\alpha }_x)^{-1})
\end{equation}
for $|\widetilde{\alpha }_x|\ge 1/{\cal O}(1)$, where we also used
that $\mu /R(\mu \widetilde{\alpha }_x)\asymp 1/R(\widetilde{\alpha
}_x)$. With the same changes of variables in (\ref{nore.6}), we get
$$Tu(\alpha ;h)=\int e^{\frac{i}{\widetilde{h}}\widetilde{\phi } (\widetilde{\alpha}
  ,\widetilde{y})}\mu ^n{\bf t}(\mu \widetilde{\alpha
}_x,\widetilde{\alpha }_\xi, \mu \widetilde{y};\mu \widetilde{h}) \chi
\left(\frac{\widetilde{y}-\Re\widetilde{\alpha
    }_x}{\widehat{R}(\widetilde{\alpha }_x)}
\right)\widetilde{u}(\widetilde{y})d\widetilde{y},$$
still with $\widehat{R}\asymp R$ as in (\ref{far.7}), so the cutoff is
the naturally scaled one. The new amplitude
$$
\widetilde{{\bf t}} (\widetilde{\alpha },\widetilde{y};\widetilde{h})=\mu ^n{\bf t}(\mu
\widetilde{\alpha }_x,\widetilde{\alpha }_\xi ,\mu
\widetilde{y}; \mu \widetilde{h}),
$$
belongs to
$$
S\left(\mu ^nh^{-\frac{3n}{4}}\widetilde{r}(\widetilde{\alpha }_\xi
)^{\frac{n}{4}}R(\mu \widetilde{\alpha }_x)^{-\frac{n}{4}}\right)=
S\left( \widetilde{h}^{-\frac{3n}{4}} \widetilde{r}(\widetilde{\alpha
  }_\xi )^{\frac{n}{4}} \widehat{R}(\widetilde{\alpha }_x)^{-\frac{n}{4}} \right),
$$
which is the right symbol class (working still in $|\widetilde{\alpha
}_x|\ge 1/{\cal O}(1)$).

Furthermore,
\[
\begin{split}
  \left|\det \begin{pmatrix}\widetilde{{\bf t}} &\partial
      _{\widetilde{ y}_1}\widetilde{{\bf t}} &... &\partial
      _{\widetilde{ y}_n}\widetilde{{\bf t}}\end{pmatrix}\right| &=\mu
  ^{(n+1)n+n}\left|\det \begin{pmatrix}{\bf t} &\partial _{y_1}{\bf t}
      &... &\partial _{y_n}{\bf t}\end{pmatrix}\right|\\
  &\asymp
  \widehat{R}^{-n}\left(\widetilde{h}^{-\frac{3n}{4}}\widetilde{r}^{\frac{n}{4}}
  \widehat{R}^{-\frac{n}{4}}\right)^{n+1},
\end{split}
\]
which is analogous to (\ref{nore.4}).

\par In conclusion, for $|\alpha _x|\ge \mu /{\cal O}(1)$, we have
\begin{equation}\label{far.20}\begin{split}
    &Tu(\alpha ;h)=\widetilde{T}\widetilde{u}(\widetilde{\alpha
    };\widetilde{h}), \Lambda _{\upsilon G}\ni \alpha =\kappa _\mu
    (\widetilde{\alpha }),\
    \widetilde{\alpha }\in \Lambda _{\upsilon G_\mu },
\\
    &\widetilde{u}(\widetilde{y})=u(y),\ y=\mu \widetilde{y},\
    \widetilde{h}=h/\mu ,
\end{split}
\end{equation}
where $\widetilde{T}$ has all the general properties of an
FBI-transform in any fixed region $|\widetilde{\alpha } _x|\ge 1/{\cal
  O}(1)$.

\par If the two order functions $m$ and $\widetilde{m}$ are related by
\begin{equation}\label{far.21}
\widetilde{m}=m\circ \kappa _\mu ,
\end{equation}
which is fulfilled under the assumptions (\ref{far.1}),
  (\ref{far.2}), when $\widetilde{m}=\langle \widetilde{\xi }\rangle^{N_0}$,
then we can define the Sobolev spaces $H(\Lambda _{\upsilon G},m)$, $H(\Lambda
_{\upsilon G_\mu },\widetilde{m})$ as in Section \ref{hlg}, by
(\ref{nore.14}), with $G$ replaced by $\upsilon G$, and its analogue,
\begin{equation}\label{far.22}
\| \widetilde{u}\|_{H(\Lambda _{\upsilon G_\mu },\widetilde{m})}= \|
\widetilde{T}\widetilde{u}\|_{L^2(\Lambda
  _{\upsilon G_\mu },\widetilde{m}^2e^{-2\upsilon H_\mu /\widetilde{h}}d\widetilde{\alpha
  })}.
\end{equation}
Here we use $H$ in (\ref{nore.14}) adapted to $G$ (cf.\
(\ref{nore.7})) so that $\upsilon H$ is adapted to $\upsilon G$,
define $H_\mu $ by the analogous
relation and notice that
  $\widetilde{H}_\mu (\widetilde{x},\widetilde{\xi })=\mu ^{-1}H(\mu
  \widetilde{x},\widetilde{\xi })$. ($G$ and $H$ also depend on $\mu $
  and we use the subscript $\mu $ to indicate when we work in the
  scaled variables $(\widetilde{x},\widetilde{\xi })$.) If we extend the definition of $\kappa _\mu $ to maps: ${\bf
  R}^n\to {\bf R}^n$ by putting $\kappa _\mu (\widetilde{y})=\mu
\widetilde{y}$, and let $\kappa _\mu ^*$ denote right composition with
$\kappa _\mu $ in the usual way, then (\ref{far.20}) tells us that
$$
\kappa _\mu ^*\circ T=\widetilde{T}\circ \kappa _\mu ^*.
$$
Moreover, $d\alpha =\mu ^nd\widetilde{\alpha }$, $dy=\mu
^nd\widetilde{y}$, so
\begin{equation}\label{far.23}\begin{split}
  \| u\|_{L^2}^2&=\mu ^n\| \widetilde{u}\|_{L^2}^2,\\ \|
  Tu\|^2_{L^2(\Lambda _{\upsilon G},m^2e^{-2\upsilon H/h}d\alpha )}&=\mu ^n
  \|\widetilde{T}\widetilde{u}\|_{L^2(\Lambda
    _{\upsilon G_\mu },\widetilde{m}^2e^{-2\upsilon H_\mu/h}d\widetilde{\alpha })} .
\end{split}
\end{equation}
Thus,
\begin{equation}\label{far.24}
\| u\|_{H(\Lambda _{\upsilon G},m)}=\mu ^{\frac{n}{2}} \|
\widetilde{u}\|_{H(\Lambda _{\upsilon G_\mu },\widetilde{m})}.
\end{equation}
This applies in particular to the spaces $H(\Lambda _{\upsilon G})$,
$H(\Lambda _{\upsilon G_\mu })$.

Now we apply the discussion at the end of Section \ref{sbd} to the
operator $P_\mu $, whose symbol properties we have verified in any
region $|\widetilde{x}|\ge 1/{\cal O}(1)$. In view of Proposition
\ref{far1} we have the strictly decreasing sequence of $H_{p_\mu
}$-convex sets in $\Sigma _{\mu ,\epsilon _0}$:
\begin{equation}\label{far.25}
\widetilde{K}_j=\pi _{\widetilde{x}}^{-1}(\overline{B(0,r_j)})\cap
\Sigma _{\mu ,\epsilon _0},\ j=0,1,2,...
\end{equation}
where $3/2>r_0>r_1>r_2,...\searrow 1$, $j\to \infty $. This gives rise
to a weight $G_\mu ^0$, vanishing over a neighborhood of $B(0,1/2)$, for which we would have
\begin{equation}\label{far.26}
-\Im (P_\mu \widetilde{u}|\widetilde{u})_{H(\Lambda
  _{\upsilon G_\mu ^0})}\ge -\upsilon {\cal O}(\widetilde{h}^{\infty })\|
\widetilde{u}\|^2_{H(\Lambda _{\upsilon  G^0_\mu },\widetilde{m}^{1/2})},
\end{equation}
had it been true that $P_\mu $ is a differential operator of the
right symbol class also inside a region $|\widetilde{x}|\le 1/{\cal
  O}(1)$. However, this symbol property is guaranteed only outside
such balls, but according to the observation at the end of Section
\ref{sbd}, we can choose the $H(\Lambda _{\upsilon G_\mu ^0})$ norm, so that
$$
-\Im (P_\mu u|u)_{H(\Lambda _{\upsilon G^0_\mu })}=-\Im
(\widetilde{P}_\mu u|u)_{H(\Lambda _{\upsilon G^0_\mu })},
$$
whenever $\mathrm{supp\,}(\widetilde{P}_\mu -P_\mu )$ is contained in
some small fixed neighborhood of $B(0,1/2)$. Moreover we can find such
a $\widetilde{P}_\mu $ satisfying (\ref{far.26}). Hence (\ref{far.26})
holds for $P_\mu $. (It suffices to take
$\widetilde{P}_\mu =(1-\chi )P_\mu (1-\chi )$, where
$\chi \in C_0^\infty (B(0,1/2);{\bf R})$ is equal to 1 on $B(0,1/3)$.)
\par
In order to get the corresponding estimates for $P$, we put
$$
G^0=\mu G_\mu ^0\circ \kappa _\mu ^{-1}:\ G^0(x,\xi )=\mu G_\mu ^0(x/\mu ,\xi ).
$$
Then (\ref{far.26}) gives,
\begin{equation}\label{far.28}
-\Im (P u|u)_{H(\Lambda
  _{\upsilon G^0})}
\ge -\upsilon {\cal O}((h/\mu )^\infty  )\| u\|^2_{H(\Lambda _{\upsilon G^0},m^{1/2})}.
\end{equation}
By Remark \ref{sbd1}, we can replace
$m$ by 1 in the Schr\"odinger case. In (\ref{far.26}), the
semi-classical parameter is $\widetilde{h}=h/\mu $ while in
(\ref{far.28}) we are back to using $h$.

\section{Resolvent estimates}\label{rest}
\setcounter{equation}{0}

\par To fix the ideas, we assume right away that $n=1$ and that
\begin{equation}\label{rest.1}
P=-h^2\partial _x^2+V(x),\ V\in C^\infty ({\bf R};{\bf R}).
\end{equation}
We adopt the general assumptions of Section \ref{far}. More precisely,
we assume (\ref{far.1}): $r=1$, $R(x)=\langle x\rangle$, $m_0(x)=1$
and (\ref{far.2}) with $N_0=2$ so that $m(x,\xi )=\langle \xi
\rangle^2$. Recall that $\widetilde{r}(x,\xi )=\langle \xi \rangle$. We
also assume (\ref{far.3}) with $p(x,\xi )=\xi ^2+V(x)$, $p_\infty (\xi )=\xi ^2$, which amounts
to
\begin{equation}\label{rest.2}
\partial _x^\alpha V(x)=o(1)\langle x\rangle ^{-|\alpha |},\ x\to
\infty .
\end{equation}
We also assume dilation analyticity near $\infty $:
\begin{equation}\label{rest.3}\begin{split}
&V\hbox{ has a holomorphic extension to }
\{x\in {\bf C};\, \Re x>C,\,\,
|\Im x|<|\Re x|/C\}\\
&\hbox{and denoting the extension also by }V, \hbox{ we have
}V(x)=o(1),\\
&\hbox{when }x\to \infty \hbox{ in the truncated sector above.}
\end{split}
\end{equation}
The earlier discussion was focused on the energy level $E=0$. Here we
will apply it with $P$ replaced by $P-E$ for $E\in [E_-,E_+]$ for
$0<E_-<E_+<+\infty $. In other terms we will mainly work in
\begin{equation}\label{rest.4}
p^{-1}([E_-,E_+]),\ 0<E_-<E_+<+\infty ,
\end{equation}
and the slight difference with the earlier discussion is that we now
take a wider energy range $[E_-,E_+]$ instead of $[-\epsilon
_0,\epsilon _0]$.

\par Assume that for a choice $E\in ]E_-,E_+[$,
\begin{equation}\label{rest.5}\begin{split}
&\hbox{the }H_p\hbox{-flow is non trapping in every unbounded}\\ &\hbox{connected
component of } p^{-1}(E). \end{split}
\end{equation}
Later, we shall strengthen this assumption to non-trapping in
$p^{-1}(E)$. Using the special structure of the symbol
$p(x,\xi )=\xi ^2+V(x)$, we see that every unbounded connected
component $\Sigma '_E$ of $p^{-1}(E)$ is a simple
smooth integral curve
$\gamma:\ {\bf R}\ni t\mapsto \gamma (t)=(x(t),\xi (t))$ of $H_p$ with
one of the following 4 properties:
\begin{itemize}
\item[1)] $t\xi (t)>0$, $x(t)\to +\infty $, when $|t|\to \infty $,
\item[2)] $t\xi (t)<0$, $x(t)\to -\infty $, when $|t|\to \infty $,
\item[3)] $\xi (t)>0$, $x(t)\to \pm\infty $, when $t\to \pm\infty $,
\item[4)] $\xi (t)<0$, $x(t)\to \mp\infty $, when $t\to \pm\infty $.
\end{itemize}
Moreover, the union $\Sigma _E$ of all unbounded
components of $p^{-1}(E)$ is the union of two
different components as above where either
\medskip\par\noindent
I) One is of type 1) and the other is of type 2),

\medskip
or

\medskip\par\noindent
II) one is of type 3) and the other is of type 4)

\medskip

\par Using a sequence of cutoffs, $\chi _j(p)$, $\chi _j\in C_0^\infty
(]E_-,E_+[)$, where
$$1_{[E_-+\delta ,E_+-\delta ]}\prec \chi _0\prec \chi _1\prec ...,$$
a corresponding sequence of escape functions $G_0, G_1, G_2, ...$ with
$G_0\prec G_1\prec G_2\prec ...$ and a dilation $\widetilde{x}\mapsto
\mu \widetilde{x}$, $\mu \ge 1$, we obtain as in Sections \ref{sbd},
\ref{far} a function $G^0=G_\mu ^0$ of class $S(\widetilde{r}R)$,
uniformly with respect to $\mu $, with support in $\Sigma
_{[E_-,E_+]}\cap \{ (x,\xi );\, |x|\ge \mu /(2C) \}$ such that we have
the semi-boundedness property (\ref{far.28}) for $0\le \upsilon \ll 1$, $\mu
\ge 1$ and
\begin{equation}\label{rest.7}
H_pG^0\ge \chi _0(p)/C,\hbox{ on }\{ |x|\ge \mu /C \} .
\end{equation}
(The only difference with Sections \ref{sbd}, \ref{far} is that we have
replaced $[-\epsilon _0,\epsilon _0]$ with $[E_-,E_+]$ which is quite
straight forward in the Schr\"odinger case.) Since we shall next turn
to resolvent estimates with more escape functions, it is convenient to
rename $G^0$:
\begin{equation}\label{rest.8}
G_{\mathrm{sbd}}^\mu =G_{\mathrm{sbd}}:=G^0.
\end{equation}

For the resolvent estimates, we need to supply a suitable escape
function in the set $\Sigma _E\cap\{ |x|\le \mu /C \}$ and to merge it to $G^\mu
_{\mathrm{sbd}}$. Here we assume that $E\in [E_-+\delta ,E_+-\delta ]$.

\par First we can find an escape function
\begin{equation}\label{rest.9}
G\in C^\infty (\Sigma _E;{\bf R})
\end{equation}
of class $S(\widetilde{r}R)$ such that
\begin{equation}\label{rest.10}H_pG>0,\end{equation}
\begin{equation}\label{rest.11}G(x,-\xi )=-G(x,\xi ),\end{equation}
\begin{equation}\label{rest.12}G(x,\xi )=x\cdot \xi ,\ |x|\gg 1.\end{equation}
Observe that, since we are in the 1D case,
\begin{equation}\label{rest.13}
\langle G(x,\xi ) \rangle\asymp \langle x\rangle,
\end{equation}
so $\langle G\rangle$ and $\langle x\rangle$ are equivalent weights.

\par Let $0\le \vartheta \ll 1$, to be fixed small enough. Let $f=f_\vartheta \in
C^\infty ({\bf R};{\bf R})$ be given by
\begin{equation}\label{rest.14}
f(0)=0,\ f'(t)=\frac{h}{\langle t\rangle^{1+\vartheta }}.
\end{equation}
Then $f$ is odd, and when $\vartheta >0$, we have
\begin{equation}\label{rest.15}
f(t)=h(\pm C_\vartheta +{\cal O}(\langle t\rangle^{-\vartheta }),\ t\to \pm\infty.
\end{equation}
Here,
$$
C_\vartheta =\int_0^{+\infty }\frac{1}{\langle t\rangle^{1+\vartheta }}dt.
$$
When $\vartheta =0$, we get
\begin{equation}\label{rest.16}
f(t)=h(\ln t +{\cal O}(1)),\ t\to +\infty  .
\end{equation}
Because of the unboundedness in this case, we assume from now on that
$\vartheta >0$.

Define the (new) function $G^0$ by
\begin{equation}\label{rest.17}
G^0=f(G).
\end{equation}
Then,
\begin{equation}\label{rest.18}
H_pG^0=f'(G)H_pG\asymp \frac{h}{\langle G\rangle ^{1+\vartheta }}\asymp \frac{h}{\langle x\rangle ^{1+\vartheta }}.
\end{equation}
Also, $H_{G^0}=f'(G)H_G=\frac{h}{\langle G\rangle^{1+\vartheta }}H_G$ and
recalling that $\|H_G\|_g={\cal O}(1)$, we get
\begin{equation}\label{rest.19}
\|H_{G^0}\|_g=\frac{{\cal O}(h)}{\langle x\rangle^{1+\vartheta }}.
\end{equation}
This also follows from (\ref{rest.20}) below.
\begin{prop}\label{rest1} We have
\begin{equation}\label{rest.20}G^0\in \dot{S}\left(\frac{h}{\langle
      x\rangle^\vartheta } \right),\end{equation}
\begin{equation}\label{rest.21}
G^0\in S(h).
\end{equation}
\end{prop}
\begin{proof}
For $k\ge 1$, we have
$$
f^{(k)}(t)={\cal O}(h)\langle t\rangle^{-\vartheta -k}
$$
and for $(\alpha ,\beta )\in {\bf N}^2\setminus {0}$, we can write
$\partial _x^\alpha \partial _\xi ^\beta G^0$ as a finite linear
combination of terms
\begin{equation}\label{rest.22}
f^{(k)}(G)\left(\partial _x^{\alpha _1}\partial _\xi ^{\beta _1}G\right)... \left(\partial _x^{\alpha _k}\partial _\xi ^{\beta _k}G\right),
\end{equation}
where $k\ge 1$, $(\alpha _j,\beta _j)\ne (0,0)$, $\alpha _1+...+\alpha
_k=\alpha $, $\beta _1+...+\beta _k=\beta $. Since $\langle
G\rangle\asymp \langle x\rangle$, it follows that the term in
(\ref{rest.22}) is
$$
{\cal O}(1)\frac{h}{\langle x\rangle^{\vartheta +k}}\langle
x\rangle^{1-\alpha _1}...\langle x\rangle^{1-\alpha _k}=\frac{{\cal
    O}(h)}{\langle x\rangle^{\vartheta +\alpha }}
$$
and (\ref{rest.20}) follows. Now (\ref{rest.21}) follows from
(\ref{rest.20}) and the fact that $G^0={\cal O}(h)$ by (\ref{rest.15}).
\end{proof}

\par Until further notice, we assume that
\begin{equation}\label{rest.22.5}
\hbox{The }H_p\hbox{ flow is non-trapping on }p^{-1}(E).
\end{equation}
In other words, $\Sigma _E$ is equal to all of
$p^{-1}(E)$. Let $\chi ^0\in C_0^\infty (\mathrm{neigh\,}(E,{\bf
  R});[0,1])$ be equal to 1 near $E$ and with its support contained in
$\chi _0^{-1}(1)$, where $\chi _0,\chi _1,...$ are the cutoffs used
before. Put
\begin{equation}\label{rest.23}
G_{\mathrm{lap}}=\chi ^0(p)G^0\in S(h).
\end{equation}
Here ``lap'' stands for ``limiting absorption principle'', because
$G_{\mathrm{lap}}$ can be used to give a quick proof of the
semi-classical limiting absorption principle of Robert--Tamura \cite{RoTa84}, also
proved by G\'erard--Martinez \cite{GeMa88}. This is
also related to Martinez' result \cite{Ma02} on the absence of
resonances for non-trapping potentials that are merely smooth on some
bounded set. We put
\begin{equation}\label{rest.24}
G_\epsilon =G_{\mathrm{lap}}+\epsilon G_{\mathrm{sbd}}, \ 0<\epsilon
\ll h,
\end{equation}
where we recall that $G_{\mathrm{sbd}}$ depends on a large parameter
$\mu $. We next choose $\mu $ as a function of $\epsilon $, and to do
so we notice that when the support of $\chi ^0$ is narrow enough,
\begin{equation}\label{rest.25}
H_pG_{\mathrm{lap}}\asymp\chi ^0(p)\frac{h}{\langle G\rangle^{1+\vartheta
  }}\asymp \chi ^0(p)\frac{h}{\langle x\rangle^{1+\vartheta }},
\end{equation}
which is $\asymp h/\langle x\rangle^{1+\vartheta }$ where $\chi _0(p)=1$
and in particular in $p^{-1}([E-\delta_0 ,E+\delta_0 ])$ when $\delta
_0>0$ is small enough. On the other
hand $H_p(\epsilon G_\mathrm{sbd})$ is ${\cal O}(\epsilon )$ and of
order of magnitude $\epsilon $ in $\{ \chi _0(p)=1 \}\cap \{ |x|\ge
\mu /C \}$, by (\ref{rest.7}). Accepting a loss due to the positivity of $\vartheta $, we
choose $\mu $ so that
$$
\frac{h}{\mu }=\epsilon ,
$$
i.e.
\begin{equation}\label{rest.26}
\mu =\frac{h}{\epsilon }\gg 1.
\end{equation}

\par Then,
\begin{equation}\label{rest.27}
G_\epsilon \in S(h+\epsilon \langle x\rangle),
\end{equation}
and in particular,
\begin{equation}\label{rest.28}
\|H_{G_\epsilon }\|_g={\cal O}\left( \frac{h}{\langle x
\rangle}+\epsilon  \right) .
\end{equation}

\par In addition to (\ref{rest.25}), (\ref{rest.7}), we know that
\begin{equation}\label{rest.29}
H_pG_{\mathrm{sbd}}\ge 0.
\end{equation}
Since $\chi ^0\prec \chi _0$, it follows that
\begin{equation}\label{rest.30}
H_pG_\epsilon \gtrsim \chi ^0(p)\left(\frac{h}{\langle
    x\rangle^{1+\vartheta }}+\epsilon 1_{\{|x|\ge \mu/C  \}} \right)
\end{equation}
to be compared with the upper bound, that follows from (\ref{rest.27})
and the fact that $\mathrm{supp\,}G_\epsilon \subset p^{-1}([E_-,E_+])
$:
\begin{equation}\label{rest.31}
H_pG_\epsilon ={\cal O}(1)\left(\frac{h}{\langle x\rangle}+\epsilon  \right),
\end{equation}
where
\begin{equation}\label{rest.32}
\frac{h}{\langle x\rangle}+\epsilon \asymp \frac{h}{\langle
  x\rangle}+\epsilon 1_{\{ |x|\ge \mu/C  \}}.
\end{equation}
\begin{remark}\label{rest2}
We define the spaces $H(\Lambda _{\upsilon \epsilon  G_{\mathrm{sbd}}})$, $H(\Lambda
_{\upsilon G_\epsilon })$ as in {\rm (\ref{nore.14})}, with $G$
replaced by $\upsilon \epsilon  G_{\mathrm{sbd}}$ and $\upsilon G_\epsilon $ respectively. Since $G_{\mathrm{lap}}\in S(h)$ by
{\rm (\ref{rest.23})}, we see that $H^{\upsilon G_\epsilon
}-H^{\upsilon \epsilon G_{\mathrm{sbd}}}={\cal O}(\upsilon h)$, where $H^{\upsilon G_\epsilon }$,
$H^{\upsilon \epsilon G_{\mathrm{sbd}}}$ are defined in {\rm (\ref{nore.7})}, with $G$
replaced by $\upsilon G_\epsilon $ and $\upsilon \epsilon  G_\mathrm{sbd}$ respectively. We conclude that
\begin{equation}
\label{rest.33}
\| u\| _{H(\Lambda _{\upsilon G_\epsilon })}\asymp \| u\|_{H(\Lambda_{\upsilon \epsilon  G_\mathrm{sbd}})},
\end{equation}
uniformly with respect to $\upsilon $, $h$, $u$, when $0\le \upsilon
\le \upsilon _0\ll 1$,
$h\le h\le h_0\ll 1$.
\end{remark}

\par We now apply Proposition \ref{nore4} and get
\begin{equation}\label{rest.34}
\Im (Pu|u)_{H(\Lambda _{\upsilon G_\epsilon })}=\int_{\Lambda _{\upsilon G_\epsilon
  }}P_\upsilon ^{\mathrm{top}}Tu\cdot \overline{Tu} e^{-2\upsilon
  H_\epsilon/h}d\alpha
+(N_\upsilon u|u)_{H(\Lambda _{\upsilon G_\epsilon })},
\end{equation}
where we have preferred the more invariant integration on $\Lambda
_{\upsilon G_\epsilon }$ rather than to reduce everything to $\Lambda
_0$. $H_\epsilon $ is defined as in (\ref{nore.7}) with $G$ replaced
by $G_\epsilon $. Here (cf.\ (\ref{rest.27}))
\begin{equation}\label{rest.35}
P_\upsilon ^{\mathrm{top}}=p_\upsilon ^{\mathrm{top}}+{\cal O}(1)\frac{\upsilon h(h+\epsilon
  \langle x\rangle )}{\langle x\rangle^2}=p_\upsilon ^{\mathrm{top}}+{\cal
  O}(1)\frac{\upsilon h}{\langle x\rangle} \left(\frac{h}{\langle x\rangle}
  +\epsilon  \right),
\end{equation}
\begin{equation}\label{rest.36}\begin{split}
p_\upsilon ^{\mathrm{top}}(\vartheta )&=\Im p(\Re \vartheta +i\upsilon H_{G_\epsilon }(\Re \vartheta ))\\
&=-\upsilon H_pG_\epsilon (\Re \vartheta )+{\cal O}(\upsilon ^3)\|H_{G_\epsilon }\|_g^3\\
&=-\upsilon H_pG_\epsilon (\Re \vartheta )+{\cal O}(\upsilon ^3)\left(\frac{h}{\langle
    x\rangle}+\epsilon  \right)^3.
\end{split}
\end{equation}
Further, $N_\upsilon $ is negligible of order $\upsilon $.

\par Here we notice that by (\ref{rest.32})
\begin{equation}\label{rest.37}
\frac{h}{\langle x\rangle^{1+\vartheta }}+\epsilon 1_{\{ |x|\ge \mu/C
  \}}\gtrsim \left( \frac{h}{\langle x\rangle^{1+\vartheta
    }}+\epsilon_\vartheta
\right)=:m_\epsilon (x;h),\ \ \epsilon _\vartheta :=\left(\frac{\epsilon }{h}  \right)^\vartheta \epsilon.
\end{equation}
(In the limiting case, $\epsilon =h/C$, where $C>0$ is a
large constant, $\mu $ is of the order of a large constant and $\epsilon
_\vartheta \asymp h \asymp m_\epsilon $.)

Thus if we fix $\vartheta >0$ small enough, we have
$$
\frac{\upsilon h}{\langle x\rangle}\left(\frac{h}{\langle x\rangle}+\epsilon
\right),\,\, \upsilon ^3\left(\frac{h}{\langle x\rangle}+\epsilon
\right)^3\ll \upsilon \left( \frac{h}{\langle x\rangle^{1+\vartheta }}+\epsilon 1_{\{ |x|\ge
    \mu/C \}}\right)
$$
It then follows from (\ref{rest.30}),
(\ref{rest.35}), (\ref{rest.36}) that
\begin{equation}\label{rest.38}
-\Im P_\upsilon ^{\mathrm{top}}(\rho )\ge \frac{\upsilon }{C}\chi ^0(p)\left(
  \frac{h}{\langle x\rangle^{1+\vartheta }}+\epsilon 1_{\{ |x|\ge \mu/C
    \}} \right) -\upsilon \widetilde{k}_\upsilon ,
\end{equation}
where the right hand side is evaluated at the point $\Re \rho $,
\begin{equation}\label{rest.39}
\widetilde{k}_\upsilon ={\cal O}(1)\left(\frac{h}{\langle x\rangle}+\epsilon
\right),\ \mathrm{supp\,}\widetilde{k}_\upsilon  \subset
p^{-1}([E_-,E_+]\setminus [E-\delta_0 ,E+\delta_0 ]).
\end{equation}
We retain from this and (\ref{rest.37}), that
\begin{equation}\label{rest.40}
-\Im P_\upsilon ^{\mathrm{top}}\ge \frac{\upsilon }{C}m_\epsilon (x;h)-\upsilon k_\upsilon ,
\end{equation}
where
\begin{equation}\label{rest.41}
k_\upsilon ={\cal O}\left(\frac{h}{\langle x\rangle} +\epsilon  \right),\
\mathrm{supp\,}k_\upsilon \cap p^{-1}([E-\delta_0 ,E+\delta_0 ])=\emptyset .
\end{equation}

\par Using again the identification of $h$-pops and $h$-tops, we see
that
\begin{equation}\label{rest.42}
\int_{\Lambda _{\upsilon G_\epsilon }}\upsilon  k_\upsilon  Tu\cdot
\overline{Tu}
e^{-2\upsilon H_\epsilon /h}d\alpha
\le \upsilon \|
R_\upsilon u\|^2_{H(\Lambda _{\upsilon G_\epsilon })}+(N_\upsilon 'u|u),
\end{equation}
where $N_\upsilon '$ is negligible of order $\upsilon $ and $R_\upsilon $ is an $h$-pop whose
symbol is ${\cal O}(1)$ in $S((h/\langle x\rangle +\epsilon )^{1/2})$
and with support disjoint from $p^{-1}([E-\delta_0 ,E+\delta_0
])$. Thus with a new negligible operator of order $\upsilon $,
\begin{equation}\label{rest.43}
-\Im (Pu|u)_{H(\Lambda _{\upsilon G_\epsilon })}\ge\frac{\upsilon }{C}\|
m_\epsilon ^{1/2} u \|_{H(\Lambda _{\upsilon G_\epsilon
  })}^2-\upsilon \| R_\upsilon u \|^2_{H(\Lambda _{\upsilon G_\epsilon })}-(N_\upsilon u|u).
\end{equation}

Assume from now on that
\begin{equation}\label{rest.44}
\upsilon >0\hbox{ is fixed and sufficiently small.}
\end{equation}

\par We next remark that the arguments work virtually without any
changes if we replace $P$ by $P-z$ for $z\in {\bf C}$, satisfying,
\begin{equation}\label{rest.45}
\Re z\in [E-\delta_0/2 ,E+\delta_0/2 ],\
-\frac{1}{C}\epsilon _\vartheta  \le \Im z\le \frac{1}{C},
\end{equation}
for $C>$ sufficiently large. Also, since the support of the symbol of
$R_\upsilon $ is contained in a region where $|p-z|\ge 1/{\cal O}(1)$, we have
for every fixed $N>0$,
\begin{multline}
\label{rest.46}
\| R_\upsilon u\|_{H(\Lambda _{\upsilon G_\epsilon })} \\
\le {\cal O}(1)\|(P-z)u\|_{H(\Lambda _{\upsilon G_\epsilon })}+{\cal O}(1
)\| (h/\langle x\rangle) ^N (h/\langle x\rangle +\epsilon )^{1/2}u\|_{H(\Lambda
  _{\upsilon G_\epsilon })}.
\end{multline}
Here (\ref{rest.46}) follows by the calculus of $h$-pseudodifferential operators associated to $\Lambda_{\upsilon G_{\epsilon}}$, see Section 6 of~\cite{HeSj86}.
Using this in (\ref{rest.43}) with $P$ there replaced by $P-z$, we get,
\begin{equation}\label{rest.47}
-\Im ((P-z)u|u)_{H(\Lambda _{\upsilon G_\epsilon })}\ge \frac{1}{C}\|
m_\epsilon ^{1/2}u\|^2_{H(\Lambda
  _{\upsilon G_\epsilon }) }-C\| (P-z)u\|^2_{H(\Lambda _{\upsilon G_\epsilon })}.
\end{equation}
Here, we have for every $\alpha >0$,
\begin{multline*}
  -\Im ((P-z)u|u)_{H(\Lambda _{\upsilon G_\epsilon })}\\
  \le \| m_\epsilon (x;h)^{-\frac{1}{2}}(P-z)u\|_{H(\Lambda
    _{\upsilon G_\epsilon })} \| m_\epsilon (x;h)^{\frac{1}{2}}u\|_{H(\Lambda _{\upsilon G_\epsilon
    })}\\
  \le \frac{\alpha }{2}\| m_\epsilon (x;h)^{-\frac{1}{2}} (P-z)u\|^2_{H(\Lambda _{\upsilon G_\epsilon })}
  +\frac{1}{2\alpha } \| m_\epsilon (x;h)^{\frac{1}{2}} u\|^2_{H(\Lambda _{\upsilon G_\epsilon })}.
\end{multline*}
Use this in (\ref{rest.47}) for a fixed large enough $\alpha $
together with the observation
$$
\| (P-z)u\|_{H(\Lambda _{\upsilon G_\epsilon })}\lesssim \| m_\epsilon (x;h)^{-\frac{1}{2}} (P-z)u\|_{H(\Lambda _{\upsilon G_\epsilon })},
$$
to conclude that
\begin{equation}\label{rest.48}
\| m_\epsilon (x;h)^{\frac{1}{2}} u\|_{H(\Lambda
  _{\upsilon G_\epsilon })}^2\le
{\cal O}(1)
\| m_\epsilon (x;h)^{-\frac{1}{2}} (P-z)u\|_{H(\Lambda
  _{\upsilon G_\epsilon })}^2 .
\end{equation}
Summing up, we have:
\begin{prop}\label{rest3}
Under the assumptions above, in particular {\rm (\ref{rest.22.5})} about non
trapping, we fix $\delta_0 >0$ small and then $\upsilon >0$ small enough. Then
for $z$ in the region {\rm (\ref{rest.45})}, where $C>0$ is large enough,
$$
P-z:H(\Lambda _{\upsilon G_\epsilon },\langle \xi \rangle ^2)\to H(\Lambda
_{\upsilon G_\epsilon })
$$
is bijective and
\begin{equation}\label{rest.49}
m_\epsilon (x;h)^{\frac{1}{2}} (z-P)^{-1}
m_\epsilon (x;h)^{\frac{1}{2}}
={\cal O}(1):\ H(\Lambda _{\upsilon G_\epsilon })\to H(\Lambda _{\upsilon G_\epsilon }),
\end{equation}
where $m_\epsilon $ is defined in {\rm (\ref{rest.37})}.
\end{prop}
\begin{remark}\label{rest4}
By Remark {\rm \ref{rest2}}, we can replace $G_\epsilon $ with
$\epsilon G_{\mathrm{sbd}}$ in {\rm (\ref{rest.49})}. Now $G_{\mathrm{sbd}}$ vanishes
for $\langle x\rangle\le \mu /(2C)$ and it follows that $\|
u\|_{H(\Lambda _{\upsilon G_\mathrm{sbd}})}\asymp \| u\|_{L^2}$ for
$u$ with support in a fixed compact set. Moreover $m_\epsilon
(x;h)\asymp h$ for $x$ in any fixed compact set. Hence from
{\rm (\ref{rest.49})}, we deduce that
$$
\chi (z-P)^{-1}\chi ={\cal O}(1/h):\, L^2\to L^2,
$$
for every fixed $\chi \in C_0^\infty ({\bf R}^n)$.
\end{remark}

\par In order to shorten the notation, we will often write
\begin{equation}\label{rest.49.5}
{\cal H}_{\mathrm{sbd}}=H(\Lambda _{\upsilon \epsilon G_{\mathrm{sbd}}}),\quad
{\cal D}_{\mathrm{sbd}}=H(\Lambda _{\upsilon \epsilon G_{\mathrm{sbd}}},\langle \xi \rangle^2),
\end{equation}
where $\upsilon >0$ is small and fixed, as above.

\par We finally treat a trapping case, namely that of a potential well
in an island, generating shape resonances. As before, let $E\in ]E_-+\delta
,E_+-\delta [$ be a fixed energy level. (We can also allow it to vary,
as we shall do in the next section, but then some geometric
quantities will also vary.)
Let ${\ddot{\mathrm{O}}}\Subset {\bf
  R}^n$ be a connected open set (still assuming
$n=1$ but trying to keep the discussion as general as possible). Let
$U_0\subset {\ddot{\mathrm{O}}}$ be a compact subset. Assume:
\begin{equation}\label{rest.50}
V-E<0\hbox{ in }{\bf R}^n\setminus \overline{{\ddot{\mathrm{O}}}},\ V-E>0 \hbox{ in
}{\ddot{\mathrm{O}}}\setminus U_0,\ V-E\le 0\hbox{ in }U_0,
\end{equation}
\begin{equation}\label{rest.51}
\mathrm{diam}_dU_0=0,
\end{equation}
where $d$ is the Lithner-Agmon distance given by the metric $(V-E)_+(x)dx^2$,
$(V-E)_+=\max (V-E,0)$,
\begin{equation}\label{rest.52}
\hbox{The }H_p\hbox{-flow has no trapped trajectories in
}{{p^{-1}(E)}_\vert}_{{\bf R}^n\setminus \ddot{\mathrm{O}}}.
\end{equation}

\par Let $M_0\subset {\ddot{\mathrm{O}}}$ be a connected compact set with
smooth boundary such that
\begin{equation}\label{rest.53}
M_0\supset \{x\in {\ddot{\mathrm{O}}}; d(x,\partial {\ddot{\mathrm{O}}})\ge \epsilon _0 \},
\end{equation}
for some small $\epsilon _0>0$. Let $P_0$ denote the Dirichlet
realization of $P$ in $M_0$, equipped with the domain ${\cal D}(P_0) = H^2(\stackrel{\circ}{M}_0) \cap H^1_0(\stackrel{\circ}{M}_0)$.
(The right hand side in (\ref{rest.53}) has smooth boundary, as we recalled after (\ref{bicsint.9}).)

\par From Agmon estimates we have the
well-known fact that if $\widetilde{M}_0\subset \ddot{\mathrm{O}}$ has
the same properties as $M_0$ with the same value of $\epsilon _0$,
then in any $o(1)$-neighborhood of $E$, the eigenvalues of $P_0$ and
$\widetilde{P}_0$ differ by ${\cal O}_\alpha (1)\exp 2(\alpha
+\epsilon _0-d(U_0,\partial
\ddot{\mathrm{O}}))/h$ for every $\alpha >0$. (Cf.\ \cite{HeSj84})

\par Let $K(h)\subset {\bf C}$ converge to
$\{ E \}$, when $h\to 0$ such that uniformly for all $z\in K(h)$,
\begin{equation}\label{rest.54}z\hbox{ satisfies (\ref{rest.45})},
\end{equation}
\begin{equation}\label{rest.54.3}
\mathrm{dist\,}(z,\sigma (P_0))\ge \lambda (h),
\end{equation}
where $\lambda (h)>0$ and
\begin{equation}\label{rest.54.6}
\ln \lambda (h)\ge -o(1)/h,\ h\to 0.
\end{equation}

Let $\widetilde{V}=V+W$, where $0\le W\in C_0^\infty
(\mathrm{neigh\,}(U_0))$ has its support in a small neighborhood of
$U_0$ and $V+W-E>0$ in ${\ddot{\mathrm{O}}}$. Let $\widetilde{P}:=-h^2\Delta
+V+W$, $\widetilde{p}(x,\xi )=\xi ^2+V+W$. Then $\widetilde{p}$
satisfies (\ref{rest.22.5}) if $E_-<E<E_+$ and $E_+-E$, $E-E_-$ are
small enough. Then the resolvent estimate
(\ref{rest.49}) applies to $\widetilde{P}$ for $z\in K(h)$ and we
recall Remark \ref{rest4}. We can then apply Agmon estimates inside
$\ddot{\mathrm{O}}$, as explained in Section 6 of~\cite{DiSj99}, and we get
\begin{equation}
\label{rest.55}
(\widetilde{P}-z)^{-1}(x,y)=\check{{\cal O}}(e^{-d(x,y)/h}),\ x,y\in {\ddot{\mathrm{O}}} ,\ z\in K(h),
\end{equation}
where the notation $\check{{\cal O}}$ is explained in Proposition 9.3
in \cite{HeSj86}. Under the assumptions (\ref{rest.54.3}),
(\ref{rest.54.6}) we get the same estimate for $P_0$, i.e.\ we can
replace $(\widetilde{P},\ddot{\mathrm{O}})$ by $(P_0,M_0)$ in (\ref{rest.55}).

Recall the elementary telescopic formula, for the moment under the a priori
assumption that $(P-z)^{-1}$ exists for $z\in K(h)$ (which will follow
from the discussion):
\begin{multline}
\label{rest.56}
(P-z)^{-1} \\ =(\widetilde{P}-z)^{-1}+(\widetilde{P}-z)^{-1}W
(\widetilde{P}-z)^{-1}+
(\widetilde{P}-z)^{-1}W (P-z)^{-1}W (\widetilde{P}-z)^{-1}.
\end{multline}
It reduces the study of $(P-z)^{-1}$ to that of $W(P-z)^{-1}W$ and we
shall make a perturbation series approach. Let $\chi \in C_0^\infty
(\stackrel{\circ }{M}_0)$ be equal to 1 on $\{ x\in {\ddot{\mathrm{O}}};\, d(x,\partial
{\ddot{\mathrm{O}}})\ge 2\epsilon _0 \}$ and let $\chi _0\in C_0^\infty
(\mathrm{neigh\,}(\mathrm{supp\,}W))$ be equal to one near
$\mathrm{supp\,}W$. Take the neighborhood small enough so that
$\mathrm{supp\,}\chi _0\cap \mathrm{supp\,}(1-\chi )=\emptyset $.
Let $\chi _1\in C_0^\infty (\stackrel{\circ}{M}_0)$ satisfy
$1_{B_d(U_0,S_0/2-\epsilon _0)}\prec \chi _1\prec 1_{B_d
  (U_0,S_0/2+\epsilon _0)}$, where $B_d(U_0,r)$ denotes the open ball
of center $U_0$ and radius $r$ for the Lithner-Agmon distance and
$S_0:=d(U_0,\partial \ddot{\mathrm{O}})$.
As a
first approximation to $(P-z)^{-1}$, we take
\begin{equation}\label{rest.57}
E=\chi (P_0-z)^{-1}\chi _1+(\widetilde{P}-z)^{-1}(1-\chi _1)=:E_0+\widetilde{E}.
\end{equation}
Then,
\begin{multline}\label{rest.58}
  (P-z)E=1+[P,\chi ](P_0-z)^{-1}\chi _1-W(\widetilde{P}-z)^{-1}(1-\chi
  _1)
  =: 1-K.
\end{multline}

\par From (\ref{rest.55}) and the corresponding estimate for $P_0$, we see
that
\begin{equation}\label{rest.60}
\| K\|_{{\cal L}(m_\epsilon ^{1/2}{\cal H}_\mathrm{sbd}, L^2(\ddot{\mathrm{O}}))}=\widetilde{{\cal O}}(1) \exp\left(-\frac{S_0}{2h}\right),
\end{equation}
where $\widetilde{\cal O}(1)$ indicates a quantity which is ${\cal O}(e^{\alpha /h})$ for some $\alpha >0$ which tends to 0 when $\epsilon _0$ and
$d(\mathrm{supp\,}\chi _0,U_0)$ tend to 0. Thus for $h>0$ small enough, the Neumann series,
\begin{equation}\label{rest.61}
  1+K+K^2+K^3+...
\end{equation}
converges to $(1-K)^{-1}=1+\widetilde{{\cal O}} (1)\exp \left(-\frac{S_0}{2h}\right)$ in
${\cal L}(m_\epsilon^{1/2}{\cal H}_\mathrm{sbd},L^2(\ddot{\mathrm{O}}))$. It
follows that $E(1-K)$ is a right inverse of $P-z:\, {\cal D}_\mathrm{sbd}\to {\cal H}_{\mathrm{sbd}}$ and since the latter is
of index 0 by the general theory of resonances (\cite{HeSj86}) it is a two-sided inverse.
\begin{prop}\label{rest5}
Let ${\bf C}\supset K(h)\to \{ E \}$, $h\to 0$ and assume
{\rm (\ref{rest.54})}--{\rm (\ref{rest.54.6})}. Then for $h>0$ small enough and
for $z\in K(h)$, $P-z:\, {\cal  D}_\mathrm{sbd}\to {\cal H}_{\mathrm{sbd}}$ is bijective with inverse
\begin{equation}\label{rest.62}
  (P-z)^{-1}=\chi (P_0-z)^{-1}\chi
  _1(1-K)^{-1}+(\widetilde{P}-z)^{-1}(1-\chi _1)(1-K)^{-1}.
\end{equation}
Here,
\begin{equation}\label{rest.63}
\| \chi (P_0-z)^{-1}\chi _1(1-K)^{-1}\|_{{\cal L}({\cal
    H}_\mathrm{sbd},{\cal H}_\mathrm{sbd})}
=\frac{{\cal O}(1)}{\mathrm{dist\,}(z,\sigma (P_0))}
\end{equation}
and by Proposition {\rm \ref{rest3}},
\begin{equation}\label{rest.64}
m_\epsilon (x;h)^{\frac{1}{2}}(\widetilde{P}-z)^{-1}(1-\chi
_1)(1-K)^{-1}m_\epsilon (x;h)^{\frac{1}{2}}={\cal O}(1):{\cal
  H}_{\mathrm{sbd}}\to {\cal H}_{\mathrm{sbd}}.
\end{equation}
\end{prop}

\par We next study the situation when $z$ gets closer to $\sigma
(P_0)$. Let $J(h)\subset {\bf R}$ be an
interval tending to $\{ E \}$ as $h\to 0$. Assume that
\begin{equation}\label{rest.65}
P_0\hbox{ has no spectrum in } \partial J(h)+[-\delta (h),\delta (h)]
\end{equation}
where the parameter $\delta (h)$ is small but not exponentially small;
$$
\ln \delta (h)\ge -o(1)/h.
$$
$\sigma (P_0)\cap J(h)$ is a discrete set of the form
$\{ \mu _1(h),...,\mu _m(h) \}$ where $m=m(h)={\cal O}(h^{-n})$ and we
repeat the eigenvalues according to their multiplicity. Let
$\Gamma (h)$ denote the set of resonances of $P$ in
$J(h)-i[0,\epsilon_\vartheta  /C]$, $C\gg 1$, also
repeated according to their (algebraic) multiplicity.
Assume that
\begin{equation}\label{rest.65.5}
\epsilon \ge e^{-1/(Ch)},
\end{equation}
for some $C\gg 1$, so that
$$
\epsilon \left(\frac{\epsilon }{h} \right)^\vartheta \ge e^{\frac{1}{{\cal O}(h)}-\frac{2S_0}{h}}.
$$
 Then we have,
\begin{prop}\label{rest6}
There is a bijection $b:\{ \mu _1,...,\mu _m \}\to \Gamma (h)$, such that
$$
b(\mu )-\mu =\widetilde{{\cal O}}(e^{-2S_0/h}),
$$
where $S_0=d(U_0,\partial {\ddot{\mathrm{O}}})$ and the tilde
indicates that the right hand side is ${\cal O}(e^{(\omega
  -2S_0)/h})$, where $\omega =\omega (\epsilon _0)>0$ and
$\omega (\epsilon _0) \to 0$, when $\epsilon _0\to 0$.
\end{prop}

\par We shall prove the proposition and also get precise information
about the resolvent by studying an associated Grushin problem. Let
$e_1^0(h),...,e_m^0(h)$ be an orthonormal system of eigenfunctions of
$P_0$ associated to the eigenvalues $\mu _1(h),...,\mu _m(h)$. Then we
know from Chapter 6 of~\cite{DiSj99} that
\begin{equation}\label{rest.66}
e_j^0=\check{\cal O}(e^{-d(U_0,x)/h}).
\end{equation}
A first trivial Grushin problem for $P_0$ is defined by the matrix
\begin{equation}\label{rest.67}
{\cal P}_0=\begin{pmatrix} P_0-z &R_-^0\\ R_+^0 &0\end{pmatrix}:
{\cal D}(P_0)\times {\bf C}^m\to L^2(M_0)\times {\bf C}^m,
\end{equation}
where
\begin{equation}\label{rest.68}
R_+^0u(j)=(u|e_j^0), \ R_-^0=(R_+^0)^*.
\end{equation}

\par Let
\begin{equation}\label{rest.69}
\widetilde{J}(h)=J(h)+i[-\epsilon_\vartheta  /C,1/C]
\end{equation}
with $C>0$ sufficiently large (cf.\ (\ref{rest.45})). Then it follows from Section 9
of~\cite{HeSj86} that for $z\in \widetilde{J}(h)$, the operator ${\cal P}_0(z)$ is bijective with inverse
\begin{equation}\label{rest.70}
{\cal E}_0=\begin{pmatrix}E^0(z) &E_+^0(z)\\ E_-^0(z)
  &E^0_{-+}\end{pmatrix}:
L^2(M_0)\times {\bf C}^m\to {\cal D}(P_0)\times {\bf C}^m,
\end{equation}
where, with $\Pi _0$ denoting the spectral projection onto the space
spanned by $e_1^0,...,e_m^0$,
\begin{equation}\label{rest.71}
E^0(z)=(P_0-z)^{-1}(1-\Pi_0)={\cal O}(1/\delta (h)):L^2\to {\cal D}(P_0),
\end{equation}
\begin{equation}\label{rest.72}
E_+^0(z)v_+=\sum v_+(j)e_j^0,\ \| E_+^0\|_{{\cal L}({\bf C}^m,{\cal
    D}(P_0))}\le {\cal O}(1),
\end{equation}
\begin{equation}\label{rest.73}
E_-^0(z)u(j)=(u|e_j^0),\ \| E_-^0\|_{{\cal L}(L^2,{\bf C}^m)}\le {\cal O}(1),
\end{equation}
\begin{equation}\label{rest.74}E^0_{-+}(z)=\mathrm{diag\,}(z-\mu_j).\end{equation}

\par Choose $\chi \in C_0^\infty (\stackrel{\circ }{M}_0)$ as after
(\ref{rest.56}) and put
\begin{equation}\label{rest.75}
R_+=R_+^0\chi :\, {\cal H}_\mathrm{sbd}\to {\bf C}^m,
\end{equation}
\begin{equation}\label{rest.76}
R_-=\chi R_-^0 :\, {\bf C}^m\to {\cal H}_\mathrm{sbd}.
\end{equation}

Define,
\begin{equation}\label{rest.81}
{\cal P}(z)=\begin{pmatrix}P-z &R_-\\ R_+ &0\end{pmatrix}:
{\cal D}_{\mathrm{sbd}}\times {\bf C}^m\to {\cal H}_{\mathrm{sbd}}\times {\bf C}^m.
\end{equation}
This is a Fredholm operator of index 0, so to show that it is bijective
it suffices to construct a right inverse.

\par Let $\chi _0$, $\chi _1$, $\chi $ be as in the construction of
$(P-z)^{-1}$ in Proposition \ref{rest5} (where the assumptions on
$z$ were different). Following the same path as there, we put
\begin{equation}\label{rest.82}
\begin{split}
\widetilde{{\cal E}}=
\begin{pmatrix}\chi E^0\chi _1 &\chi E_+^0\\ E_-^0\chi &E^0_{-+}\end{pmatrix}+
\begin{pmatrix}(\widetilde{P}-z)^{-1}(1-\chi _1) &0\\ 0
  &0\end{pmatrix}=:\begin{pmatrix}\widetilde{E} &\widetilde{E}_+\\
  \widetilde{E}_- &\widetilde{E}_{-+}\end{pmatrix},
\end{split}
\end{equation}
\begin{equation}\label{rest.82.5}
\widetilde{E}=\underbrace{{\cal O}(1/\delta (h))}_{{\cal
    H}_\mathrm{sbd}\to {\cal D}_\mathrm{sbd}}+\underbrace{{\cal
    O}(1)}_{m_\epsilon ^{1/2}{\cal H}_\mathrm{sbd}\to m_\epsilon
  ^{-1/2}{\cal D}_\mathrm{sbd}}={\cal O}(1/\epsilon _\vartheta ):{\cal
  H}_\mathrm{sbd}\to {\cal D}_\mathrm{sbd}.
\end{equation}
A straight forward calculation gives
\begin{equation}\label{rest.83}
{\cal P}(z)\widetilde{{\cal E}}(z)=\begin{pmatrix}A_{11} &A_{12}\\
  A_{21} &A_{22}\end{pmatrix},
\end{equation}
where
\[
\begin{split}
A_{11}&=(P-z)\chi E^0\chi _1+\chi R^0_-E_-^0\chi +1-\chi
_1-W(\widetilde{P}-z)^{-1}(1-\chi _1),\\
A_{12}&=(P-z)\chi E_+^0+\chi R^0_-E_{-+}^0,
\\
A_{21}&=R^0_+\chi^2 E^0\chi _1+R_+^0\chi (\widetilde{P}-z)^{-1}(1-\chi _1)
\\
A_{22}&=R_+^0\chi^2 E_+^0.
\end{split}
\]
Here, using standard Lithner-Agmon estimates for $E^0$, $(\widetilde{P}-z)^{-1}$, together with the fact that
$$
E_-^0 \chi = E_-^0 \chi_1 + \widetilde{{\cal O}} (e^{-S_0/(2h)}),
$$
we get
\[\begin{split}
(P-z)\chi E^0\chi _1+\chi R^0_-E_-^0\chi
&=\chi (\underbrace{(P-z)E^0+R^0_-E_-^0}_{=1})\chi _1+
\widetilde{{\cal O}} (e^{-S_0/(2h)})\\
&=\chi _1+\widetilde{{\cal O}} (e^{-S_0/(2h)}),
\end{split}
\]
and
$$
W(\widetilde{P}-z)^{-1}(1-\chi _1)=\widetilde{{\cal O}} (e^{-S_0/(2h)}).
$$
Hence \begin{equation}\label{rest.84}A_{11}=1+\widetilde{{\cal
      O}}(e^{-S_0/(2h)})\end{equation}
Similarly,
\begin{equation}\label{rest.85}
  A_{12}=\chi (\underbrace{(P-z)E_+^0+R^0_-E^0_{-+}}_{=0})+
\underbrace{[P,\chi ]E_+^0}_{=\widetilde{{\cal O}} (e^{-S_0/h})}=\widetilde{{\cal O}} (e^{-S_0/h}),
\end{equation}
\begin{equation}\label{rest.86}
\begin{split}
&A_{21}=\\
&\underbrace{R^0_+E^0}_{=0}\chi _1+\underbrace{R_+^0(\chi ^2 -1_{M_0}
  )E^0\chi _1}_{=\widetilde{{\cal
      O}}(e^{-{3S_0/(2h)}})}+\underbrace{R^0_+\chi (\widetilde{P}-z)^{-1}(1-\chi
  _1)}_{=\widetilde{{\cal O}}(e^{-S_0/(2h)})}
=\widetilde{{\cal O}}(e^{-S_0/(2h)}),
\end{split}
\end{equation}
\begin{equation}\label{rest.87}
A_{22}=\underbrace{R_+^0E_+^0}_{=1}+R_+^0(\chi ^2 -1_{M_0}
  )E_+^0=1+\widetilde{{\cal O}}(e^{-2S_0/h}).
\end{equation}

\par A first conclusion is that
$$
{\cal P}(z)\widetilde{{\cal E}}(z)=1+\widetilde{{\cal O}}(e^{-S_0/(2h)}): {\cal H}_{\mathrm {sbd}} \times {\bf C}^m
\rightarrow {\cal H}_{\mathrm {sbd}} \times {\bf C}^m,
$$
where the remainder has entries with distribution kernels supported in
$M_0\times M_0$, $M_0\times {\bf C}^m$, ${\bf C}^m\times M_0$,
$\emptyset $ respectively.
so ${\cal P}(z)$ is bijective with inverse
\begin{equation}\label{rest.88}\begin{split}
{\cal E}(z)=&\widetilde{{\cal E}}(z)(1+\widetilde{{\cal
    O}}(e^{-S_0/(2h)}))=\\
&\widetilde{{\cal E}}(z)+\widetilde{{\cal O}}(e^{-S_0/(2h)}):
m_\epsilon ^{\frac{1}{2}}{\cal H}_{\mathrm{sbd}}\times {\bf C}^m\to
m_\epsilon ^{-\frac{1}{2}}{\cal D}_{\mathrm{sbd}}\times {\bf C}^m
\end{split}\end{equation}
In particular, if we write
\begin{equation}\label{rest.89}
{\cal E}(z)=\begin{pmatrix}E(z) &E_+(z)\\ E_-(z) &E_{-+}(z)\end{pmatrix},
\end{equation}
we get
$$
E_{-+}-E_{-+}^0=\widetilde{{\cal O}}(e^{-S_0/(2h)})
$$
and $E$ satisfies the estimate in (\ref{rest.82.5}).

\par We shall improve this estimate by working with exponential
weights in $\ddot{\mathrm{O}}$. For $\phi \in C_0^\infty
(\stackrel{\circ }{M}_0)$ real-valued, we put
$$
{\cal H}^\phi _{\mathrm{sbd}}=e^{\phi /h}{\cal H}_{\mathrm{sbd}}
$$
equipped with the norm $\| e^{-\phi /h} u\|_{{\cal H}_\mathrm{sbd}}
$. As a vector space it is equal to ${\cal H}_\mathrm{sbd}$. The
constructions above work without any great changes if we assume that
$$
\mathrm{supp\,}\nabla \phi \cap U_0=\emptyset ,\ (\nabla \phi )^2\le
V-E-\epsilon _0.
$$
By varying $\phi $ we see that
\begin{equation}\label{rest.90}
\begin{split}
  &(m_\epsilon ^{\frac{1}{2}}Em_\epsilon
  ^{\frac{1}{2}})(x,y)=\check{{\cal O}}(e^{-d(x,y)/h}),\
  (m_\epsilon ^{\frac{1}{2}}E_+)(x)=\check{{\cal
      O}}(e^{-d(x,U_0)/h}),\\
  &(E_-m_\epsilon ^{\frac{1}{2}})(y)=\check{{\cal O}} (e^{-d(U_0,y)/h}),\ x,y\in B_d(U_0,S_0),
\end{split}
\end{equation}
where we use the same symbols to denote the distribution kernels of
$E$, $E_\pm$.

\par For $v_+\in {\bf C}^m$, the solution $(u,u_-)$ of the problem
\begin{equation}\label{rest.91}
\begin{cases}
(P-z)u+R_-u_-=0,\\
R_+u=v_+,
\end{cases}
\end{equation}
is given by $u=E_+v_+$, $u_-=E_{-+}v_+$. As an approximate solution to
(\ref{rest.91}), we take $u^0=\chi E_+^0v_+$,
$u_-^0=E^0_{-+}v_+$. Then
$$
\begin{cases}
(P-z)u^0+R_-u_-^0=[P,\chi ]E_+^0v_+,\\
R_+u^0=v_++R_+^0(\chi^2-1_{M_0})E_+^0v_+,
\end{cases}
$$
so we get the solution to (\ref{rest.90}) in the form
\begin{equation}\label{rest.92}
\begin{cases}
u=u^0-E[P,\chi ]E_+^0v_+-E_+R_+^0(\chi^2-1_{M_0})E_+^0v_+,\\
u_-=u_-^0-E_-[P,\chi ]E_+^0v_+-E_{-+}R_+^0(\chi^2-1_{M_0})E_+^0v_+.
\end{cases}
\end{equation}
Now it follows from (\ref{rest.90}) and the corresponding estimates
for $E^0$, $E_\pm^0$, that
$$
|u_--u_-^0|=\widetilde{{\cal O}}(e^{-2S_0/h})|v_+|,
$$
which means that
\begin{equation}\label{rest.93}
E_{-+}-E_{-+}^0=\widetilde{{\cal O}}(e^{-2S_0/h}).
\end{equation}
\begin{proofof} Proposition \ref{rest6}. By (\ref{rest.74}),
(\ref{rest.93}) reads
\begin{equation}\label{rest.94}
E_{-+}-\mathrm{diag\,}(z-\mu _j)=\widetilde{{\cal O}}(e^{-2S_0/h}).
\end{equation}
Now the resonances of $P$ in $\widetilde{J}(h)$ are the zeros of $\det
(E_{-+})$ and we get the proposition by means of elementary arguments
for zeros of holomorphic functions of one variable.
\end{proofof}

\par The first equation in (\ref{rest.92}) can be written
$$
E_+v_+-\chi E_+^0v_+=-\left(
E[P,\chi ]E_+^0+E_+R_+^0(\chi^2-1_{M_0})E_+^0
 \right)v_+
$$
and it follows that
$$
\|m_\epsilon ^{\frac{1}{2}}(E_+v_+-\chi E_+^0v_+) \|_{{\cal
    H}_{\mathrm{sbd}}}=\widetilde{{\cal O}}(e^{-S_0/h})|v_+|,
$$
i.e.\
\begin{equation}\label{rest.95}
m_\epsilon ^{\frac{1}{2}}(E_+-\chi E_+^0)=\widetilde{{\cal
    O}}(e^{-S_0/h}):\ {\bf C}^m\to {\cal H}_{\mathrm{sbd}} .
\end{equation}

\par Taking the adjoints with respect to the scalar product on $L^2({\bf
  R}^n)\times {\bf C}^m$, we have
$$
{\cal P}(z)^*{\cal E}(z)^*=1,
$$
where
$$
{\cal P } (z)^*=\begin{pmatrix}P^*-z &R_+^*\\ R_-^* &0\end{pmatrix},\ \
{\cal E}(z)^*=\begin{pmatrix}E(z)^* &E_-(z)^
  *\\ E_+(z)^* &E_{-+}(z)^*\end{pmatrix},
$$
and hence ${\cal E}(z)^*$ can be constructed by starting with
$$
\widehat{{\cal E}}(z)=\begin{pmatrix}
\chi {E^0}^* \chi _1 &\chi {E_-^0}^* \\ {E_+^0}^*\chi  &{E_{-+}^0}^*
\end{pmatrix}
+
\begin{pmatrix}
(\widetilde{P}^*-\overline{z})^{-1}(1-\chi _1) &0\\ 0 &0 .
\end{pmatrix}
$$
In analogy with (\ref{rest.95}) we get
$$
m_\epsilon ^{\frac{1}{2}}(E_-^* -\chi {E_-^0}^*)=
\widetilde{{\cal O}}(e^{-S_0/h}):\ {\bf C}^m\to {\cal H}_{\mathrm{sbd}}^*,
$$
where we notice that ${\cal H}_{\mathrm{sbd}}^*=H(\Lambda_{-\upsilon G_{\mathrm{sbd}}})$, in view of Proposition 8.8 of~\cite{HeSj86}. By duality,
\begin{equation}
\label{rest.96}
(E_--E_-^0\chi )m_\epsilon ^{\frac{1}{2}}=\widetilde{{\cal O}}(e^{-S_0/h}):\ {\cal H}_{\mathrm{sbd}}\to {\bf C}^m.
\end{equation}

\par Recall the standard formula for Grushin problems:
\begin{equation}\label{rest.97}
(z-P)^{-1}=-E(z)+E_+(z)E_{-+}(z)^{-1}E_-(z),\ z\in
\widetilde{J}\setminus \Gamma (h).
\end{equation}
Here, $E(z)$ is holomorphic and by (\ref{rest.89}), (\ref{rest.82}) , (\ref{rest.88}),
\begin{equation}\label{rest.98}
E(z)=\chi E^0\chi _1+(\widetilde{P}-z)^{-1}(1-\chi _1)+{\cal
  O}(e^{-S_0/(2h)}):
m_\epsilon ^{\frac{1}{2}}{\cal H}_{\mathrm{sbd}}\to
m_\epsilon ^{-\frac{1}{2}}{\cal D}_{\mathrm{sbd}}.
\end{equation}
\begin{equation}\label{rest.99}
E(z)={\cal O}(h/\delta )+{\cal O}(1):\ m_\epsilon ^{\frac{1}{2}}{\cal
  H}_\mathrm{sbd}
\to m_\epsilon ^{-\frac{1}{2}}{\cal D}_\mathrm{sbd}.
\end{equation}
Here the term ${\cal O}(h/\delta )$ represents the term $\chi E^0\chi
_1$ which is ${\cal O}(1/\delta )$ as an operator on $L^2({\bf
  R}^n)$.

\par When either $m=1$ or $\mathrm{dist\,}(z,\sigma (P_0))\ge
\widetilde{{\cal O}}(e^{-2S_0/h})$, it follows from (\ref{rest.93})
that
\begin{equation}\label{rest.100}
E_{-+}^{-1}={\cal O}\left( \frac{1}{\mathrm{dist\,}(z,\Gamma (h))}
\right) .
\end{equation}
Here we also assumed for simplicity that
$\mathrm{dist\,}(z,\sigma(P_0)\cap J(h))=\mathrm{dist\,}(z,\sigma(P_0))$ which can be achieved by a slight
shrinking of the interval $J(h)$. Thus, when (\ref{rest.100}) holds, we
get
\begin{equation}\label{rest.101}
E_+E_{-+}^{-1}E_-={\cal O}\left( \frac{1}{\mathrm{dist\,}(z,\Gamma (h))}
\right) :\ m_\epsilon ^{\frac{1}{2}}{\cal H}_{\mathrm{sbd}}\to
m_\epsilon ^{-\frac{1}{2}}{\cal D}_{\mathrm{sbd}},
\end{equation}
where we also used that by (\ref{rest.95}), (\ref{rest.96}),
\begin{equation}\label{rest.102}
E_+={\cal O}(1):\ {\bf C}^m\to m_\epsilon ^{-\frac{1}{2}}{\cal D}_{\mathrm{sbd}},
\end{equation}
\begin{equation}\label{rest.103}
E_-={\cal O}(1):\ m_\epsilon ^{\frac{1}{2}}{\cal H}_{\mathrm{sbd}}\to
{\bf C}^m.
\end{equation}
Now (\ref{rest.101}) implies that $\psi E_+E_{-+}^{-1}E_- \psi $ is
${\cal O}(1/(h\,\mathrm{dist\,}(z,\Gamma ))):\, L^2\to L^2$ for every $\psi \in C_0^\infty $. Using (\ref{rest.95}), (\ref{rest.96})
more directly, we get
\begin{equation}\label{rest.104}
E_+E_{-+}^{-1}E_--\chi E^0_+E_{-+}^{-1} E^0_-\chi =
\frac{1}{\mathrm{dist\,}(z,\Gamma )}\widetilde{{\cal O}}(e^{-S_0/h}):\
m_\epsilon ^{\frac{1}{2}}{\cal H}_{\mathrm{sbd}}\to
m_\epsilon ^{-\frac{1}{2}}{\cal D}_{\mathrm{sbd}}
\end{equation}
and here
\begin{equation}\label{rest.105}
\chi E^0_+E_{-+}^{-1} E^0_-\chi ={\cal
  O}(h/\mathrm{dist\,}(z,\Gamma )):\ m_\epsilon ^{\frac{1}{2}}{\cal H}_{\mathrm{sbd}}\to
m_\epsilon ^{-\frac{1}{2}}{\cal D}_{\mathrm{sbd}}.
\end{equation}

\par From (\ref{rest.97}), (\ref{rest.98}), (\ref{rest.99}),
(\ref{rest.100}), (\ref{rest.102}), (\ref{rest.103}), we get
\begin{prop}
\label{rest7}
We let $z$ vary in the set $\widetilde{J}(h)$ in {\rm (\ref{rest.69})}.
Assume that $m=1$ or $\mathrm{dist\,}(z,\sigma (P_0))\ge
\widetilde{{\cal O}}(e^{-2S_0/h})$, and also that $\mathrm{dist\,}(z,\sigma
(P_0)\cap J(h))=\mathrm{dist\,}(z,\sigma (P_0))$. Then we have,
\begin{equation}
\label{rest.106}
  (z-P)^{-1}={\cal O}(h/\delta )+{\cal O}(1)+{\cal
    O}(h/\mathrm{dist\,}(z,\Gamma )): \ m_\epsilon ^{\frac{1}{2}}{\cal H}_{\mathrm{sbd}}\to
  m_\epsilon ^{-\frac{1}{2}}{\cal D}_{\mathrm{sbd}},
\end{equation}
where the first two terms to the right are holomorphic in $z$.
\end{prop}

\section{Back to adiabatics}\label{bics}
\setcounter{equation}{0}

Let $I \subset {\bf R}$ be an interval and let
\begin{equation}
\label{bics.1}
V_t=V(t,x)\in C_{b}^\infty (I\times {\bf R}^n;{\bf R}).
\end{equation}
We assume that (cf.\ (\ref{rest.3}))
\begin{equation}\label{bics.2}
\begin{split}
&V_t\hbox{ has a holomorphic extension (also denoted } V_t \hbox{)
  to}\\
&\{x\in {\bf C}^n;\, |\Re x|>C,\ |\Im x|<|\Re x|/C \}\\&\hbox{such that
}
V_t(x)=o(1),\ x\to \infty .
\end{split}
\end{equation}
\begin{equation}\label{bics.3}\begin{split}
\partial _tV_t(x)=0\hbox{ for }|x|\ge C,\hbox{ for some constant }C>0
\end{split}
\end{equation}
It is tacitly assumed that $V(t,x)$ does not depend on $h$. However,
when considering a narrow potential wells in an island, of diameter
$\asymp h$, we will have to make an exception and
allow such an $h$-dependence  in a small neighborhood of the well.

\par Let $0<E_-<E_-'<E_+'<E_+<\infty $ and let
\begin{equation}\label{bics.4}
E_0(t)\in C_b^\infty (I;[E_-',E_+']).
\end{equation}
We assume that $V_t-E_0(t)$ has a  potential well in an island as in
Section \ref{rest}.

\par Let $\ddot{\mathrm{O}}=\ddot{\mathrm{O}}(t)\Subset {\bf R}^n$ be a
connected open set and let $U_0(t)\subset \ddot{\mathrm{O}}(t)$ be
compact. Assume (cf.\ (\ref{rest.50})),
\begin{equation}\label{bics.5}
V_t-E_0(t)\begin{cases}
<0 \hbox{ in } {\bf R}^n\setminus \overline{\ddot{\mathrm{O}}}(t),\\
>0 \hbox{ in } \ddot{\mathrm{O}}(t)\setminus U_0(t),\\
\le 0\hbox{ in }U_0(t),
\end{cases}
\end{equation}
\begin{equation}\label{bics.6}
\mathrm{diam}_{d_t}(U_0(t))=0.
\end{equation}
Here $d_t$ is the Lithner-Agmon distance on $\ddot{\mathrm{O}}(t)$,
given by the metric $(V_t-E_0(t))_+dx^2$.

\par Also assume that with $p_t=\xi ^2+V_t(x)$,
\begin{equation}\label{bics.7}
\hbox{the }H_{p_t}\hbox{-flow has no trapped trajectories in
}{p_t^{-1}(E_0(t))_|}_{{\bf R}^n\setminus \ddot{\mathrm{O}}(t)}.
\end{equation}
It follows that
\begin{equation}\label{bics.8}
d_xV_t\ne 0 \hbox{ on }\partial \ddot{\mathrm{O}}(t),
\end{equation}
so $\partial \ddot{\mathrm{O}}(t)$ is smooth and depends smoothly on
$t$. Thus $\ddot{\mathrm{O}}(t)$ is a manifold with smooth boundary,
depending smoothly on $t$. Further, $U_0(t)$ depends continuously on
$t$. (This will still be true when we allow $h$-dependence near
$U_0(t)$.)

\par For $\epsilon _0>0$ small, we define
\begin{equation}\label{bics.9}
M_0(t)=\{ x\in \ddot{\mathrm{O}}(t);\, d_t(x,\partial
\ddot{\mathrm{O}}(t))\ge \epsilon _0 \} ,
\end{equation}
so $M_0(t)\Subset \ddot{\mathrm{O}}(t)$ is a compact set with smooth
boundary, depending smoothly on $t$. (Here we use the structure of
$d_t(x,\ddot{\mathrm{O}}(t))$ that follows from (\ref{bics.8}), see
\cite{HeSj86}).

When $I$ is a fixed compact interval, the assumptions above are
fulfilled uniformly in $t$. Since we also want to allow $I$ to be a
very long interval, we add the following compactness assumption:
\begin{equation}\label{bics.9.5}
\begin{split}  &(V_t,E_0(t))\in {\cal K},\ \forall t\in I,\hbox{ where }{\cal K}
  \hbox{ is a compact subset of }\\
&\{V\in C_b^\infty ({\bf R}^n;{\bf R});\, V\hbox{ satisfies (\ref{bics.2}) with
  a fixed constant }C
\}\times [E_-',E_+']\\
&\hbox{such that }(V,E)\hbox{ satisfies the assumptions (\ref{bics.3}),
  with a fixed }C\\ &\hbox{as well as
(\ref{bics.5}), (\ref{bics.6}), (\ref{bics.7}).}
\end{split}
\end{equation}

\par Let $P_0(t)$ denote the Dirichlet realization of $P(t)=-h^2\Delta
+V_t(x)$ on $M_0(t)$. If we enumerate the eigenvalues of $P_0(t)$ in
$]E_-,E_+[$ in increasing order (repeated with multiplicities) we know
(as a general fact for 1-parameter families of self-adjoint
operators), that they are uniformly Lipschitz functions of $t$. Let
$\mu_0 (t)=\mu_0 (t;h)$ be such an eigenvalue and assume (cf.\ (\ref{rest.65})),
\begin{equation}\label{bics.10}
\mu_0 (t;h)=E_0(t)+o(1),\ h\to 0,\hbox{ uniformly in }t.
\end{equation}
\begin{equation}\label{bics.11}\begin{split}
&\mu_0 (t;h) \hbox{ is a simple eigenvalue and}\\ &\sigma (P_0(t))\cap
[E_0(t)-\delta (h),E_0(t)+\delta (h)]=\{ \mu_0 (t;h) \}.
\end{split}
\end{equation}
Here, as in Section \ref{rest}, $\delta (h)>0$ is small but not
exponentially small,
\begin{equation}\label{bics.12}
\ln \delta (h)\ge -o(1)/h,\ h\to 0.
\end{equation}

\par We restrict the spectral parameter $z$ to $D(\mu_0 (t),\delta
(h)/2)$. In this region we have,
\begin{equation}\label{bics.13}
(z-P_0(t))^{-1}={\cal O}\left(\frac{1}{|z-\mu_0 (t)|} \right):\ L^2\to
{\cal D}(P_0(t)),
\end{equation}
and more generally,
\begin{equation}\label{bics.14}
  \partial _t^k(z-P_0(t))^{-1}={\cal O}\left(\frac{1}{|z-\mu_0 (t)|^{1+k}} \right):\ L^2\to
  {\cal D}(P_0(t)).
\end{equation}
Strictly speaking, we work on sufficiently small time intervals, where
we can replace $P_0(t)$ with the unitarily equivalent operator
$U(t)^{-1}P_0(t)U(t)$, where $U(t):L^2(M_0(t))\to L^2(M_0(t_0))$ is
induced by a diffeomorphism $\kappa _t:M_0(t_0)\to M_0(t)$, depending
smoothly on $t$.
The spectral projection $\Pi _0(t)$, associated to $(P_0(t),\mu_0(t))$ is
given by
\begin{equation}\label{bics.15}
\Pi _0(t)=\frac{1}{2\pi i}\int_{\partial D(\mu_0
  (t),r)}(z-P_0(t))^{-1}dz,\ 0<r\le \delta (h)/2,
\end{equation}
and choosing $r=\delta (h)/2$, we see that
\begin{equation}\label{bics.16}
\partial _t^k\Pi _0(t)={\cal O}(\delta (h)^{-k}):\ L^2\to {\cal D}(P_0(t)).
\end{equation}
It follows that we can choose a normalized eigenfunction $e_0(t)$:
\begin{equation}\label{bics.17}
P_0(t)e_0(t)=\mu_0 (t)e_0(t),\ \|e_0\|_{L^2}=1,
\end{equation}
such that
\begin{equation}\label{bics.18}
\partial _t^ke_0(t)={\cal O}(\delta (h)^{-k})\hbox{ in }{\cal
  D}(P_0(t)),\ k=1,2,...,
\end{equation}
Now it is classical that
$$
\mu_0 (t)=(P_0(t)e_0(t)|e_0(t)),
$$
$$
\partial _t\mu_0 (t)=(\partial _tP_0e_0|e_0)+(P_0\partial _te_0|e_0)+(P_0e_0|\partial _te_0),
$$
where the sum of the last two terms is equal to 0:
$$
(\partial _te_0|P_0 e_0)+(P_0e_0|\partial _te_0)=
\mu_0 (t)((\partial _te_0|e_0)+(e_0|\partial _te_0))=\mu_0 (t)\partial
_t(e_0|e_0)=0.
$$
Thus,
\begin{equation}\label{bics.19}
\partial _t\mu_0 (t)=(\partial _tP_0e_0|e_0)={\cal O}(1),
\end{equation}
and after differentiating in $t$:
$$
\partial _t^k \mu_0 (t)={\cal O}(\delta (h)^{-k+1}),\ k=1,2,...
$$
For our purposes, it will be enough to work with the weaker estimate
\begin{equation}\label{bics.20}
\partial _t^k\mu_0 (t)={\cal O}(\delta (h)^{-k}),\ k=1,2,...
\end{equation}

\par Before discussing shape resonances, it will be convenient to
discuss some simple symmetry properties. In \cite{HeSj86}, (7.17) it
was shown that
$$
(u|v)_{H(\Lambda _G)}=(Bu|v)_{L^2({\bf R}^n)},\ u,v\in H(\Lambda _{G}),
$$
where $B:\, H(\Lambda _G)\to H(\Lambda _{-G})$ is the sum of an
elliptic Fourier integral operator and a nop of order 1. (Here $G$
denotes the function ``$\upsilon G$" in Section \ref{rest}, where the
parameter ``$\upsilon $' is
fixed according to (\ref{rest.44}).) Taking complex
conjugates and exchanging $u$ and $v$, we get
\begin{equation}\label{bics.21}
(u|v)_{H(\Lambda _G)}=(u|Bv)_{L^2({\bf R}^n)},\ u,v\in H(\Lambda _G),.
\end{equation}
Write
\begin{equation}\label{bics.22}\langle u|v\rangle=\int_{{\bf R}^n} uv dx \end{equation}
for the bilinear scalar product on $L^2$, so that
$$
\langle u|v\rangle =(u|\Gamma v)_{L^2},\hbox{ where }\Gamma v:=\overline{v}.
$$
\begin{prop}\label{bics1}
We have $\Gamma ={\cal O}(1):\, H(\Lambda _G)\to H(\Lambda
_{\check{G}})$, where $\check{G}(x,\xi ):=G(x,-\xi )$.
\end{prop}
\begin{proof}
This follows from 3 easily checked facts, where $\check{G}(x,\xi
)=G(x,-\xi )$:
\begin{itemize}
\item[1)] $(x,\xi )\in \Lambda _{\check{G}}\Longleftrightarrow
  (\overline{x},-\overline{\xi })\in \Lambda _G.$
\item[2)] If $T$ is an FBI-transformation adapted to $\Lambda _{\check{G}}$,
  then for $u\in H(\Lambda _G)$,
$$
T\Gamma u(\alpha )=\overline{\widetilde{T}u(\overline{\alpha
  }_x,-\overline{\alpha }_\xi )},\ \alpha \in \Lambda _{\check{G}},
$$
where $\widetilde{T}$ is an FBI-transformation adapted to $\Lambda _G$.
\item[3)] Let $H$ be the function on $\Lambda _G$, defined in
  (\ref{nore.7}) and let $\check{H}$ be the corresponding function on
  $\Lambda _{\check{G}}$. Then
$$
\check{H}(x,\xi )=H(\overline{x},-\overline{\xi }).
$$
\end{itemize}
\end{proof}

Then,
\begin{equation}\label{bics.23}
\langle u|\Gamma Bv\rangle=(u|v)_{H(\Lambda _G)},\ u,v\in H(\Lambda _G),
\end{equation}
and here $\Gamma B$ is an antilinear bijection $H(\Lambda _G)\to
H(\Lambda _{-\check{G}})$ with $\Gamma B$ and $(\Gamma B)^{-1}$
uniformly bounded.

\par Since $P=P^\mathrm{t}$ is symmetric (with ``t'' indicating
transpose for the bilinear scalar product),
$$\langle Pu| v\rangle= \langle u|Pv\rangle , $$
and hence $p(x,-\xi )=p(x,\xi )$ (also clear from the explicit formula
$p(x,\xi )=\xi ^2+V(x)$), we see that $-\check{G}$ is also an escape
function and hence also $\frac{1}{2}G-\frac{1}{2}\check{G}$. Replacing
$G$ with the latter we get a new escape function $G$ satisfying
\begin{equation}\label{bics.24}
-\check{G}=G.
\end{equation}
Now $\Gamma B$ becomes an antilinear bijection: $H(\Lambda _G)\to
H(\Lambda _G)$, uniformly bounded with its inverse. Replacing $v$ with
$(\Gamma B)^{-1}v$ in (\ref{bics.23}), we get
\begin{equation}\label{bics.25}
\langle u|v\rangle =(u|(\Gamma B)^{-1}v)_{H(\Lambda _G)},\ u,v\in
H(\Lambda _G).
\end{equation}
Then $\langle u|v\rangle$ is a bilinear nondegenerate scalar product
on $H(\Lambda _G)$
(In the case of ordinary complex scaling, this is seen more directly
by a shift of contour in (\ref{bics.22}).)

We resume the earlier discussion with $G=G_{\mathrm{sbd}}$ (assuming
for simplicity that the parameter ``$\upsilon $'' in Section \ref{rest} is
equal to 1) and apply Propositions \ref{rest6}, \ref{rest7} with
$J=[\mu _0(t;h)-\delta (h)/2,\mu _0(t;h)+\delta (h)/2]$. Let
\begin{equation}\label{bics.26}
  \Omega (t):=\{ z\in D(\mu _0(t;h),\delta (h)/2);\, \Im z\ge -\epsilon
  _\vartheta/C  \},
\end{equation}
where we recall that $\epsilon _\vartheta =(\epsilon
  /h)^{\vartheta }\epsilon$ .
Then $P(t)$ has a unique resonance $\lambda _0(t)=\lambda _0(t;h)$ in
$\Omega (t)$. It is simple and
\begin{equation}\label{bics.27}
\lambda _0(t)-\mu _0(t)=\widetilde{{\cal O}}(e^{-2S_t/h}),\
S_t:=d_t(U_0(t),\partial \ddot{\mathrm{O}}(t)).
\end{equation}
(\ref{rest.106}) gives
\begin{equation}\label{bics.28}
(z-P(t))^{-1}={\cal O}\left(\frac{h}{\delta (h)}\right)+{\cal O}(1)+{\cal
    O}\left( \frac{h}{z-\lambda _0(t)}\right):\ m_\epsilon
  ^{\frac{1}{2}}{\cal H}_\mathrm{sbd}\to m_\epsilon
  ^{-\frac{1}{2}}{\cal D}_\mathrm{sbd},
\end{equation}
where the first two terms in the right hand side are holomorphic in
$z$. In addition to (\ref{rest.24}), (\ref{rest.65.5}) and the
assumption $\ln \delta (h)\ge -o(1)/h$, we assume from now on that
\begin{equation}\label{bics.28.5}
\epsilon \le \delta (h),
\end{equation}
so that the first two terms in (\ref{bics.28}) drop out when
$|z-\lambda _0(t)|\le \epsilon _\vartheta $.

\par We have the spectral projection
\begin{equation}\label{bics.29}
\pi _0(t)=\frac{1}{2\pi i}\int_{\partial D(\lambda _0(t),r)}(z-P(t))^{-1}dz,
\end{equation}
where $0<r\le \epsilon _\vartheta /(2C)$ and choosing the maximal
value for $r$, we get from (\ref{bics.28}),
\begin{equation}\label{bics.30}
\pi _0(t)={\cal O}(h):\ m_\epsilon ^{\frac{1}{2}}{\cal
  H}_\mathrm{sbd}\to m_\epsilon ^{-\frac{1}{2}}{\cal D}_\mathrm{sbd}.
\end{equation}
For the higher $t$ derivatives, we write
\begin{equation}\label{bics.31}
  \partial _t^k\pi _0(t)=\frac{1}{2\pi i}\int_{\partial D(\lambda
    _0(t),r)}\partial _t^k(z-P(t))^{-1}dz,
\end{equation}
where the integrand is a linear combination of terms,
\begin{equation}\label{bics.32}
(z-P(t))^{-1}(\partial _t^{k_1}P)(z-P(t))^{-1}(\partial _t^{k_2}P)...
(\partial _t^{k_\ell}P)(z-P(t))^{-1},
\end{equation}
with $k_j\ge 1$, $k_1+k_2+...+k_\ell =k$. In view of (\ref{rest.37}),
we have $m_\epsilon \asymp h$ on every fixed compact set and since $P$
is independent of $t$ outside such a set, we conclude that
$$
\partial _t^{k_j}P(t)={\cal O}(1/h):\ m_\epsilon ^{-\frac{1}{2}}{\cal
  H}_{{\mathrm{sbd}}}
\to m_\epsilon ^{\frac{1}{2}}{\cal H}_{{\mathrm{sbd}}}.
$$
Also, for $z\in \Omega (t)$, we have
$$
(z-P(t))^{-1}=\frac{{\cal O}(h)}{z-\lambda _0(t)}: m_\epsilon ^{\frac{1}{2}}{\cal
  H}_{{\mathrm{sbd}}}
\to m_\epsilon ^{-\frac{1}{2}}{\cal H}_{{\mathrm{sbd}}}, \hbox{ when } |z-\lambda _0(t)|\lesssim
\epsilon _\vartheta .
$$
Hence the term (\ref{bics.32}) is
$$
\frac{{\cal O}(h)}{(z-\lambda _0(t))^{\ell +1}}: m_\epsilon ^{\frac{1}{2}}{\cal
  H}_{{\mathrm{sbd}}}
\to m_\epsilon ^{-\frac{1}{2}}{\cal D}_{{\mathrm{sbd}}},
$$
so the integrand in (\ref{bics.31}) is
$$
\frac{{\cal O}(h)}{(z-\lambda _0(t))^{k+1}}: m_\epsilon ^{\frac{1}{2}}{\cal
  H}_{{\mathrm{sbd}}}
\to m_\epsilon ^{-\frac{1}{2}}{\cal D}_{{\mathrm{sbd}}},
$$
and we conclude that
\begin{equation}\label{bics.33}
\partial _t^k\pi _0(t)=\frac{{\cal O}(h)}{\epsilon _\vartheta  ^k}:
m_\epsilon ^{\frac{1}{2}}{\cal
  H}_{\mathrm{sbd}}
\to m_\epsilon ^{-\frac{1}{2}}{\cal D}_{\mathrm{sbd}}.
\end{equation}

\par Let us fix $t\in I$ for a while and write $P=P(t)$. $P$ is
symmetric for the bilinear scalar product (\ref{bics.25}) and so is
the Grushin operator ${\cal P}$ in (\ref{rest.81}) ($m=1$) if we use
$$
\langle \begin{pmatrix}u\\ u_-\end{pmatrix}
| \begin{pmatrix}\widetilde{u}\\ \widetilde{u}_-\end{pmatrix} \rangle
= \langle  u| \widetilde{u}\rangle + u_-\widetilde{u}_-,
$$
and take care to use real eigenfunctions of $P_0$, when defining
$R_\pm$. Then the inverse ${\cal E}:{\cal H}_{\mathrm{sbd}}\times {\bf
  C}\to {\cal D}_{\mathrm{sbd}}\times {\bf C}$ is symmetric:
\begin{equation}\label{bics.34}
E^{\mathrm{t}}=E,\ E_+^{\mathrm{t}}=E_-.
\end{equation}

Using that
$$
(z-P)^{-1}=-E(z)+E_+(z)E_{-+}(z)^{-1}E_-(z)
$$
in (\ref{bics.29}), we get (with $\lambda _0=\lambda _0(t)$ etc.),
\begin{equation}\label{bics.35}
\pi _0=\frac{1}{E'_{-+}(\lambda _0)}E_+(\lambda _0)E_-(\lambda _0).
\end{equation}
Here, by (\ref{rest.94}) and the Cauchy inequality,
\begin{equation}\label{bics.36}
E'_{-+}=1+\widetilde{{\cal O}}(1)e^{-2S_t/h}/\epsilon _\vartheta.
\end{equation}
From (cf.\ (\ref{rest.65.5})), we get
\begin{equation}\label{bics.37}
\epsilon _\vartheta \ge e^{-S_t/h}.
\end{equation}

\par Now, by (\ref{rest.95}), (\ref{bics.34}) we have
\begin{equation}\label{bics.38}
E_+(\lambda _0)v_+=v_+e_+,\ E_-(\lambda _0)v=\langle v| e_+\rangle ,
\end{equation}
\begin{equation}\label{bics.39}
m_\epsilon ^{\frac{1}{2}}(e_+-\chi e_0)=\widetilde{{\cal
    O}}(e^{-S_t/h})\hbox{ in }{\cal D}_\mathrm{sbd}.
\end{equation}
Since $m_\epsilon \ge \epsilon _\vartheta $, this implies that,
\begin{equation}\label{bics.39.3}
e_+-\chi e_0=\widetilde{{\cal O}}(e^{-S_t/h}/\epsilon _\vartheta
^{1/2})\hbox{ in }{\cal D}_\mathrm{sbd},
\end{equation}
\begin{equation}\label{bics.39.6}
m_\epsilon ^{-\frac{1}{2}}(e_+-\chi e_0)=\widetilde{{\cal O}}(e^{-S_t/h}/\epsilon _\vartheta
)\hbox{ in }{\cal D}_\mathrm{sbd}.\end{equation}

\par From (\ref{bics.35}), (\ref{bics.38}), we get
\begin{equation}\label{bics.40}
\pi _0u=\frac{1}{E'_{-+}(\lambda _0)}\langle u|e_+\rangle e_+,
\end{equation}
and in particular that $e_+$ is a resonant state.
The reproducing property $\pi _0^2=\pi _0$ means that $\pi _0e_+=e_+$,
which by (\ref{bics.40}) is equivalent to
\begin{equation}\label{bics.41}
\frac{1}{E'_{-+}(\lambda _0)}\langle e_+|e_+\rangle =1.
\end{equation}
Put
\begin{equation}\label{bics.42}
e^0=(E'_{-+}(\lambda _0))^{-\frac{1}{2}}e_+=(1+\widetilde{{\cal
    O}}(e^{-2S_t/h}/\epsilon _\vartheta ))e_+.
\end{equation}
Then
\begin{equation}\label{bics.43}
\pi _0u=\langle u|e^0\rangle e^0,\ \ \langle e^0|e^0\rangle=1.
\end{equation}

\par We next estimate the $t$-derivatives of $e^0=e^0(t)$. We work in
a small neighborhood of a variable point $t_0\in I$. Then
$$
f(t)=\pi _0(t)e^0(t_0)
$$
is collinear to $e^0(t)$ and we recover $e^0(t)$ from the formula
\begin{equation}\label{bics.44}
e^0(t)=\langle f(t)|f(t)\rangle^{-\frac{1}{2}}f(t).
\end{equation}
By (\ref{bics.33}), we have
\begin{equation}\label{bics.45}
\| m_\epsilon ^{1/2}\partial _t^k f(t)\|_{{\cal
    D}_{\mathrm{sbd}}}={\cal O}(h\epsilon _\vartheta ^{-k})\|
m_\epsilon ^{-1/2}e^0(t_0)\|_{{\cal H}_\mathrm{sbd}}.
\end{equation}

\par Using that $m_\epsilon \ge \epsilon _\vartheta $, we conclude
that
\begin{equation}\label{bicsny.46}
\|\partial _t^kf(t)\|_{{\cal D}_\mathrm{sbd}}={\cal
  O}(1)\frac{h}{\epsilon _\vartheta ^{1+k}}\| e^0\|_{{\cal
    H}_\mathrm{sbd}}={\cal O}(1)\frac{h}{\epsilon _\vartheta ^{k+1}},
\end{equation}
for $k\ge 1$. For $k=0$ we have $\| f\|_{{\cal D}_\mathrm{sbd}}=1$ so
we have the simpler but weaker estimate,
\begin{equation}\label{bicsny.47}
\|\partial _t^k f(t)\|_{{\cal D}_\mathrm{sbd}}={\cal
  O}(1)\left(\frac{h}{\epsilon _\vartheta ^2} \right)^k, k\ge 0.
\end{equation}
From (\ref{bics.44}) it then follows that
\begin{equation}\label{bicsny.48}
\|\partial _t^k e^0(t)\|_{{\cal D}_\mathrm{sbd}}={\cal
  O}(1)\left(\frac{h}{\epsilon _\vartheta ^2} \right)^k, k\ge 0,
\end{equation}
and this implies that
\begin{equation}\label{bicsny.49}
\partial _t^k \pi _0(t)={\cal
  O}(1)\left(\frac{h}{\epsilon _\vartheta ^2} \right)^k:\, {\cal
  H}_\mathrm{sbd}\to {\cal D}_\mathrm{sbd},\ k\ge 0.
\end{equation}

Next, we estimate $\partial ^k\lambda _0(t)$,
$k=1,2,...$. We start with
\begin{equation}\label{bics'.1}
\lambda _0(t)=\langle P(t)e^0(t)|e^0(t)\rangle .
\end{equation}
By the symmetry of $P(t)$ and the fact that $\langle
e^0(t)|e^0(t)\rangle =1$, we get
\begin{equation}\label{bics'.2}
\partial _t\lambda _0(t)=\langle (\partial _tP(t))e^0(t)|e^0(t)\rangle
={\cal O}(1).
\end{equation}
It follows from (\ref{bicsny.48}) that
$$
\partial _t^{k+1}\lambda _0={\cal O}(1)\left( \frac{h}{\epsilon
    _\vartheta ^2} \right)^k,\ k\ge 0,
$$
and hence,
\begin{equation}\label{bics'.9}
\partial _t^k\lambda _0={\cal O}(1)\left( \frac{h}{\epsilon
    _\vartheta ^2} \right)^{(k-1)_+},\ k\ge 0.
\end{equation}
Recall from (\ref{rest.24}), (\ref{bics.28.5}), that
\begin{equation}\label{bicsny.54}
\epsilon \le \min (h/C,\delta ),\ C\gg 0.
\end{equation}
$\lambda _0$ does not depend on the choice of $\epsilon $ and if we
make the maximal choice in (\ref{bicsny.54}), we get $\epsilon
_\vartheta \asymp \min (1,\delta /h)^\vartheta \min (h,\delta )$ and
(\ref{bics'.9}) gives
\begin{equation}\label{bics'.14}
\partial _t^k\lambda _0={\cal O}(1)\left(h\min (1,\delta
  /h)^{2+2\vartheta }\right)^{-(k-1)_+}
\end{equation}

Then (\ref{bics.20}) implies similar subexponential estimates for $\partial
^k(\lambda _0-\mu _0)$, $k\ge 1$. Combining this with (\ref{bics.27})
and elementary interpolation estimates, we get
\begin{equation}\label{bics'.15}
\partial _t^k(\lambda _0-\mu _0)=\widetilde{{\cal O}}(e^{-2S_t/h}),\
k\ge 0,
\end{equation}
and we can then use (\ref{bics.20}) again, to get
\begin{equation}\label{bics'.16}
\partial _t^k\lambda _0={\cal O}(\delta (h)^{-k}),\ k\ge 1.
\end{equation}

We next study $(\lambda _0(t)-P(t))^{-1}(1-\pi _0(t))$ and its
derivatives. In the discussion leading to (\ref{bics.33}) we have
seen that
$$
m_\epsilon ^{\frac{1}{2}}\partial _t^k (z-P(t))^{-1}m_\epsilon
^{\frac{1}{2}}=\frac{{\cal O}(h)}{(z-\lambda _0)^{k+1}}:\ {\cal
  H}_\mathrm{sbd}\to {\cal D}_\mathrm{sbd},\ |z-\lambda _0|\le
\epsilon _\vartheta /C
$$
and hence
$$
\partial _t^k (z-P(t))^{-1}=\frac{{\cal O}(h)}{\epsilon _\vartheta (z-\lambda _0)^{k+1}}:\ {\cal
  H}_\mathrm{sbd}\to {\cal D}_\mathrm{sbd},\ |z-\lambda _0|\le
\epsilon _\vartheta /C.
$$
Combining this with (\ref{bicsny.49}), we get
$$
\partial _t^k((z-P(t))^{-1}(1-\pi _0(t)))={\cal
  O}(1)\sum_{k_1+k_2=k}\frac{h}{\epsilon _\vartheta (z-\lambda
  _0)^{k_1+1}}\left(\frac{h}{\epsilon _\vartheta ^2} \right)^{k_2}:\
{\cal H}_\mathrm{sbd}\to {\cal D}_{\mathrm{sbd}}.
$$
When $|z-\lambda _0|\asymp \epsilon _\vartheta $, the majorant is
\begin{multline*}
\le {\cal O}(1)\sum_{k_1+k_2=k}\frac{h}{\epsilon _\vartheta
  ^2}\frac{1}{\epsilon _\vartheta ^{k_1}}\left(\frac{h}{\epsilon
    _\vartheta ^2} \right)^{k_2}\\
\le {\cal O}(1)\sum_{k_1+k_2=k}\frac{h}{\epsilon _\vartheta
  ^2}\left(\frac{h}{\epsilon _\vartheta^2}\right) ^{k_1}\left(\frac{h}{\epsilon
    _\vartheta ^2} \right)^{k_2}\le {\cal O}(1)\left(\frac{h}{\epsilon
  _\vartheta ^2} \right)^{k+1}.
\end{multline*}
Since $(z-P(t))^{-1}(1-\pi _0(t))$ and its $t$-derivatives are
holomorphic near $z=\lambda _0(t)$ the maximum principle gives for
$|z-\lambda _0(t)|\le \epsilon _\vartheta /C$:
\begin{equation}\label{bics.63}
\partial _t^k((z-P(t))^{-1}(1-\pi _0(t)))={\cal
  O}(1)\left(\frac{h}{\epsilon _\vartheta ^2} \right)^{k+1}:\ {\cal
  H}_\mathrm{sbd}\to {\cal D}_\mathrm{sbd}.
\end{equation}
With the Cauchy inequalities, this extends to
\begin{multline}\label{bics.64}
\partial _z^\ell\partial _t^k((z-P(t))^{-1}(1-\pi _0(t)))=\\ {\cal
  O}(1)\epsilon _\vartheta ^{-\ell}\left(\frac{h}{\epsilon _\vartheta
    ^2} \right)^{k+1}
={\cal O}(1)\left(\frac{h}{\epsilon _\vartheta ^2} \right)^{k+\ell +1}
:\ {\cal
  H}_\mathrm{sbd}\to {\cal D}_\mathrm{sbd}.
\end{multline}

Finally we put $z=\lambda _0(t)$ and get with the natural meaning of ``lincomb''
\begin{multline*}
\partial _t^k\left( (\lambda _0(t)-P_0(t))^{-1}(1-\pi _0(t)) \right)\\
=\underset{m+\ell_1+..+\ell_\lambda =k,\atop \ell_j\ge
  1}{\mathrm{lincomb\,}}
\partial _t^m\partial _z^\lambda \left( (z-P_0(t))^{-1}(1-\pi _0(t))
\right)_{z=\lambda _0(t)}(\partial _t^{\ell_1}\lambda _0)... (\partial
_t^{\ell_\lambda }\lambda _0)
\end{multline*}
Using (\ref{bics.64}), (\ref{bics'.16}), we see that the ${\cal L}({\cal H}_\mathrm{sbd},{\cal D}_\mathrm{sbd})$-norm of
the general term is
\begin{multline*}
{\cal O}(1)\epsilon _\vartheta ^{-\lambda }\left(\frac{h}{\epsilon
    _\vartheta ^2} \right)^{m+1}\delta ^{-(\ell_1+...+\ell_\lambda )}\\
\le {\cal O}(1)\left(\frac{1}{\delta \epsilon _\vartheta }
\right)^{k-m}\left(\frac{h}{\epsilon _\vartheta ^2} \right)^{m+1}\le \\
{\cal O}(1)\frac{h}{\epsilon _\vartheta ^2}\left(\frac{\max (h,\epsilon
  _\vartheta /\delta ) }{\epsilon _\vartheta ^2} \right)^k,
\end{multline*}
where we used that $\lambda \le \ell_1+...\ell_\lambda =k-m$. Thus for
every $k\in {\bf N}$,
\begin{equation}\label{bics.64.5}
\partial_t ^k\left( (\lambda _0(t)-P_0(t))^{-1}(1-\pi _0(t)\right) =
{\cal O}(1)\frac{h}{\epsilon _\vartheta ^2}\left(\frac{\max (h,\epsilon
  _\vartheta /\delta ) }{\epsilon _\vartheta ^2} \right)^k
:\ {\cal H}_\mathrm{sbd}\to {\cal D}_\mathrm{sbd}.
\end{equation}
When $\epsilon $ is exponentially small,
$\epsilon =\exp (-1/{\cal O}(h))$, or more generally when
$\epsilon _\vartheta \le \delta h $, the estimate simplifies to
\begin{equation}\label{bics.64.7}
\partial_t ^k\left( (\lambda _0(t)-P_0(t))^{-1}(1-\pi _0(t)\right) =
{\cal O}(1)\left(\frac{h}{\epsilon _\vartheta ^2} \right)^{k+1}
:\ {\cal H}_\mathrm{sbd}\to {\cal D}_\mathrm{sbd},\ k\ge 0.
\end{equation}

We next consider formal adiabatic solutions in the spirit of
Proposition \ref{1eig1}. For the moment, we let $\epsilon $,
$\varepsilon $ be independent parameters.
\begin{prop}\label{bics2}
Under the assumptions above, there exist two for\-mal asy\-mp\-to\-tic series,
\begin{equation}\label{bics.65}
\nu (t,\varepsilon )\sim \nu _0(t)+\varepsilon \nu _1(t)+\varepsilon
^2 \nu _2(t)+...\ \hbox{in }C^\infty (I;{\cal D}_{\mathrm{sbd}}),
\end{equation}
\begin{equation}\label{bics.66}
\lambda (t,\varepsilon )\sim \lambda _0(t)+\varepsilon \lambda _1(t)+\varepsilon
^2 \lambda _2(t)+...\ \hbox{in }C^\infty (I),
\end{equation}
such that
\begin{equation}\label{bics.67}
(\varepsilon D_t+P(t)-\lambda (t,\varepsilon ))\nu (t,\varepsilon
)\sim 0
\end{equation}
as a formal asymptotic series in $C^\infty (I;{\cal H}_\mathrm{sbd})$. Here,
\begin{equation}\label{bics.68}
  \partial _t^k\nu _j=   {\cal O}(1) (h/\hat{\epsilon }_\vartheta ^2)
  ^{2j+k}\hbox{ in }{\cal D}_\mathrm{sbd},\ j\ge 0,\ k\ge 0,
\end{equation}
\begin{equation}\label{bics.69}
\partial _t^k\lambda _j= {\cal O}(1)(h/\hat{\epsilon }_\vartheta ^2
)^{2j-1+k},\ j\ge 1, k\ge 0.
\end{equation}
Here,
\begin{equation}\label{70.5}
\hat{\epsilon }_\vartheta :=\frac{\epsilon _\vartheta }{\max
  (1,\epsilon _\vartheta /(\delta h) )^{1/2}}=\min (\epsilon _\vartheta
,(\epsilon _\vartheta\delta h) ^{1/2}).
\end{equation}
\end{prop}
\begin{proof}
We sacrifice optimal sharpness for simplicity and work with the weaker
form of (\ref{bics.64.5}):
\begin{equation}\label{bics.70}
\partial_t ^k\left( (\lambda _0(t)-P_0(t))^{-1}(1-\pi _0(t)\right)
={\cal O}(1)(h/\hat{\epsilon }_\vartheta ^2)^{1+k},
\end{equation}
which is equivalent to (\ref{bics.64.7}) in the most interesting case
when $\epsilon _\vartheta \le \delta h $.
We shall use (\ref{bics'.16}): $\partial _t^k\lambda _0={\cal
  O}(\delta ^{-k})$ and the following weakened form of
(\ref{bicsny.48}):
\begin{equation}\label{bics.71}
\partial _t^k e^0(t)={\cal O}(1)(h/\hat{\epsilon }_\vartheta
^2)^k\hbox{ in }{\cal D}_\mathrm{sbd}.
\end{equation}

We follow the proof of Proposition \ref{1eig1} and annihilate
successively the powers of $\varepsilon $ in the right hand side of
(\ref{1eig.4}).  The first equation is then
\begin{equation}\label{bics.72}
(P(t)-\lambda _0(t))\nu _0(t)=0,
\end{equation}
so we choose
\begin{equation}\label{bics.73}
\nu _0(t)=\theta _0(t)e^0(t)
\end{equation}
with the condition
\begin{equation}\label{bics.74}
  \theta _0(t)\asymp 1,\ \partial _t^k\theta _0={\cal O}(1)(h/\hat{\epsilon }_\vartheta ^2 )^k,
\end{equation}
so that $\nu _0$ satisfies (\ref{bics.68}).
Then the $\varepsilon ^0$ term in (\ref{1eig.4}) vanishes.

\par To annihilate the $\varepsilon ^1$-term, we need to solve (\ref{1eig.8})
which is solvable precisely when (cf.\ (\ref{1eig.9}))
$$
0=\langle \lambda _1(t)\nu _0(t)-D_t\nu _0(t)|e^0(t)\rangle
=\theta _0(t)\lambda _1(t)-\langle D_t\nu _0(t)|e^0(t)\rangle .
$$
Here,
$$
\langle D_t\nu _0(t)|e^0(t)\rangle =D_t\theta _0(t)+\theta
_0(t)\langle D_te^0(t)|e^0(t)\rangle =D_t\theta _0(t),
$$
since
$$\langle D_te^0|e^0\rangle =
\frac{1}{2}D_t \langle e^0|e^0\rangle=0,\hbox{ recalling that }\langle
e^0| e^0\rangle=1.
$$
Thus, $\lambda _1$ should satisfy $\theta _0(t)\lambda
_1(t)-D_t\theta _0(t)=0$,
\begin{equation}\label{bics.75}
\lambda _1(t)=\frac{D_t\theta _0}{\theta _0},
\end{equation}
and in particular,
$\partial _t^k \lambda _1(t)={\cal O}(1)(h/\hat{\epsilon }_\vartheta
^2)^{1+k}$, so $\lambda _1$ satisfies (\ref{bics.69}).
\begin{remark}\label{bics3}
A natural choice of $\theta _0$ is $\theta _0=1$. Then we get $\lambda_1=0$ in {\rm (\ref{bics.75})}.
\end{remark}
 With
this unique choice of $\lambda _1$, we can solve (\ref{1eig.8}) and
the general solution is
\begin{equation}\label{bics.76}
\nu _1(t)=(P(t)-\lambda _0(t))^{-1}(1-\pi _0(t))(\lambda _1(t)\nu
_0-D_t\nu _0(t))+z(t)e^0(t),
\end{equation}
where we are free to choose $z(t)$, and we will take $z(t)=0$ for
simplicity. From (\ref{bics.70}), the estimate (\ref{bics.69}) for
$\lambda _1$ and (\ref{bics.68}) for $\nu _0$, we get
\begin{equation}\label{bics.77}
\partial _t^k \nu _1
={\cal O}(1) (h/\hat{\epsilon }_\vartheta ^2)^{2+k} \hbox{ in }{\cal
  D}_\mathrm{sbd},\ k\ge 0,
\end{equation}
i.e.\ $\nu _1$ satisfies (\ref{bics.68}).

\par The equation for annihilating the $\varepsilon^j$-term in
(\ref{1eig.4}) is
\begin{equation}\label{bics.78}
(P(t)-\lambda _0(t))\nu _j=(\lambda _1-D_t)\nu _{j-1}+\lambda _2\nu
_{j-2}+...+\lambda _{j-1}\nu _1+\lambda _j\nu _0.
\end{equation}
Let $N\ge 2$ and assume that we have already constructed $\nu _j$,
$\lambda _j$ for $j\le N-1$, satisfying (\ref{bics.68}),
(\ref{bics.69}), (\ref{bics.78}) for $j\le N-1$. Consider
(\ref{bics.78}) for $j=N$. The condition for finding a solution
$\nu _N$ is that the right hand side is orthogonal (for
$\langle \cdot |\cdot \rangle$) to $e^0$ and since
$\langle \nu _0|e^0\rangle = \theta _0(t)$, we get
\begin{equation}\label{bics.79}
\lambda _N=\theta _0^{-1}\langle (D_t-\lambda _1)\nu _{N-1}+\lambda _2\nu
_{N-2}+...+\lambda _{N-1}\nu _1\, |\, e^0\rangle.
\end{equation}
Here
\begin{equation}\label{bics.80}
\partial _t^kD_t\nu _{N-1}={\cal O}(1)(h/\hat{\epsilon }_\vartheta ^2
)^{2(N-1)+k+1}={\cal O}(1)(h/\hat{\epsilon }_\vartheta ^2)^{2N-1+k}
\end{equation}
and for $1\le \ell \le N-1$:
\begin{equation}\label{bics.81}
\partial _t^k (\lambda _\ell \nu _{N-\ell}) ={\cal O}(1)(h/\hat{\epsilon }_\vartheta ^2
)^{2\ell -1 +2(N-\ell )+k}={\cal O}(1)(h/\hat{\epsilon }_\vartheta ^2)^{2N-1+k}.
\end{equation}
Using also (\ref{bics.74}), we see that $\lambda _N$ satisfies
(\ref{bics.69}).

\par We can now solve for $\nu _N$ in (\ref{bics.78}):
\begin{multline}\label{bics.82}
\nu _N=(P(t)-\lambda _0(t))^{-1}(1-\pi _0(t))\\((\lambda _1-D_t)\nu _{N-1}+\lambda _2\nu
_{N-2}+...+\lambda _{N-1}\nu _1+\lambda _N\nu _0)+z(t)e^0(t).
\end{multline}
Again we take $z=0$ for simplicity and get, using (\ref{bics.70}),
(\ref{bics.80}), (\ref{bics.81}):
$$
\partial _t^k \nu _N={\cal O}(1)(h/\hat{\epsilon }_\vartheta ^2
)^{1+2N-1+k}=
{\cal O}(1)(h/\hat{\epsilon }_\vartheta ^2)^{2N+k},
$$
so $\nu _N$ satisfies (\ref{bics.68}) and this finishes the
inductive proof.
\end{proof}

\begin{remark}\label{bics4}
The construction of $\nu _j$, $\lambda _j$ is independent of the
choice of ambient spaces and if we choose $\epsilon $ maximal in
{\rm (\ref{bicsny.54})} we see as after that inequality that {\rm (\ref{bics.69})}
becomes
\begin{equation}\label{bics.83}
\partial _t^k\lambda _j= {\cal O}(1)
\left(\min (1,\delta /h)^\vartheta \min(\delta ,\min (1,\delta /h)^{1+\vartheta})
\right)^{-(2j+k-1)},\ j\ge 1,\ k\ge 0.
\end{equation}
When $\delta \le h$ this simplifies to
$$\partial _t^k\lambda _j= {\cal O}(1)\left(\left(\frac{h}{\delta }
  \right)^{1+\vartheta }\frac{1}{\delta } \right)^{2j-1+k}.$$
This can probably be improved as in the proof of {\rm (\ref{bics'.16})}.
\end{remark}

We continue the discussion under the assumptions of Proposition
\ref{bics2}. Put for $N\ge 1$
\begin{equation}\label{bics.84}
\nu ^{(N)}=\nu _0+\varepsilon \nu _1+...+\varepsilon ^N\nu _N,
\end{equation}
\begin{equation}\label{bics.85}
\lambda ^{(N)}=\lambda _0+\varepsilon \lambda _1+...+\varepsilon ^N
\lambda _N,\ \ N\ge 1.
\end{equation}
Then by construction (cf.\ (\ref{1eig.4})),
\begin{equation}\label{bics.86}
(\varepsilon D_t+P(t)-\lambda ^{(N)})\nu ^{(N)}=r^{(N+1)},
\end{equation}
where
\begin{equation}\label{bics.87}
r^{(N+1)}=\varepsilon ^{N+1}D_t\nu _N-\sum_{j,k\le N \atop j+k\ge N+1}\varepsilon
^{j+k}\lambda _j\nu _k.
\end{equation}
From the estimates in Proposition \ref{bics2}, we get
$$
r^{(N+1)}=
{\cal O}(1)\left(\varepsilon ^{N+1}(h/\hat{\epsilon }_\vartheta
  ^2)^{2N+1}
+\sum_{j,k\le N\atop j+k\ge N+1}\varepsilon ^{j+k}(h/\hat{\epsilon }_\vartheta ^2)^{2(j+k)-1}
 \right)
\hbox{ in }{\cal D}_\mathrm{sbd}.
$$

\par In the following, we assume that
\begin{equation}\label{bics.88}
    \frac{\varepsilon ^{\frac{1}{2}}h}{\hat{\epsilon }_\vartheta ^2}\ll 1 .
\end{equation}
Recall from (\ref{bics.12}) that $\delta =\delta (h)$ is small, but
not exponentially small and that $\epsilon _\vartheta =(\epsilon
/h)^\vartheta \epsilon $. Then (\ref{bics.88}) holds if we assume that
$\varepsilon $ is exponentially small:
\begin{equation}\label{bics.89}
0<\varepsilon \le {\cal O}(1)\exp \left(-1/(Ch)\right),\hbox{ for some }C>0,
\end{equation}
and choose
\begin{equation}\label{bics.90}
\epsilon \ge \varepsilon ^{\frac{1}{4(1+\vartheta ) }-\alpha },
\end{equation}
for some $\alpha \in ]0,1/(4(1+\vartheta ))[$.

\par Having assumed (\ref{bics.88}) we get $r^{(N+1)}={\cal O}(1)
\varepsilon ^{\frac{1}{2}}\left(\frac{\varepsilon ^{\frac{1}{2}}h}{\hat{\epsilon }_\vartheta ^2}\right)^{2N+1}$ in ${\cal D}_\mathrm{sbd}$ and more generally,
\begin{equation}\label{bics.91}
  \partial _t^k r^{(N+1)}={\cal O}(1)
  \varepsilon ^{\frac{1}{2}}\left(\frac{\varepsilon
      ^{\frac{1}{2}}h}{\hat{\epsilon }_\vartheta ^2}\right)^{2N+1}
  \left(\frac{h}{\hat{\epsilon }_\vartheta ^2} \right)^k\hbox{ in }{\cal D}_\mathrm{sbd}.
\end{equation}
Also, since $\| \nu _0(t)\|_{{\cal H}_\mathrm{sbd}}=\| e^0(t)\|_{{\cal
    H}_\mathrm{sbd}}$, we get
\begin{equation}\label{bics.91.2}
  \| \nu ^{(N)}(t)\|_{{\cal H}_\mathrm{sbd}}
  =(1+{\cal O}(\varepsilon h^2 /\hat{\epsilon }_\vartheta ^4))\| \nu
  _0(t)\|_{{\cal H}_\mathrm{sbd}}\asymp 1.
\end{equation}

\par Recall (\ref{far.28}) with the subsequent observation and the
choice of $\mu $ in (\ref{rest.26}):
\begin{equation}\label{bics.92}
-\Im (P(t)u|u)_{{\cal H}_\mathrm{sbd}}\ge -{\cal
  O}(\epsilon ^\infty )\| u\|^2_{{\cal H}_\mathrm{sbd}}.
\end{equation}

\par Let $I\ni t\mapsto u(t)\in H(\Lambda _{\epsilon G},\langle \xi \rangle)$
be continuous such that $\partial _tu$ is continuous with values in
$H(\Lambda _{\epsilon G},\langle \xi \rangle^{-1})$, $G=G_\mathrm{sbd}$. Assume
that $u$ is a solution of
$$
(\varepsilon D_t+P(t))u(t)=0.
$$
Then,
$$
\varepsilon \partial _t \| u(t)\|_{{\cal H}_\mathrm{sbd}}^2
=2\Im (P(t)u|u)\le {\cal O}(\epsilon ^\infty )\| u\|^2_{{\cal H}_\mathrm{sbd}},
$$
implying
$$
\| u(t)\|_{{\cal H}_\mathrm{sbd}}\le e^{{\cal O}(\epsilon
^\infty )(t-s)/\varepsilon }\| u(s)\|_{{\cal H}_\mathrm{sbd}},\ t\ge s.
$$
Assume
\begin{equation}\label{bics.92.5}
\epsilon \le {\cal O}(\varepsilon ^{1/N_0}),\hbox{ for some fixed }N_0>0.
\end{equation}
Then, 
\begin{equation}\label{bics.93}
\| u(t)\|_{{\cal H}_\mathrm{sbd}}\le e^{{\cal O}(\epsilon
^\infty )(t-s)}\| u(s)\|_{{\cal H}_\mathrm{sbd}},\ t\ge s.
\end{equation}

From (\ref{bics.92})and the fact that $P(t_0)-z:{\cal
  D}_\mathrm{sbd}\to {\cal H}_\mathrm{sbd}$ is Fredholm of index 0,
when $\Im z>0$, we see that
\begin{equation}\label{bics.94}
\| (P(t_0)-z)^{-1}\|_{{\cal L}({\cal
    H}_\mathrm{sbd},{\cal H}_\mathrm{sbd})}
\le \frac{1}{\Im z-y_\epsilon },\hbox{ for }\Im z>y_\epsilon ,
\end{equation}
where $y_\epsilon ={\cal O}(\epsilon ^\infty )$ can be chosen
independent of $t_0\in I$. By the Hille-Yosida theorem, $-iP(t_0)$
generates a strongly continuous semi-group leading to: If $u_0\in {\cal
  D}_\mathrm{sbd}$, then $\exists ! u\in C([0,+\infty [;{\cal
  D}_\mathrm{sbd})\cap
C^1([0,+\infty [;{\cal H}_\mathrm{sbd})
$
such that
\begin{equation}\label{bics.95}
(\varepsilon D_t+P(t_0))u(t)=0 \hbox{ for }t\ge 0,\ \ u(0)=u_0,
\end{equation}
where we entered the parameter $\varepsilon >0$ to conform to the
general discussion. Now $P(t)-P(t_0)$ is a smooth function of $t$ with
values in ${\cal L}({\cal H}_\mathrm{sbd},{\cal H}_\mathrm{sbd})$ and
an application of~\cite[Theorem 6.1 and Remark 6.2]{Ka70}, allows us to conclude 
that for every $u_0\in {\cal D}_\mathrm{sbd}$ and every $s\in I$, there exist
$u_0\in C(I\cap [s,\infty [;{\cal D}_\mathrm{sbd})\cap C^1(I\cap [s,\infty [;{\cal
  H}_\mathrm{sbd})$ such that
\begin{equation}\label{bics.96}
(\varepsilon D_t+P(t))u(t)=0\hbox{ for }s\le t\in I,\ \ u(s)=u_0.
\end{equation}
Again the solution satisfies (\ref{bics.93}).

\par This allows us to define the forward fundamental matrix $E(t,s)$,
$I\ni t\ge s\in I$ of
$\varepsilon D_t+P(t)$:
$$
\begin{cases}(\varepsilon D_t+P(t))E(t,s)=0,\ t\ge s,\\ E(t,t)=1
\end{cases}
$$
and from~\cite[Theorem 6.1 and Remark 6.2]{Ka70} we infer, in particular, that $E(t,s)$ is strongly continuous in the
${\cal H}_\mathrm{sbd}$-norm both in $t$ and $s$, such that 
\begin{equation}\label{bics.97}
\| E(t,s)\|_{{\cal L}({\cal H}_\mathrm{sbd},{\cal H}_\mathrm{sbd})}\le
\exp ((t-s){\cal O}(\epsilon ^\infty/\varepsilon  )),\ t\ge s,\ t,s\in I.
\end{equation}
If $v\in C (I; {\cal H}_\mathrm{sbd})$ vanishes for $t$ near $\inf I$,
we can solve $(\varepsilon D_t+P(t))u=v$
on $I$ by
$$
u(t)=\frac{i}{\varepsilon }\int _{\inf I}^tE(t,s)v(s) ds.
$$

\par Now return to (\ref{bics.84})--(\ref{bics.86}) with $\lambda _j$,
$\nu _j$ as in Proposition \ref{bics2} and $r^{(N+1)}$ satisfying
(\ref{bics.91}). We notice that
\begin{equation}\label{bics.98}
\lambda ^{(N)}=\lambda _0+{\cal O}(1)\varepsilon ^{\frac{1}{2}}\frac{\varepsilon
  ^{\frac{1}{2}}h}{\hat{\epsilon }_\vartheta ^2},
\end{equation}
and that this improves to
\begin{equation}\label{bics.99}
\lambda ^{(N)}=\lambda _0+{\cal O}(1)\varepsilon ^{\frac{1}{2}}\left(\frac{\varepsilon
  ^{\frac{1}{2}}h}{\hat{\epsilon }_\vartheta ^2}\right)^3,
\end{equation}
if we take
\begin{equation}\label{bics.100}
\theta _0=1,\ \lambda _1=0
\end{equation}
as in Remark \ref{bics3}. We assume (\ref{bics.100}) in the following.

\par Assume, to fix the ideas, that $0\in I$, and restrict the
attention to $I_+=\{ t\in I;\, t\ge 0 \}$. From (\ref{bics.86}), we
get
\begin{equation}\label{bics.103}(\varepsilon D_t+P(t))u^{(N)}=\rho
  ^{(N+1)},\ t\in I_+,\end{equation}
where
\begin{equation}\label{bics.104}
u^{(N)}=e^{-i\int_0^t \lambda ^{(N)} ds/\varepsilon }\nu ^{(N)},\
\rho ^{(N+1)}=e^{-i\int_0^t \lambda ^{(N)} ds/\varepsilon }r ^{(N+1)}.
\end{equation}
By (\ref{bics.91.2}), (\ref{bics.91}), we have
\begin{equation}\label{bics.105}
\| \rho ^{(N+1)}\|_{{\cal H}_\mathrm{sbd}}=
{\cal O}(1) \varepsilon ^{\frac{1}{2}}\left(\frac{\varepsilon
  ^{\frac{1}{2}}h}{\hat{\epsilon }_\vartheta ^2} \right)^{2N+1}
\| u^{(N)}\|_{{\cal H}_\mathrm{sbd}} .
\end{equation}
Taking the imaginary part of the scalar product in ${\cal
  H}_\mathrm{sbd}$ with $u^{(N)}$, we get with norms and scalar
products in ${\cal H}_\mathrm{sbd}$:
$$
-\frac{1}{2}\varepsilon \partial _t\| u^{(N)}\|^2+(\Im Pu^{(N)}|u^{(N)})=\Im (\rho ^{(N+1)}|u^{(N)}),
$$
\[
\begin{split}
\varepsilon \partial _t \| u^{(N)}\|^2&= 2(\Im
Pu^{(N)}|u^{(N)})-2\Im (\rho ^{(N+1)}|u^{(N)})\\
&\le {\cal O}(\epsilon ^\infty )\| u^{(N)}\|^2 +2\| \rho ^{(N+1)}\|
\|u^{(N)}\|.
\end{split}
\]
Hence, by (\ref{bics.105}) and the assumption (\ref{bics.92.5}), 
$$
\varepsilon \partial _t \| u^{(N)}\|^2\le {\cal O}(1) \varepsilon ^{\frac{1}{2}} \left(\frac{\varepsilon
  ^{\frac{1}{2}}h}{\hat{\epsilon }_\vartheta ^2}\right)^{2N+1}\| u^{(N)}\|^2,
$$
leading to
\begin{equation}\label{bics.105.2}
\| u^{(N)}(t)\|\le e^{{\cal O}(1)t\varepsilon ^{-1/2}(\varepsilon
^{1/2}h/\hat{\epsilon }_\vartheta ^2)^{2N+1}}\| u^{(N)}(0)\|,\
0\le t\in I.
\end{equation}
Assume,
\begin{equation}\label{bics.105.4}
(\sup I)\varepsilon ^{-\frac{1}{2}}\left( \frac{\varepsilon
    ^{\frac{1}{2}}h}{\hat{\epsilon }_\vartheta ^2} \right)^{2N+1}
\le {\cal O}(1).
\end{equation}
Then, for $0\le t\in I$,
\begin{equation}\label{bics.105.6} \begin{split}
  \| u^{(N)}(t)\|_{{\cal H}_\mathrm{sbd}}&\le {\cal O}(1) \| u^{(N)}(0)\|_{{\cal H}_\mathrm{sbd}},\\
  \| \rho ^{(N+1)}(t)\|_{{\cal H}_\mathrm{sbd}}&\le {\cal O}(1) \varepsilon ^{\frac{1}{2}}\left( \frac{\varepsilon
      ^{\frac{1}{2}}h}{\hat{\epsilon }_\vartheta ^2} \right)^{2N+1}\|
  u^{(N)}(0)\|_{{\cal H}_\mathrm{sbd}}.
\end{split}
\end{equation}

\par Using the fundamental matrix $E(t,s)$ to correct the error $\rho
^{(N+1)}$ we have the exact solution $u=u^{(N)}_\mathrm{exact}$,
\begin{equation}\label{bics.106}
u=u^{(N)}-\frac{i}{\varepsilon }\int_0^t E(t,s)\rho ^{(N+1)}(s) ds
\end{equation}
of the equation
$$
(\varepsilon D_t+P(t))u=0\hbox{ on }I_+.
$$

\par From (\ref{bics.105.4}) we get
\begin{equation}\label{bics.109}
\sup I\le \varepsilon ^{-N_0},
\end{equation}
for some fixed finite $N_0$. Then by (\ref{bics.97}), (\ref{bics.92.5}), 
\begin{equation}\label{bics.110}
\| E(t,s)\|_{{\cal L}({\cal H}_\mathrm{sbd},{\cal H}_\mathrm{sbd})}\le
e^{{\cal O}(\varepsilon ^\infty )}=1+{\cal O}(\varepsilon ^\infty ),
\end{equation}
and using this and (\ref{bics.105.6}) in (\ref{bics.106}), we get
\begin{equation}\label{bics.111}
\| u-u^{(N)}\|_{{\cal H}_\mathrm{sbd}}\le {\cal O}(1)\varepsilon
^{-1}(\sup I) \varepsilon ^{\frac{1}{2}}\left( \frac{\varepsilon
      ^{\frac{1}{2}}h}{\hat{\epsilon }_\vartheta ^2} \right)^{2N+1}\|
  u^{(N)}(0)\|_{{\cal H}_\mathrm{sbd}}.
\end{equation}
This estimate is the main result of the present work. Let us recollect
the assumptions and the general context in the following theorem.
\begin{theo}\label{bics5}
  Let $V_t=V(t,x)\in C_b^\infty (I\times {\bf R}^n;{\bf R})$, where
  $n=1$, $0<E_-<E_-'<E_+'<E_+<\infty $,
  $E_0(t)\in C^\infty (I;[E_-',E_+'])$,
  $\ddot{\mathrm{O}}(t)\Subset {\bf R}^n$,
  $U_0(t)\subset \ddot{\mathrm{O}}(t)$ be as in the discussion around and
  including {\rm (\ref{bics.1})}--{\rm (\ref{bics.8})}, {\rm (\ref{bics.9.5})}.
  Let $\mu _0(t)$ be a Dirichlet eigenvalue of $-h^2\Delta +V(t,\cdot )$ on $M_0$ as
  in {\rm (\ref{bics.9})} -- {\rm (\ref{bics.12})}. The operator $P(t)$ has a unique resonance
  $\lambda _0(t)$ in the set $\Omega (t)$ in {\rm (\ref{bics.26})}. It is
  simple and satisfies {\rm (\ref{bics.27})}. Here $\epsilon _\vartheta
  =(\epsilon /h)^\vartheta \epsilon  $, for some $\epsilon \in
  [e^{-1/(Ch)},h/C]$ for some sufficiently large constant $C>0$ and
  $\vartheta >0$ is a fixed small constant. Assume {\rm (\ref{bicsny.54})},
  {\rm (\ref{rest.65.5})}:
\begin{equation}\label{bics.111.5}
e^{-1/(C_1h)}\le \epsilon \le \min (h/C_2,\delta ),\ C_1,\, C_2\gg 1.
\end{equation}

\par Define the spaces ${\cal H}_\mathrm{sbd}=H(\Lambda _{\epsilon G_{\mathrm{sbd}}})$,
${\cal D}_\mathrm{sbd}=H(\Lambda _{\epsilon G_{\mathrm{sbd}}}, \langle \xi \rangle^2)$ as earlier in this section, so that $\lambda
  _0(t)$ is the unique eigenvalue in $\Omega (t)$ of $P(t):\, {\cal H}_\mathrm{sbd}\to
  {\cal H}_\mathrm{sbd}$ with domain ${\cal D}_\mathrm{sbd}$.

  \par Then we have the formal asymptotic series
  $\nu (t,\varepsilon )$, $\lambda (t,\varepsilon )$ in Proposition
  {\rm \ref{bics2}}, where we choose $\nu _0(t)$ in {\rm (\ref{bics.73})} with
  $\theta _0(t)=1$, so that $\lambda _1(t)=0$. For $N\ge 1$, define
  the partial sums $\nu ^{(N)}$, $\lambda ^{(N)}$ as in
  {\rm (\ref{bics.84})}, {\rm (\ref{bics.85})}. Let $\varepsilon $ be small enough
  so that {\rm (\ref{bics.88})} holds (and notice that this would follow
  from {\rm (\ref{bics.89})}, {\rm (\ref{bics.90})}) and bounded from below by some positive power of $\epsilon$ as in
  {\rm (\ref{bics.92.5})}. Assume (to fix the ideas)
  that $0\in I$, and assume {\rm (\ref{bics.105.4})} so that $\sup I\le \varepsilon ^{-N_0}$ for some constant
  $N_0>0$ and put $I_+=I\cap [0,+\infty [$. Let
  $u(t)\in C^1(I_+;{\cal H}_\mathrm{sbd})\cap C^0(I_+;{\cal
    D}_\mathrm{sbd})$ be the solution of
\begin{equation}\label{bics.112}
(\varepsilon D_t+P(t))u=0 \hbox{ on }I_+,\ u(0)=u^{(N)}(0),
\end{equation}
where $u^{(N)}$ is defined in {\rm (\ref{bics.104})}. Then {\rm (\ref{bics.111})}
holds uniformly for $t\in I_+$.
\end{theo}

We shall finally describe a situation appearing in the mesoscopic
problems studied in \cite{JoPrSj95}, \cite{PrSj96}, where the potential
well is of diameter $\asymp h$. The potential will be result of drilling
a well of width $\sim h$ in a ``filled potential''.  Let us first
describe the filled potential
$\widetilde{V}_t=\widetilde{V}(t,x)$. Assume,
\begin{equation}\label{bics.113}
\widetilde{V}(t,x)\in C_b^\infty (I\times {\bf R}^n;{\bf R})
\end{equation}
still with $n=1$ for the moment.
\begin{equation}\label{bics.114}
\begin{split}
&\widetilde{V}_t\hbox{ has a holomorphic extension (also denoted } \widetilde{V}_t \hbox{)
  to}\\
&\{x\in {\bf C}^n;\, |\Re x|>C,\ |\Im x|<|\Re x|/C \}\\&\hbox{such that
}
\widetilde{V}_t(x)=o(1),\ x\to \infty \hbox{ for some constant }C>0.
\end{split}
\end{equation}

\par We next define $V_t(x)$ by drilling a thin well of $t$-dependent
depth and of diameter $2h$. Fix a point $x_0\in {\bf R}^n$ and assume,
\begin{equation}\label{bics.115}
\widetilde{V}_t(x_0)\ge 1/C,\ t\in I.
\end{equation}

\par Let $I\ni t\mapsto \alpha ^t\in [1/C,C]$ (where $C>0$ is a new constant) be a smooth function
with
\begin{equation}\label{bics.116}
\partial _t^k\alpha ^t={\cal O}_k(1),\ k\in {\bf N}
\end{equation}
and put
\begin{equation}\label{bics.117}
V_t(x)=\widetilde{V}_t(x)-\alpha ^t1_{U_0}(x),\ U_0=B(x_0,h).
\end{equation}
We need a first reference operator. Choose $\widetilde{\widehat{ V}}_t(x)\in
C_b^\infty ({\bf R}^n;{\bf R})$ such that
\begin{equation}\label{bics.118}
\widetilde{\widehat{ V}}_t(x)=\widetilde{V}_t(x)\hbox{ in a
  neighborhood of }x_0,\hbox{ independent of }t,h,
\end{equation}
\begin{equation}\label{bics.119}
\widetilde{\widehat{ V}}_t\ge \widetilde{V}_t(x_0)-\delta _0,
\end{equation}
for some small fixed constant $\delta _0>0$. Put
$\widehat{V}_t=\widetilde{\widehat{ V}}_t-\alpha ^t1_{U_0}$. Then
$\widehat{P}_t:=-h^2\Delta +\widehat{V}_t$ is
a self-adjoint operator (defined by means of Friedrichs extension)
with purely discrete spectrum in $]-\infty ,V(x_0)-\delta _0[$,
bounded from below by $\min (\widetilde{V}_t(x_0)-\delta
_0,\widetilde{V}_t(x_0)-\alpha ^t-{\cal O}(h))$.

\par The eigenvalues in the interval $]-\infty ,V(x_0)-2\delta _0[$
can be obtained  by scaling and simple semi-classical analysis: Let
$$
e_0(\alpha )<e_1(\alpha )<...<e_{k(\alpha )}(\alpha )<0
$$
be the negative eigenvalues of $-\partial ^2-\alpha 1_{]-1,1[}(x)$ on
${\bf R}$. Then for $h<0$ small enough, the eigenvalues of
$\widehat{P}_t$ in $]-\infty ,V(x_0)-2\delta _0[$ are of the form
\begin{equation}\label{bics.120}
\widehat{E}_k(t;h)\sim E_{k,0}(t)+hE_{k,1}(t)+...
\end{equation}
where
\begin{equation}\label{bics.121}
E_{k,0}(t)=\widetilde{V}_t(x_0)+e_k(\alpha ^t)
\end{equation}
belongs to $]-\infty ,\widetilde{V}_t(x_0)-3\delta _0/2[$ (in the limit of small $h$)
and we get all such eigenvalues this way (one for each $k$) when $h>0$
is small enough.

\par Now fix a $k\in {\bf N}$ and assume that we have the well-defined
eigenvalue $\widehat{E}_k(t;h)=:\widehat{E}(t;h)$ of $\widehat{P}_t$
in $]\delta _0,\widetilde{V}_t(x_0)-2\delta _0[$ for all $t\in I$ for
$0<h\ll 1$. (The positivity is required since we look for shape
resonances of $P_t$.) Define the $t$-dependent potential island
\begin{equation}\label{bics.122}
{\ddot{\mathrm{O}}}(t)=\{ x\in {\bf R}^n;\, \widetilde{V}_t(x)>E_0(t)
\} ,\ \ E_0(t):=E_{k,0}(t).
\end{equation}

\par Let $\widetilde{p}_t=\xi ^2+\widetilde{V}_t(x)$ be the
semi-classical principal symbol of $\widetilde{P}_t$. Assume that
\begin{equation}\label{bics.123}
\hbox{The }H_{\widetilde{p}_t}\hbox{-flow is non-trapping on
}{{(\widetilde{p}_t)^{-1}(E_0(t))}_\vert}_{{\bf R}^n\setminus {\ddot{\mathrm{O}}}}.
\end{equation}

\par In $\overline{{\ddot{\mathrm{O}}}(t)}$ we have the Lithner-Agmon distance $d_t$, associated
to the metric $(\widetilde{V}_t(x)-E_0)dx^2$. Let
$S_t:=d_t(x_0,\partial {\ddot{\mathrm{O}}}(t))>0$,
$$M_0(t):=\{ x\in {\ddot{{\mathrm{O}}}}(t);\,
\widetilde{V}_t(x)>E_0(t)+\delta  \},$$
where $\delta >0$ is a small constant. Notice that $M_0(t)$ has
smooth boundary and depends smoothly on $t$. Let $P_0^t$ be the
Dirichlet realization of $P_t$ in $M_0(t)$. Then $P_0^t$ has a unique
eigenvalue $E(t;h)$ such that $E(t;h)=\widehat{E}(t;h)=o(1)$, $h\to 0$
and the two eigenvalues are exponentially close:
\begin{equation}\label{bics.124}
E(t;h)=\widehat{E}(t;h)+{\cal O}(e^{-1/(Ch)}).
\end{equation}
As in the beginning of this section \ref{rest} we know that
$P_t$ has a unique resonance with
$$
\Re \lambda _0(t)-E(t;h)=o(1),\ \Im \lambda _0(t)\ge -Ch\ln (1/h).
$$
Moreover, we have
\begin{equation}\label{bics.125}
\lambda _0(t;h)=E(t;h)+\widetilde{{\cal O}}(e^{-2S_t/h}).
\end{equation}

This means that (apart from the fact that our potential is
$h$-dependent near $U_0$) we can apply Theorem \ref{bics5} with
$\delta (h)\asymp 1$.

\end{document}